\documentclass{no-jocg}

\usepackage{graphicx}
\usepackage{amsmath}
\usepackage{amssymb}
\usepackage{amsthm}
\usepackage{thmtools}
\usepackage{xspace}
\usepackage[linesnumbered,algoruled,longend,vlined,german]{algorithm2e}
\usepackage{cleveref}
\usepackage{subcaption}

\newtheorem{definition}{Definition}
\newtheorem{lemma}{Lemma}
\newtheorem{theorem}{Theorem}

\newtheorem{claim}{Claim}
\newtheorem{observation}{Observation}
\newtheorem{remark}{Remark}

\crefname{remark}{Remark}{Remarks}
\crefname{observation}{Observation}{Observations}
\crefname{claim}{Claim}{Claims}
\crefname{figure}{Figure}{Figures}
\crefname{equation}{Equation}{Equations}
\crefname{lemma}{Lemma}{Lemmas}
\crefname{theorem}{Theorem}{Theorems}
\crefname{corollary}{Corollary}{Corollaries}
\crefname{definition}{Definition}{Definitions}
\crefname{section}{Section}{Sections}

\usepackage{xcolor}
\definecolor{defblue}{rgb}{0.121,0.47,0.705}
\let\emph\relax
\DeclareTextFontCommand{\emph}{\color{defblue}\em}

\usepackage{tikz}
\usetikzlibrary{calc, shapes, backgrounds, patterns, arrows, decorations, decorations.pathreplacing, decorations.pathmorphing, fit, arrows.meta,patterns,patterns.meta,decorations.pathmorphing}
\graphicspath{{figures/}}

\usepackage{hyperref}

\usepackage{booktabs}

\newcommand{\OLPlong}{\textsc{Ordered Level Planarity}\xspace}
\newcommand{\OLP}{\textsc{OLP}\xspace}
\newcommand{\PLPlong}{\textsc{Partial Level Planarity}\xspace}
\newcommand{\PLP}{\textsc{PLP}\xspace}
\newcommand{\CLPlong}{\textsc{Constrained Level Planarity}\xspace}
\newcommand{\CLP}{\textsc{CLP}\xspace}

\newcommand{\NPh}{\textsc{NP}-hard\xspace}

\newcommand{\Wh}[1][1]{$W[#1]$-hard\xspace}
\newcommand{\Whness}[1][1]{$W[#1]$-hardness\xspace}
\newcommand{\XP}{\textsc{XP}\xspace}
\newcommand{\XNLP}{\textsc{XNLP}\xspace}
\newcommand{\XNLPh}{\textsc{XNLP}-hard\xspace}
\newcommand{\XNLPhness}{\textsc{XNLP}-hardness\xspace}
\newcommand{\FPT}{\textsc{FPT}\xspace}
\newcommand{\Oh}{\ensuremath{\mathcal{O}}}
\newcommand{\G}{\ensuremath{\mathcal{G}}\xspace}
\newcommand{\J}{\ensuremath{\mathcal{J}}\xspace}
\newcommand{\C}{\ensuremath{\mathcal{C}}\xspace}
\DeclareMathOperator{\Span}{span}
\DeclareMathOperator{\idx}{idx}

\newenvironment{claimproof}[1][\proofname]
{ \proof[#1]  }
{ \endproof }

\newcommand{\myproofparagraph}[1]{\bigskip\noindent\textit{#1}\;}

\title{\MakeUppercase{Constrained and Ordered Level Planarity Parameterized by the Number of Levels}%
	\thanks{A preliminary version of this paper appeared in the proceedings of the 40th International Symposium on Computational Geometry (SoCG 2024)~\cite{DBLP:conf/compgeom/BlazejKKS0024}.}
}

\author{%
	Václav Blažej,%
	\thanks{\affil{University of Warwick, Coventry, United Kingdom}
	}\,
	Boris Klemz,%
	\thanks{\affil{Universität Würzburg, Würzburg, Germany}
	}\,
	Felix Klesen,\footnotemark[3]\,
	{Marie Diana} Sieper,\footnotemark[3]\,
	Alexander Wolff,\footnotemark[3]\,
	and Johannes Zink\footnotemark[3]
}

\begin{document}
	
	\thispagestyle{empty}
	
	\maketitle
	
	\begin{abstract}
		The problem \textsc{Level Planarity} asks for a crossing-free
		drawing of a graph in the plane such that vertices are placed at
		prescribed y-coordinates (called \emph{levels}) and such that every
		edge is realized as a y-monotone curve.  In the variant \CLPlong
		(\CLP), each level $y$ is equipped with a partial order $\prec_y$ on
		its vertices and in the desired drawing the left-to-right order of
		vertices on level~$y$ has to be a linear extension of~$\prec_y$.
		\OLPlong (\OLP) corresponds to the special case of \CLP where the
		given partial orders $\prec_y$ are total orders.  Previous results
		by Br\"uckner and Rutter [SODA 2017] and Klemz and Rote [ACM Trans.\
		Alg.\ 2019]
		state that both \CLP and \OLP are NP-hard even in severely
		restricted cases.  In particular, they remain NP-hard even when
		restricted to instances whose \emph{width} (the maximum number of
		vertices that may share a common level) is at most two.  In this
		paper, we focus on the other dimension: we study the parameterized
		complexity of \CLP and \OLP with respect to the \emph{height} (the
		number of levels).
		
		We show that \OLP parameterized by the height is complete with
		respect to the complexity class XNLP, which was first studied by
		Elberfeld, Stockhusen, and Tantau [Algorithmica 2015] (under a
		different name) and recently made more prominent by Bodlaender,
		Groenland, Nederlof, and Swennenhuis [FOCS 2021].  It contains all
		parameterized problems that can be solved nondeterministically in
		time $f(k)\cdot n^{\mathcal O(1)}$ and space $f(k)\cdot \log n$
		(where~$f$ is a computable function, $n$ is the input size, and $k$
		is the parameter).  If a problem is XNLP-complete, it
		lies in XP, but is W[$t$]-hard for every $t$.
		
		In contrast to the fact that \OLP parameterized by the height lies
		in XP, it turns out that \CLP is NP-hard even when restricted to
		instances of height~4.  We complement this result by showing that
		\CLP can be solved in polynomial time for instances of height at
		most~3.
	\end{abstract}

	\section{Introduction}
	
	In an \emph{upward} drawing of a directed graph, every edge~$e=(u,v)$ is realized as a y-monotone curve 
	that goes upwards from~$u$ to~$v$, i.e., the y-coordinate strictly increases when traversing~$e$ from~$u$ towards~$v$.
	Also known as poset diagrams, these drawings provide a natural way to visualize a partial order on a set of items.
	The classical problem \textsc{Upward Planarity} asks whether a given directed graph admits a drawing that is both upward and planar (i.e., crossing-free).
	It is known to be NP-hard~\cite{DBLP:journals/siamcomp/GargT01}, but becomes solvable in polynomial time if the y-coordinate of each vertex is prescribed~\cite{DBLP:journals/tsmc/BattistaN88,DBLP:conf/gd/HeathP95,DBLP:conf/gd/JungerLM98}.
	In contrast, when both the y-coordinate and the x-coordinate of each vertex is prescribed, the problem is yet again NP-hard~\cite{DBLP:journals/talg/KlemzR19}.
	The paper at hand is concerned with the parameterized complexity of (a generalization of) the latter variant of \textsc{Upward Planarity}, the parameter being the number of levels.
	Next, we define these problems more precisely, adopting the notation and terminology used in~\cite{DBLP:journals/talg/KlemzR19}.
	
	\paragraph{Level planarity.}
	A \emph{level graph}~$\mathcal G=(G,\gamma)$ is a directed graph~$G=(V,E)$ together with a \emph{level assignment}, 
	which is a surjective map $\gamma\colon V\rightarrow \lbrace 1,2,\dots ,h\rbrace$ where~$\gamma (u)<\gamma (v)$ for every edge~$(u,v)\in E$.
	The vertex set $V_i=\lbrace v\mid \gamma (v)=i\rbrace$ is called the~$i$-th \emph{level} of~$\mathcal G$.
	The \emph{width} of level $V_i$ is~$|V_i|$.
	The \emph{levelwidth} 
	of~$\mathcal G$ is the maximum width of any level in~$\mathcal G$ and the \emph{height} of~$\mathcal G$ is the number~$h$ of levels.
	A \emph{level planar drawing} of~$\mathcal G$ is an upward planar drawing of~$G$ where the y-coordinate of each vertex~$v$ is~$\gamma (v)$. 
	We use~$L_i$ to denote the horizontal line with y-coordinate~$i$.
	Algorithms for computing level planar drawings usually just determine a \emph{level planar embedding} of a level planar drawing, which for each $i~\in \lbrace 1,2,\dots ,h\rbrace$ lists the left-to-right sequence of vertices and edges intersected by $L_i$.
	Note that this corresponds to an equivalence class of drawings from which an actual drawing is easily derived.
	The level graph~$\mathcal G$ is called \emph{proper} if $\gamma (v)=\gamma (u)+1$ for every edge~$(u,v)\in E$.
	
	The problem \textsc{Level Planarity} asks whether a given level graph admits a level planar drawing.
	It can be solved in linear time~\cite{DBLP:journals/tsmc/BattistaN88,DBLP:conf/gd/HeathP95,DBLP:conf/gd/JungerLM97,DBLP:conf/gd/JungerLM98}; see~\cite{fulek2013hanani} for a more detailed discussion on this series of papers.
	It is easy to see that \textsc{Level Planarity} is polynomial time/space equivalent to the variant where~$\gamma$ maps to $h$ arbitrary distinct real numbers.
	
	\paragraph{Constrained and ordered level planarity.}
	In 2017, Br\"uckner and Rutter~\cite{DBLP:conf/soda/BrucknerR17} and
	Klemz and Rote~\cite{DBLP:journals/talg/KlemzR19} independently
	introduced and studied two closely related variants of \textsc{Level
		Planarity}, defined as follows.
	A \emph{constrained (ordered) level graph}
	$\G=(G,\gamma,(\prec_i)_{1\le i\le h})$ is a triplet
	corresponding to a level graph $(G, \gamma)$ of height~$h$  equipped with a family containing, for each $1\le i\le h$, a partial (total) order on the 
	vertices in~$V_i$.
	A \emph{constrained (ordered)} level planar drawing of \G is a level
	planar drawing of $(G, \gamma)$ where, for each $1\le i\le h$, the
	left-to-right order of the vertices in~$V_i$ corresponds to a linear
	extension of~$\prec_i$ (is $\prec_i$).
	For a pair of vertices $u,v\in V_i$ with $u\prec_i v$, we refer to $u\prec_i v$ as a \emph{constraint} on $u$ and $v$.
	
	The problem \CLPlong (\CLP) / \OLPlong (\OLP)
	asks whether a given constrained~/ ordered level graph admits a constrained~/ ordered level planar drawing, in which case the input is called a \emph{constrained}~/ \emph{ordered level planar graph}.
	The special case where the height of all instances is restricted to a given value $h$ is called $h$-level \CLP~/ $h$-level \OLP.
	In \CLP, each partial order $\prec_i$ is assumed to be given in form of a directed acyclic graph including all of its transitive edges.
	In \OLP, each total order $\prec_i$ is encoded by equipping each vertex of level $V_i$ with an integer that is equal to its rank in the order~$\prec_i$.
	Note that \OLP is polynomial time/space equivalent to the variant of \textsc{Level Planarity} where each vertex is equipped with a prescribed x-coordinate, implying that the only challenge is to draw the edges.

	Klemz and Rote~\cite{DBLP:journals/talg/KlemzR19} showed that \OLP (and, thus, \CLP) is NP-hard even when restricted to the case where the underlying undirected graph of~$G$ is a disjoint union of paths.
	Note that such a graph has bounded pathwidth (and treewidth), maximum
	degree, and feedback vertex set number, ruling out efficient
	parameterized algorithms with respect to these classical parameters.
	Additionally, the instances produced by their reduction have a
	levelwidth of only two.
	Independently, Br\"uckner and Rutter~\cite{DBLP:conf/soda/BrucknerR17}
	provided a very different reduction for showing that \CLP is NP-hard (in fact, their reduction shows the NP-hardness of \PLPlong, which can be seen as a generalization of \OLP and a special case of \CLP; see below).
	The instances constructed by their reduction are connected and have bounded maximum degree.
	They also present a polynomial time algorithm for \CLP for the case where the graph has a single sink, which they later~\cite{DBLP:conf/isaac/BrucknerR20} sped up for the special case where the single-sink graph is also biconnected.
	Very recently, Klemz and Sieper~\cite{clp-vc} showed that \CLP (and, thus,
	\OLP) is FPT if parameterized by the vertex cover number of the input graph.
	They also observed that the reduction by Klemz and Rote (Br\"uckner and
	Rutter) can easily be modified to show that OLP (CLP) is NP-hard even if
	restricting to (proper) instances of bounded treedepth.
	
	\paragraph{Other related work.}
	The problem \PLPlong (\PLP) (introduced and studied by Br\"uckner and Rutter~\cite{DBLP:conf/soda/BrucknerR17}), asks whether a given level planar drawing of a subgraph~$H$ of the input graph~$G$ can be extended to a level planar drawing of~$G$.
	This can be seen as a generalization of \OLP and, in the proper case, as a specialization of \CLP.
	Several other problems related to the construction of level planar drawings have been studied, including problems with other kinds of ordering constraints (e.g.,
	\textsc{Clustered Level Planarity}~\cite{DBLP:conf/sofsem/ForsterB04,DBLP:journals/tcs/AngeliniLBFR15,DBLP:journals/talg/KlemzR19}
	and
	\textsc{T-Level Planarity}~\cite{DBLP:journals/dam/WotzlawSP12,DBLP:journals/tcs/AngeliniLBFR15,DBLP:journals/talg/KlemzR19}),
	problems with a more geometric touch
	(see, e.g., \cite{DBLP:journals/jda/HongN10,DBLP:conf/esa/Klemz21}),
	and variants of \textsc{Level Planarity} seeking drawings on surfaces different from the plane
	(see, e.g., \cite{DBLP:journals/jgaa/BachmaierBF05,DBLP:journals/tcs/AngeliniLBFPR20,DBLP:conf/esa/BachmaierB08}).
	
	\paragraph{Contribution and organization.}
	As discussed above, the parameterized complexity of \OLP and \CLP with
	respect to classical graph parameters (vertex cover number, feedback vertex set number, treedepth, pathwidth, treewidth, maximum degree) has been explored exhaustively~\cite{DBLP:journals/talg/KlemzR19,DBLP:conf/soda/BrucknerR17,clp-vc}.
	In terms of more problem specific parameters, Klemz and Rote~\cite{DBLP:journals/talg/KlemzR19} showed
	that \OLP (and, thus, \CLP) is NP-hard even if restricted to
	instances of levelwidth two.
	In this paper, we focus on the other ``dimension'':
	we study
	\CLP and \OLP parameterized by height.
	
	We show that \OLP parameterized by the height is complete with
	respect to the complexity class \XNLP, which was first studied by
	Elberfeld, Stockhusen, and Tantau~\cite{DBLP:journals/algorithmica/ElberfeldST15} (under the name $\mathrm N[f\ \mathrm{poly},f\ \mathrm{log}]$) and recently made more prominent by Bodlaender,
	Groenland, Nederlof, and Swennenhuis~\cite{XNLP2021}.  It contains all
	parameterized problems that can be solved nondeterministically in
	time $f(k)\cdot n^{\mathcal O(1)}$ and space $f(k)\cdot \log n$
	(where $f$ is a computable function, $n$ the input size, and $k$
	the parameter).
	Elberfeld et al.\ and Bodlaender et al.\ study properties of
	(problems in) this class and provide several
	problems that are \XNLP-complete. 
	In particular, if a problem is \XNLP-complete, it
	lies in \XP, but is \Wh[t] for every~$t$~\cite{XNLP2021}.
	
	\begin{theorem}
		\textsc{Ordered Level Planarity} parameterized by
		the height of the input graph is \XNLP-complete
		(and, thus, it lies in \XP, but is \Wh[t] for every~$t$).
		\XNLP-hardness holds even when restricted to the case where the input graph is connected.
		Moreover, there is a constructive XP-time algorithm for \OLPlong (w.r.t.\ the height).
		\label{thm:olp-XNLP-complete}
	\end{theorem}

	Parameterizing \OLP by height captures the ``linear'' nature of the
	solution -- this is reminiscent of recent results by
	Bodlaender, Groenland, Jacob, Jaffke, and Lima~\cite{Bodlaender2022}
	who established \XNLP-completeness for several problems parameterized
	by linear width measures
	(e.g., \textsc{Capacitated Dominating Set} by pathwidth and
	\textsc{Max Cut} by linear cliquewidth).
	However, to the best of our knowledge, this is the first graph drawing (or computational geometry) problem shown to be \XNLP-complete.
	The algorithms are described in \cref{sec:olp-alg} (\cref{thm:olp-XNLP}), whereas the hardness is shown in \cref{sec:olp-hardness} (\cref{thm:olp_connected_xnlp_hard}).
	
	In contrast to the fact that \OLP parameterized by the height lies
	in \XP, 
	it is not difficult to see that the socket/plug gadget described by
	Br\"uckner and Rutter~\cite{DBLP:conf/soda/BrucknerR17}
	can be utilized in the context of a reduction from \textsc{3-Partition} to show that 
	(\PLP and, thus) \CLP remains NP-hard even when restricted to instances of constant height.
	In fact, the unpublished full version of~\cite{DBLP:conf/soda/BrucknerR17} features such a construction with a height of~7~\cite{ignaz-pc}.
	Here, we present a reduction that is tailor-made for \CLP, showing that it is NP-hard even if restricted to instances of height~4.
	We complement this result by showing that the (surprisingly challenging) case of instances with height at most~3 can be solved in polynomial time.
	
	\begin{theorem}
		\label{thm:clp-summary}
		\CLPlong is \NPh even if restricted to
		height~4, but instances of height at most~3 can be solved
		constructively in polynomial time.
	\end{theorem}
	
	We show the hardness in \cref{sec:clp-hardness}
	(\cref{thm:CLP-NPhard}) and present the algorithm in
	\cref{sec:clp-3lvl} (\cref{thm:clp-3lvl}).  As a warm-up for the
	rather technical positive result, we show that \CLP can be solved
	in linear time when restricted to instances of height at most~2;
	see \cref{sec:clp-2lvl} (\cref{thm:clp-2lvl}).
	
	We conclude by stating some open problems in \cref{sec:conc}.
	
	\paragraph{Notation and conventions.}
	Given an integer $k>0$, we use $[k]$ as shorthand for
	$\{1, 2, \dots, k\}$.
	Given a directed or undirected graph~$G$, let $V(G)$ denote the vertex
	set of~$G$, and let $E(G)$ denote the edge set of~$G$.
	Recall that level graphs are directed with each edge $(u,v)$ pointing upwards, i.e., $\gamma(u)<\gamma(v)$.
	As a shorthand and when the direction is not important, we use both $uv$ and~$vu$ to refer to a directed edge $(u,v)$.

	\section{An XP / XNLP Algorithm for Ordered Level Planarity}
	\label{sec:olp-alg}
	
	In this section, we show that \textsc{Ordered Level Planarity} is in
	\textsc{XNLP} (and thus in \textsc{XP})
	if parameterized by the height~$h$ of the input graph.
	Moreover, we show how to construct an ordered level planar drawing (if it exists) in \XP-time.
	The main idea of our approach is to continuously sweep the plane with an unbounded y-monotone curve~$s$ from left to right in a monotone fashion
	such that for each edge $(u,v)$, there is a point in time where~$u$ and~$v$ are consecutive vertices along~$s$.
	When this happens, the edge can be drawn without introducing any crossings due to the fact that~$s$ moves monotonically.
	To discretize this idea and turn it into an algorithm,
	we instead determine a sequence of unbounded y-monotone curves $S=(s_1,s_2,\dots,s_z)$
	that is sorted from left to right (i.e., no point of $s_{i+1}$ is to the left of~$s_i$),
	has a length of $z\in \mathcal O(n)$,
	and contains for every edge $(u,v)$ a curve $s_i$ along which~$u$ and~$v$ are consecutive vertices.
	Now, given~$S$, the desired drawing can be constructed in polynomial time;
	for an illustration see \cref{fig:olp-DP-example}.
	Moreover, the sequence~$S$ can be obtained in \XP-time/space by exhaustively enumerating all possibilities
	or in \XNLP-time/space by nondeterministic guessing.
	Let us proceed to formalize these ideas.
	
	\paragraph{Gaps and positions.}
	Let $\mathcal G=(G=(V,E),\gamma,(\prec_i)_{1\le i\le h})$ be an ordered level graph and consider one of its levels~$V_i$.
	Let $(v_1,v_2,\dots,v_{\lambda_i})$ be the linear order of~$V_i$ corresponding to~$\prec_i$.
	In an ordered level planar drawing of~$\mathcal G$, the vertices~$V_i$ divide the line~$L_i$ into a sequence of open line-segments and rays, which we call the \emph{gaps} of~$L_i$.
	A \emph{position} on~$L_i$ is a gap of~$L_i$ or a vertex of~$V_i$.
	Each position on~$L_i$ is encoded by an index in $P_i=\{0\}\cup [2\lambda_i]$:
	the index $0$ \emph{represents} the gap that precedes~$v_1$;
	an odd index~$p$ \emph{represents} the vertex $v_{\lceil p/2\rceil}$; and
	an even index~$p\neq 0$ \emph{represents} the gap that succeeds $v_{\lceil p/2\rceil}$.
	
	\paragraph{Separations.}
	A \emph{separation} for~$\mathcal G$ is an element of $P_1\times P_2\times \dots \times P_h$.
	Intuitively, a separation $s=(p_1,p_2,\dots,p_h)$ represents the equivalence class of unbounded y-monotone curves that intersect line~$L_i$ in position~$p_i$ for each $1\le i\le h$;
	see \cref{fig:olp-DP-example}.
	We say that a vertex~$v\in V$ is \emph{on} $s$ if $p_{\gamma(v)}$ represents~$v$.
	Moreover, we say that~$v$ is \emph{to the left (right)} of~$s$
	if the index corresponding to the position of~$v$ is
	strictly smaller (larger) than $p_{\gamma(v)}$.
	Consider two vertices $u,v\in V$ that are on~$s$ and where $\gamma(u)<\gamma(v)$.
	We say that~$u$ and~$v$ are \emph{consecutive} on~$s$ if all of the indices $p_{\gamma(u)+1},p_{\gamma(u)+2},\dots,p_{\gamma(v)-1}$ represent gaps.
	In this case, $v$ is the \emph{successor} of~$u$ on~$s$ and~$u$ is the \emph{predecessor} of~$v$ on~$s$.
	We say that~$s$ \emph{uses} an edge $e=(u,v)\in E$
	if~$u$ and~$v$ are consecutive along~$s$.
	
	\paragraph{Sweeping sequences.}
	A \emph{sweeping sequence} for~$\mathcal G$ is a sequence $S=(s_1,s_2,\dots,s_z)$ of separations for~$\mathcal G$ such that for each $j\in [z-1]$, we have $s_j\le s_{j+1}$, componentwise.
	We say that a sweeping sequences \emph{uses} an edge~$e\in E$, if it contains a separation that uses~$e$.
	
	\paragraph{Drawing algorithm.}
	The following lemma formalizes the idea of our drawing algorithm.
	
	\begin{lemma}\label{lem:sweep-algo}
		Let $\mathcal G=(G=(V,E),\gamma,(\prec_i)_{1\le i\le h})$ be an ordered level graph
		and let $S=(s_1,s_2,\dots,s_z)$ be a sweeping sequence for $\mathcal G$ that uses every edge of~$E$.
		Given~$\mathcal G$ and~$S$, an ordered level planar drawing of~$\mathcal G$ can be constructed in $\mathcal O(zh+n^2)$ time,
		where~$n$ is the number of vertices.
	\end{lemma}
	
	\begin{proof}
		Without loss of generality, we assume that $s_1=(0,0,\dots,0)$
		(otherwise, we simply prepend this separation to~$S$).
		Recall that a separation $s_j=(p_1,p_2,\dots,p_h)$ corresponds to an equivalence class of unbounded y-monotone curves that intersect line~$L_i$ in position~$p_i$ for each $1\le i\le h$.
		For each $j\in [z]$, let $c_j$ denote a member of the class of curves corresponding to~$s_j$.
		The fact that $S$ is a sweeping sequence implies that we may assume that for every $j\in [z-1]$, no point of $c_{j+1}$ lies to the left of $c_j$ and that the only points shared by $c_j$ and~$c_{j+1}$ are those that correspond to vertices.
		In particular, the curve~$c_1$ passes through the left-most gap on each level.
		
		To prove the existence of the desired drawing, it suffices to show that for each $j\in [z]$
		there is an ordered level planar drawing~$\Gamma_j$ that contains
		all vertices that are not to the right of~$s_j$
		and all edges that are used by one of the separations $s_1,s_2,\dots,s_j$
		and where the vertices on $s_j$ and the edges used by $s_j$
		are drawn on the curve~$c_j$
		and the remaining vertices and edges are drawn to the left of~$c_j$.
		Note that, indeed, $\Gamma_z$ is the desired drawing of~$\mathcal G$.
		
		The (empty) drawing~$\Gamma_1$ is trivially constructed.
		To construct~$\Gamma_{j+1}$ from~$\Gamma_j$,
		we simply place all (isolated) vertices that are to the right of~$s_j$ and to the left of~$s_{j+1}$ between~$c_j$ and~$c_{j+1}$
		and then draw on~$c_{j+1}$ all the new vertices on~$s_{j+1}$, as well as the not yet drawn edges that are used by~$s_{j+1}$.
		By definition of the curves $c_1,c_2,\dots,c_{j+1}$, this process cannot introduce any crossings.
		For illustrations, refer to \cref{fig:olp-DP-example}.
		
		The above inductive proof directly corresponds to an iterative algorithm (note that we do not have to explicitly construct the curves $c_1,c_2,\dots,c_z$).
		The total time for adding the vertices is $\mathcal O(zh+n)$.
		To add the edges,
		whenever there is a new pair of consecutive vertices on the current separation,
		we need to check whether these vertices are adjacent,
		which takes~$\mathcal O(n)$ time.
		New pairs of consecutive vertices are only created when a vertex appears or disappears from the current separation.
		Moreover, the appearance/disappearance of a vertex can only create~$\mathcal O(1)$ new pairs of consecutive vertices.
		Combined with the fact that for each vertex there is an interval of separations in~$S$ on which it is located,
		it follows that the total number of adjacency checks we have to perform is~$\mathcal O(n)$.
		Thus, the total time for the adjacency checks is~$\mathcal O(n^2)$.
		Drawing an edge takes time linear in the number of levels it spans,
		which is at most $h\in \mathcal O(n)$.
		Moreover,
		by planarity (note that the inductive proof shows that $G$ is planar), the number of edges is~$\mathcal O(n)$.
		Hence, the time to draw all the edges is~$\mathcal O(n^2)$, 
		and the total runtime is $\mathcal O(zh+n^2)$, as claimed.
	\end{proof}
	
	In the remainder of this section, we will show that for each ordered level graph that admits an ordered level planar drawing there exists a particularly well-structured sweeping sequence that uses all of its edges.
	Moreover, we will show that the existence of such a sequence can be tested efficiently.
	Combined with \cref{lem:sweep-algo}, this results in the desired XP / XNLP algorithms.
	
	\paragraph{Nice and exhaustive sweeping sequences.}
	Let $S=(s_1,s_2,\dots,s_z)$ be a sweeping sequence for our ordered level graph~$\mathcal G$.
	We say that~$S$ is \emph{nice} if for each $i\in [z-1]$, the separation~$s_{i+1}$ is obtained from~$s_i$ by incrementing exactly one component by~$1$.
	Moreover, we say that~$S$ is \emph{exhaustive} if it is nice and
	$s_1=(0,0,\dots,0)$ and
	$s_z=(|P_1|-1,|P_2|-1,\dots,|P_h|-1)$; see \cref{fig:olp-DP-example}.
	
	\begin{lemma}\label{lem:existence-exhaustive-sequence}
		Let $\mathcal G=(G=(V,E),\gamma,(\prec_i)_{1\le i\le h})$ be an ordered level graph.
		Then~$\mathcal G$ admits an ordered level planar drawing
		if and only if
		there is an exhaustive sweeping sequence for~$\mathcal G$ that uses every edge in~$E$.
	\end{lemma}
	
	\begin{proof}
		The ``if''-direction follows immediately from \cref{lem:sweep-algo}.
		For the ``only if''-direction, suppose that~$\mathcal G$
		admits an ordered level planar drawing~$\Gamma$.
		We augment~$\mathcal G$ and~$\Gamma$ as follows:
		first, we insert a new bottommost and topmost level
		with a single vertex~$a$ and~$b$, respectively.
		Second, we insert a maximal set of y-monotone crossing-free edges
		(i.e., the resulting drawing is supposed to be crossing-free).
		We explicitly allow the introduction of parallel edges,
		but we require that for each pair of parallel edges~$e,e'$
		there is a vertex in the interior of the simple closed curve formed by~$e,e'$.
		Let~$\mathcal G'$ and~$\Gamma'$ denote the resulting graph and drawing,
		respectively, which is an internal triangulation:
		
		\begin{claim}\label{claim:triangulation}
			The outer face of~$\Gamma'$ is bounded by two parallel $(a,b)$-edges
			and
			every internal face of~$\Gamma'$ is a \emph{triangle} (i.e., its boundary corresponds to a cycle of length three).
		\end{claim}
		
		\begin{claimproof}
			The statement regarding the outer face obviously holds.
			Now consider an internal face~$f$.
			Towards a contradiction, assume that~$f$ is not a triangle.
			Let~$v$ be a bottommost vertex of~$f$,
			which is incident to (at least) two edges of~$f$.
			Let~$e_1=(v,v_1),e_2=(v,v_2)$ be two edges
			that are consecutive in the cyclic order of edges incident to~$v$
			and consecutive along the boundary of~$f$.
			Without loss of generality, we may assume that $\gamma(v_2)\le \gamma(v_1)$.
			Moreover, without loss of generality, we may assume that~$e_2$ is to the right of~$e_1$.
			We distinguish two main cases.
			
			\myproofparagraph{Case 1.} There is a vertex~$u$ with
			$\gamma(v)<\gamma(u)\le \gamma(v_2)$
			that is located to the right of~$e_1$ and to the left of~$e_2$.
			We can add an edge from~$v$ to the bottommost vertex
			with these properties in a y-monotone crossing-free fashion in~$f$.
			Moreover, this edge cannot be parallel to some edge of~$\Gamma'$;
			a contradiction to the definition of~$\Gamma'$.
			
			\myproofparagraph{Case 2.} There is no vertex~$u$ with
			$\gamma(v)<\gamma(u)\le \gamma(v_2)$
			that is located to the right of~$e_1$ and to the left of~$e_2$.
			We distinguish two subcases.
			
			\myproofparagraph{Case 2.1.}
			$\gamma(v_2)< \gamma(v_1)$.
			In this case, we can add an edge from~$v_2$ to~$v_1$
			in a y-monotone crossing-free fashion in~$f$ such that,
			even if this edge is parallel to some other edge,
			the interior of the region bounded by these edges
			contains at least one vertex (recall that~$f$ is not a triangle by assumption).
			So again, we obtain a contradiction to the definition of~$\Gamma'$.
			
			\myproofparagraph{Case 2.2.}
			$\gamma(v_2)=\gamma(v_1)$.
			In this case, we draw a y-monotone curve~$c$ in~$f$
			that starts at~$v$ and extends upwards until it crosses some edge~$e_3=(u,w)$ with~$e_3\notin \{e_1,e_2\}$.
			The existence of~$e_3$ follows from the existence of the two edges of the outer face.
			By applying the assumption of Case~2 to $u$, it follows that the upper endpoint~$w$ of~$e_3$ is neither~$v_1$ nor~$v_2$.
			Hence, we can add an edge from~$v$ to~$w$ in a y-monotone crossing-free fashion in~$f$ such that,
			even if this edge is parallel to some other edge,
			the interior of the region bounded by these edges contains at least one vertex ($v_1$ or $v_2$).
			So again, we obtain a contradiction to the definition of~$\Gamma'$.
		\end{claimproof}
		
		Let~$e_\ell$ and~$e_r$ denote the left and right edge of the outer face of~$\Gamma'$.
		To prove the lemma, it suffices to show that there is a nice sweeping sequence  for~$\mathcal G'$
		that uses all of its edges and starts with separation $s_1=(1,0,0,\dots,1)$
		and ends with separation $s_z=(1,|P_1|-1,|P_2|-1,\dots,|P_z|-1,1)$,
		which use~$e_\ell$ and~$e_r$, respectively.
		To this end, we show:
		
		\begin{claim}
			Let~$P$ be a directed (y-monotone) $ab$-path in~$\Gamma'$ and
			let~$s$ be the (unique) separation that corresponds to~$P$.
			Then there is a nice sweeping sequence from~$s_1$ to~$s$
			that uses all edges of~$\Gamma'$ that are not located to the right of~$P$.
			\label{claim:peeling}
		\end{claim}
		
		\begin{claimproof}
			The proof is by induction on the number~$i$ of triangles to the left of~$P$.
			If~$i=0$ (that is, $P=(e_\ell)$), the statement obviously holds.
			So assume~$i>0$.
			
			Suppose there is an edge~$(u,v)$ on~$P$ that is incident to a triangle~$uvw$,
			where~$w$ is to the left of~$P$ and $\gamma(u)<\gamma(w)<\gamma(v)$;
			we call this \emph{Configuration~(a)}.
			In this case, we can replace $(u,v)$ with $(u,w)$ and $(w,v)$
			to obtain a directed $ab$-path~$P'$ from~$P$
			with~$i-1$ triangles to its left and the claim follows by induction.
			
			Next, suppose that there is a triangle $uvw$ to the left of~$P$
			whose edges $(u,w)$ and $(w,v)$ belong to~$P$;
			we call this \emph{Configuration~(b)}.
			We can replace $(u,w)$ and $(w,v)$ with $(u,v)$
			to obtain a directed $ab$-path~$P'$ from~$P$
			with~$i-1$ triangles to its left and the claim follows by induction.
			
			We claim that one of these two configurations always exists.
			Towards a contradiction, assume otherwise.
			Let $e=(u,v)$ be an edge on~$P$ that is incident to a triangle~$\triangle=uvw$,
			where~$w$ is to the left of~$P$.
			Since Configuration~(a) does not exist,
			we may assume without loss of generality that $(\gamma(u)<)\gamma(v)<\gamma(w)$.
			Further, we may assume without loss of generality that $e$ is the topmost edge on~$P$ with this property.
			Since Configuration~(b) does not exist,
			the edge $(v,w)$ of~$\triangle$ does not belong to~$P$.
			Let~$x$ denote the successor of~$v$ along~$P$
			and let~$vxy$ denote the triangle to the left of $(v,x)$ on~$P$.
			Since the edge $(v,w)$ of~$\triangle$ is crossing-free and Configuration~(b) does not exist,
			it follows that $(\gamma(v)<)\gamma(x)<\gamma(y)$;
			a contradiction to the choice of $e$.
		\end{claimproof}
		
		The desired sequence from~$s_1$ to~$s_z$ is obtained by Claim~\ref{claim:peeling}, choosing~$P=(e_r)$.
	\end{proof}
	
	We remark that an exhaustive sweeping sequence using all edges corresponds directly to a
	particularly well-structured path decomposition.
	Thus, \cref{lem:existence-exhaustive-sequence} implies that every (ordered)
	level planar drawing of height~$h$ represents a graph of pathwidth at
	most~$h-1$; a statement that was independently proven
	in~\cite{dfklmnrrww-pclgd-Algorithmica08}.
	However, the path decompositions constructed in the proof
	of~\cite[Lemma 1]{dfklmnrrww-pclgd-Algorithmica08} do not exhibit the
	same properties that are inherent to exhaustive sweeping sequences
	and on which our algorithms heavily rely.
	In particular, these path decompositions may contain bags with
	multiple vertices of a given
	level (unless the drawing is proper). Moreover, the existence of a path
	decomposition of width at most $h-1$ for an ordered level graph~\G of
	height $h$ does not characterize the fact that~\G is ordered level planar
	(recall that, in fact,
	\OLP is NP-hard even when restricted to instances of
	pathwidth~$1$~\cite{DBLP:journals/talg/KlemzR19}).
	
	\paragraph{Computing suitable sweeping sequences.}
	In view of \cref{lem:existence-exhaustive-sequence}, recognizing ordered level planar graphs is equivalent to testing for the existence of suitable sweeping sequences. We now show that this can be done efficiently.
	
	\begin{lemma}
		\label{lem:algo-exhaustive-sequence}
		There is an algorithm that determines whether a given ordered level graph admits an exhaustive sweeping sequence using all of its edges.
		It can be implemented deterministically using
		$\mathcal O^*(2^{\binom{h}{2}}\prod_{j\in [h]}(2\lambda_j+1))
		\subseteq\mathcal O^*(2^{\binom{h}{2}}(2\lambda +1)^h)
		\subseteq \mathcal O^*(2^{\binom{h}{2}}(2n+1)^h)$
		time and space, or
		nondeterministically using polynomial time and $\mathcal O(h^2 + h\log I)$ space,
		where~$\lambda$ and~$h$ denote the width and height of the input graph, respectively,
		$n$ denotes the number of vertices,
		$\lambda_j$ denotes the width of level $j\in [h]$, and
		$I$ denotes the input size.
		Further, the deterministic version can report the sequence (if it exists).
	\end{lemma}

	\begin{proof}
		Let $\mathcal G=(G=(V,E),\gamma,(\prec_i)_{1\le i\le h})$ be an ordered level graph,
		let $s=(p_1,p_2,\dots,p_h)$ be a separation for~$\mathcal G$,
		and let~$U\subseteq E$ 
		be a subset of the edges that are joining pairs of (not necessarily
		consecutive) vertices on~$s$.
		We define $T[s,U]=\texttt{true}$ if there exists a nice sweeping sequence for~$\mathcal G$ that
		starts with $s_1=(0,0,\dots,0)$,
		ends with~$s$,
		and uses all edges in~$U$,
		as well as all edges in~$E$ incident to at least one vertex to the left of~$s$.
		Otherwise, $T[s,U]=\texttt{false}$.
		Additionally, we allow $U=\bot$, in which case $T[s,U]=\texttt{false}$.
		\cref{fig:olp-DP-example} illustrates several \texttt{true} table entries $T[s_i,U_i]$ along with the corresponding sweeping sequences $s_1,s_2,\dots,s_i$.
		Our goal is to determine $T[(|P_1|-1,|P_2|-1,\dots,|P_h|-1),\emptyset]$, which is \texttt{true} if and only if there exists an exhaustive sweeping sequence for~$\mathcal G$ using all edges in~$E$.
		
		\myproofparagraph{Recurrence relation.}
		We will determine the entries $T[s,U]$ by means of a dynamic programming recurrence.
		For the base case, we simply set $T[(0,0,\dots,0),\emptyset]=\texttt{true}$.
		Now assume that~$s\neq (0,0,\dots,0)$.
		For each index $1\le j\le h$ where $p_j\ge 1$, 
		we define a separation
		\[s_j'=(p_1,p_2,\dots,p_{j-1},p_j-1,p_{j+1},\dots,p_h).\]
		Further, we define an edge set $U_j'\subseteq E$ as follows:
		\begin{itemize}%
			\item If $p_j$ represents a vertex~$v$, then
			\begin{itemize}%
				\item $U_j'=\bot$ if $U$ contains edges joining $v$ with vertices on $s$ that are not its predecessor or successor on~$s$;
				\item otherwise $U_j'$ is created from~$U$ by removing the (up to two) edges incident to~$v$.
			\end{itemize}
			\item If $p_j$ represents a gap and, thus, $p_{j-1}$ represents a vertex~$v$,
			then
			\begin{itemize}%
				\item $U_j'=\bot$ if $v$ is adjacent to a vertex to the right of $s$;
				\item otherwise~$U_j'$ is created from~$U$ by
				\begin{itemize}%
					\item
					removing the edge between the predecessor and successor of~$v$ along~$s_j'$ (if it exists and is contained in~$U$); and
					\item adding all edges in~$E$ that join~$v$ with some vertex on~$s_j'$.
				\end{itemize}
			\end{itemize}
		\end{itemize}
		
		\begin{figure}[tb]
			\centering
			\includegraphics[page=2]{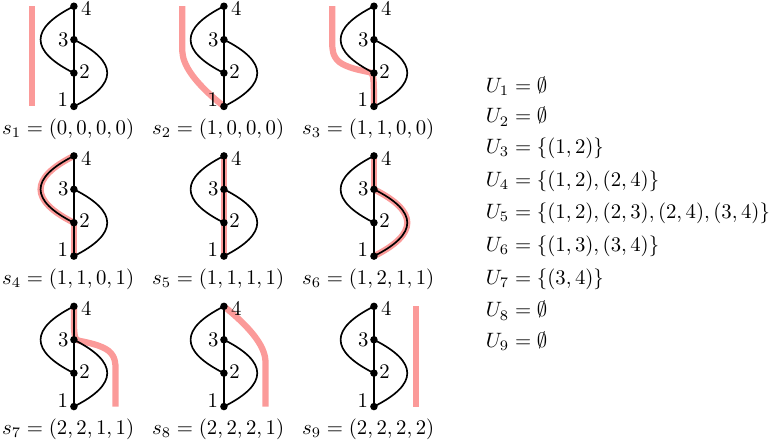}
			\caption{
				An exhaustive sweeping sequence using all edges of the depicted graph
				and the corresponding (cf.\ \cref{lem:sweep-algo}) ordered level planar drawing,
				as well as the corresponding sequence of \texttt{true} dynamic programming
				table entries $T[s,U]$ (cf.\
				\cref{lem:algo-exhaustive-sequence}).
				Note that in any ordered level planar drawing of the graph,
				exactly one of the edges $(2,4),(1,3)$ is located  to the left of the path $(1,2,3,4)$, while the other is located to the right.
				Similarly, in
				any exhaustive sweeping sequence containing separation $s_5$, exactly one of the edges $(2,4),(1,3)$ is used by a separation preceding $s_5$, while the other is used by a separation succeeding $s_5$.
				Hence, when iteratively building an exhaustive sweeping sequence, it is key to remember which edges between vertices of the current separation have already been used~-- this is exactly the purpose of the sets $U$.
				E.g., the fact that $(2,4)\in U_5$ corresponds to $(2,4)$
				being used before $s_5$ (in $s_4$), from which one can infer
				how to proceed.
			}%
			\label{fig:olp-DP-example}
		\end{figure}
		
		For illustrations, refer to \cref{fig:olp-DP-example} (where for each $i\in [9]$ and $s=s_i$ and $U=U_i$, we have $s_j'=s_{i-1}$ and $U_j'=U_{i-1}$ for some $j\in [4]$). The following technical claim describes our recurrence relation.
		
		\begin{claim}\label{lem:DP}
			$T[s,U]=\texttt{true}$ if and only if there exists an index $1\le j\le h$ where $p_j\ge 1$ and such that $T[s_j',U_j']=\texttt{true}$.
		\end{claim}

		\begin{claimproof}
			For the ``if''-direction, assume that there exists an index $1\le j\le h$ where $p_j\ge 1$ and such that $T[s_j',U_j']=\texttt{true}$.
			Then there exists a nice sweeping sequence
			$S'=(s_1,s_2,\dots,s_z=s_j')$ for~$\mathcal G$ that
			starts with $s_1=(0,0,\dots,0)$,
			ends with~$s_j'$,
			and uses all edges in~$U_j'$,
			as well as all edges in~$E$ incident to at least one vertex to the left of~$s_j'$.
			Let~$S$ denote the nice sweeping sequence obtained by appending~$s$ to~$S'$.
			We will show that~$S$ is a certificate for the fact that $T[s,U]$ is indeed  \texttt{true}.
			
			We start by showing that~$S$ uses all edges that are incident to at least one vertex to the left of~$s$.
			Let~$e\in E$ be an edge with an endpoint~$u$ to the left of~$s$.
			If at least one endpoint of~$e$ is also to the left of~$s_j'$, the definition of~$S'$ implies that ($S'$ and, hence) $S$ uses~$e$.
			So assume that no endpoint of~$e$ is to the left of~$s_j'$.
			It follows that $p_j$ represents a gap, $p_{j-1}$ represents the vertex~$u$, and the other endpoint of~$e$ is on $s_j'$ (it cannot lie to the right of ($s$ and, hence) $s_j'$ since $T[s_j',U_j']=\texttt{true}$ and, hence, $U_j'\neq \bot$).
			Consequently, $e\in U_j'$ (by definition of~$U_j'$) and, thus, ($S'$ and) $S$ uses~$e$, as desired.
			
			It remains to show that~$S$ uses all edges in~$U$.
			Obviously, ($s$ and) $S$ use all edges in~$U$ that join two vertices that are consecutive along~$s$.
			So let~$e\in U$ be an edge that joins two vertices that are nonconsecutive along~$s$.
			When creating~$U_j'$ from~$U$, only edges between consecutive vertices on~$s$ are removed.
			Hence, $e\in U_j'$ and, thus, ($S'$ and) $S$ uses~$e$, as desired.
			
			Altogether, it follows that~$S$ is indeed a certificate for the fact that $T[s,U]$ is \texttt{true}.
			
			For the ``only if''-direction, assume that $T[s,U]=\texttt{true}$.
			Then there exists a nice sweeping sequence
			$S=(s_1,s_2,\dots,s_z=s)$ for~$\mathcal G$ that
			starts with $s_1=(0,0,\dots,0)$,
			ends with~$s$,
			and uses all edges in~$U$,
			as well as all edges in~$E$ incident to at least one vertex to the left of~$s$.
			Since~$S$ is nice, there is an index $1\le j\le h$ such that
			\[s_{z-1}=(p_1,p_2,\dots,p_{j-1},p_j-1,p_{j+1},\dots,p_h),\]
			i.e., $s_{z-1}=s_j'$.
			Let~$S'=(s_1,s_2,\dots,s_{z-1}=s')$.
			We will show that~$S'$ is a certificate for the fact that $T[s_j',U_j']$ is indeed  \texttt{true}.
			
			We begin by showing that~$U_j'\neq \bot$.
			To this end, assume otherwise, i.e., $U_j'=\bot$.
			We distinguish two cases.
			First, assume that $p_j$ represents a vertex~$v$.
			This implies that $U$ contains an edge~$e$ joining $v$ with a vertex on $s$ that is not its predecessor or successor.
			The edge~$e$ is not used by~$s$.
			However, it is also not used by~$S'$ since~$v$ is to the right of~$s_j'$;
			a contradiction to the fact that~$S$ uses all edges in~$U$.
			Second, assume that~$p_j$ represents a gap and, thus, $p_{j-1}$ represents a vertex~$v$.
			This implies that~$v$ has an edge~$e$ to a vertex to the right of~$s$.
			Thus, the sequence $S$ does not use~$e$.
			However, the endpoint~$v$ of~$e$ is to the left of~$s$;
			a contradiction to the fact that~$S$ uses all edges with at least one endpoint to the left of~$s$.
			Altogether, this shows that indeed~$U_j'\neq \bot$.
			
			We now show that~$S'$ uses all edges that are incident to at least one vertex to the left of~$s_j'$.
			To this end, let~$u\in V$ be a vertex to the left of~$s_j'$.
			By definition of~$s_j'$, this vertex is also to the left of~$s$.
			By definition of~$S$, all edges incident to~$u$ are used by~$S$.
			However, none of these edges can be used by~$s$ (since $u$ is to the left of~$s$) and,
			hence, each of these edges is used by (some separation in)~$S'$, as desired.
			
			It remains to show that~$S'$ uses all edges in~$U_j'$.
			To this end, let~$e\in U_j'$.
			First, assume that $e\in U_j'\setminus U$.
			The definition of~$U_j'$ implies that $p_j$ represents a gap and $p_{j-1}$ represents a vertex~$v$ that is incident to~$e$.
			By definition of~$S$ and the fact that~$v$ is to the left of~$s$, it follows that~$S$ uses~$e$.
			In fact, since~$s$ cannot use~$e$, it follows that (some separation in)~$S'$ uses~$e$, as desired.
			It remains to consider the case $e\notin U_j'\setminus U$.
			Note that this case assumption implies that both endpoints of~$e$ belong to~$s$.
			We distinguish two subcases:
			first assume that the endpoints of~$e$ are consecutive vertices along~$s$.
			Combining this assumption with the fact that $e\in U_j'\cap U$ and the definition of~$U_j'$ implies that the endpoints of~$e$ are also consecutive vertices along~$s_j'$.
			Hence, ($s'$ and) $S'$ uses~$e$, as desired.
			For the other subcase, assume that the two endpoints of~$e$ are nonconsecutive along~$s$.
			In this case, the separation~$s$ cannot use~$e$ and, hence, the definition of~$S$ implies that (some separation in)~$S'$ uses~$e$, as desired.
			
			Altogether, we have established that~$S'$ is indeed a certificate for the fact that~$T[s_j',U_j']$ is \texttt{true}.
			This concludes the proof of the claim.
		\end{claimproof}
		
		In view of Claim~\ref{lem:DP}, we can now state our algorithms.
		
		\myproofparagraph{Deterministic construction algorithm.}
		Our dynamic programming table~$T$ has a total of
		$2^{\binom{h}{2}}\prod_{j\in [h]}(2\lambda_j+1)$
		entries.
		To compute the value of a table entry $T[s,U]$, we simply construct the tuples $(s_1',U_1'),\dots,(s_h',U_h')$ and then set $T[s,U] = \bigvee_{j \in [h]} T[s_j',U_j']$ (the set~$U_j'$ is only constructed and $T[s_j',U_j']$ only taken into account when $s_j'$ contains no negative entry).
		Disregarding the time spend for the recursive calls,
		the time to process a table entry is clearly polynomial.
		Hence, by employing memoization, we can fill the table using $\mathcal O^*(2^{\binom{h}{2}}\prod_{j\in [h]}(2\lambda_j+1))$ time and space.
		To construct the solution, we can employ the usual back-linking strategy.
		In particular, for each \texttt{true} entry~$T[s,U]$,
		it suffices to store an index~$j$ such that~$T[s_j',U_j']$ is \texttt{true},
		which does not change the asymptotic time/space requirements.
		\cref{fig:olp-DP-example} illustrates a sequence of table entries
		corresponding to a solution.
		
		\myproofparagraph{Nondeterministic decision algorithm.}
		We perform $2n$ steps.
		In each step, we non\-deter\-mi\-ni\-sti\-cally guess the next separation $s$
		and its set of prescribed edges~$U$,
		and check whether for some $j\in [h]$,
		the previous separation $s'$ is equal to~$s_j'$
		and the previous set of prescribed edges~$U'$ is equal to~$U_j'$
		(starting with $s'=(0,0,\dots,0)$ and $U'=\emptyset$).
		At each point in time, we only need to keep two separations and
		two edge sets in memory.
		A separation can be stored using~$\mathcal O(h\log I)$ space.
		Each edge set $U$ can be encoded by means of a $h\times h$ Boolean adjacency
		matrix of size~$\mathcal O(h^2)$.
		The verification in each step is easy to do using polynomial time and
		$\mathcal O(\log I)$ additional space
		(in particular, we only need to store $\mathcal O(1)$ pointers~/ indices).
	\end{proof}
	
	As a corollary of \cref{lem:sweep-algo,lem:existence-exhaustive-sequence,lem:algo-exhaustive-sequence},
	we obtain the algorithmic statements in \cref{thm:olp-XNLP-complete}:
	
	\begin{theorem}
		\OLPlong parameterized by
		the height is \XNLP.
		Moreover, there is a constructive \XP-time algorithm for \OLPlong parameterized by the height.
		\label{thm:olp-XNLP}
	\end{theorem}
	
	\section{\XNLP-Hardness of Ordered Level Planarity}
	\label{sec:olp-hardness}

	For the ease of presentation, 
	we first show that the \OLP problem is \Wh.
	To this end, we use a {\em parameterized reduction}\footnote{For an
		overview of this standard technique, refer to a standard
		textbook~\cite{CyganFKLMPPS2015}.}
	from \textsc{Multicolored Independent Set} (defined below) with $k$ colors
	to \textsc{Ordered Level Planarity} with $\Oh(k)$ levels.
	We then describe how to extend this reduction to obtain \XNLPhness of \OLP.
	We start by introducing the building blocks of our reductions.
	Recall that, normally, the level assignment $\gamma$ of a level graph surjectively maps to a set $[h]$ of consecutive numbers.
	In this section, to facilitate the description of our gadgets, we relax this condition by temporarily allowing level graphs in which not every level is occupied~-- nevertheless, the final outcome of our reduction will be an ordered level graph~$\G$ in the original sense.
	
	\paragraph{Basic building blocks of our reduction.}
	Our construction of~$\G$ is heavily based on two gadgets
	that we call \emph{plugs} and \emph{sockets} (a very basic version of these gadgets was already used earlier, in an NP-hardness proof for \CLP~\cite{DBLP:conf/soda/BrucknerR17}; here we introduce generalized versions).
	We define both in terms of the list of levels their vertices occupy.
	Their $\prec_i$ orderings are according to the indices of their vertices.
	The vertices are deliberately given in an unintuitive
	order to allow
	for the ordering in $\prec_i$ by indices
	and to make the (degenerate) cases behave nicely later.
	Let $\ell_1, \ell_2, \ell_3, \ell_4 \in [h]$
	such that $\ell_1 < \ell_2 < \ell_3 < \ell_4$.
	A (non-degenerate) \emph{$(\ell_1,\ell_2,\ell_3,\ell_4)$-plug}, see \cref{fig:plugs_sockets_a}, contains vertices $u_1$, $u_2$, $u_3$, $u_4$, $u_5$, and $u_6$,
	where $\gamma(u_5)=\ell_1$,
	$\gamma(u_3)=\gamma(u_6)=\ell_2$,
	$\gamma(u_2)=\gamma(u_4)=\ell_3$,
	and $\gamma(u_1)=\ell_4$.
	It contains the edges $u_1u_2$, $u_2u_3$, $u_3u_4$, $u_4u_6$, and $u_6u_5$;
	i.e., a~plug is a path that traverses its four levels in the order
	$\ell_4,\ell_3,\ell_2,\ell_3,\ell_2,\ell_1$.
	Similarly, an \emph{$(\ell_1,\ell_2,\ell_3,\ell_4)$-socket}, see \cref{fig:plugs_sockets_b}, consists of
	vertices $v_1, v_2, \dots, v_{10}$ such that $\gamma(v_3)=\ell_1$,
	vertices $v_2,v_4,v_5,v_9$ occupy level $\ell_2$,
	vertices $v_1,v_6,v_7,v_{10}$ occupy level $\ell_3$,
	and $\gamma(v_8)=\ell_4$.
	It contains the edges $v_1v_2$, $v_2v_3$, $v_3v_5$, $v_5v_7$, $v_4v_6$, $v_6v_8$, $v_8v_{10}$, and $v_{10}v_9$.
	Observe that a socket consists of two disconnected paths
	whose vertices interleave on levels $\ell_2$ and $\ell_3$.
	Let vertices~$v_1$ and~$v_9$ of plugs and
	vertices~$u_1$ and~$u_5$ of sockets be \emph{connecting vertices}.
	Connecting vertices of sockets will be identified with other vertices of the construction
	and connecting vertices of some plugs may be connected to other plugs via additional edges.
	
	\begin{figure}[tb]
		\centering
		\begin{subfigure}[t]{.18\textwidth}
			\centering
			\includegraphics[page=2]{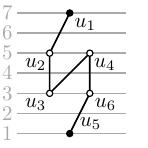}
			\subcaption{\centering}
			\label{fig:plugs_sockets_a}
		\end{subfigure}
		\hfill
		\begin{subfigure}[t]{.2\textwidth}
			\centering
			\includegraphics[page=3]{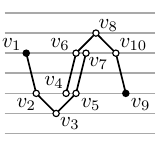}
			\subcaption{\centering}
			\label{fig:plugs_sockets_b}
		\end{subfigure}
		\hfill
		\begin{subfigure}[t]{.17\textwidth}
			\centering
			\includegraphics[page=4]{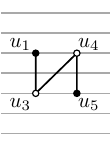}
			\subcaption{\centering}
			\label{fig:plugs_sockets_c}
		\end{subfigure}
		\hfill
		\begin{subfigure}[t]{.17\textwidth}
			\centering
			\includegraphics[page=5]{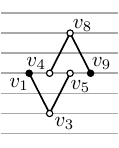}
			\subcaption{\centering}
			\label{fig:plugs_sockets_d}
		\end{subfigure}
		\hfill
		\begin{subfigure}[t]{.23\textwidth}
			\centering
			\includegraphics[page=6]{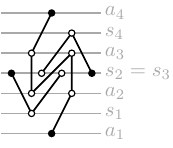}
			\subcaption{\centering~~~~~~~~~~}
			\label{fig:plugs_sockets_e}
		\end{subfigure}
		\caption{Plugs and sockets; (a) a $(1,3,5,7)$-plug,
			(b)~a~$(2,3,5,6)$-socket,
			(c)~a~degenerate $(3,3,5,5)$-plug,
			(d)~a~degenerate $(2,4,4,6)$-socket,
			(e) a $(1,3,5,7)$-plug that is linked
			to a degenerate $(2,4,4,6)$-socket.
			Connecting vertices are filled in black.
			Note how in degenerate gadgets ((c) and (d)), the edges
			and vertices of the repeated levels are contracted.
		}
		\label{fig:plugs_sockets}
	\end{figure}
	
	Now we lift the strict inequality restriction on the levels of our gadgets.
	For plugs we only require $\ell_1 \le \ell_2 < \ell_3 \le \ell_4$ and for sockets we require $\ell_1 < \ell_2 \le \ell_3 < \ell_4$
	and we call a plug or a socket \emph{degenerate} if it has at least one pair of repeated levels,
	i.e., for some $i \in [3]$, $\ell_i=\ell_{i+1}$.
	We create the degenerate plugs and sockets by contracting the edges
	between vertices of the repeated levels while keeping the vertex with the lower index;
	see \cref{fig:plugs_sockets_c,fig:plugs_sockets_d}.
	
	Let $P$ be a plug, and let $\mathrm{com}(P)$ be the connected component of~$P$ in $G$.
	A plug \emph{fits} into a socket if the gadgets could be ``weaved'' as illustrated in \cref{fig:plugs_sockets_e}.
	Formally, we say an $(a_1,a_2,a_3,a_4)$-plug~$P$ fits into a
	$(s_1,s_2,s_3,s_4)$-socket~$S$ when $\min_{v \in \mathrm{com}(P)}\{\gamma(v)\} \le s_1$, $\max_{v \in \mathrm{com}(P)}\{\gamma(v)\} \ge s_4$,
	and $s_1 < a_2 < s_2 \le s_3 < a_3 < s_4$.
	In an ordered level planar drawing of~\G, we say a plug~$P$ \emph{links} to
	an $(s_1,s_2,s_3,s_4)$-socket~$S$ when $P$ fits into $S$,
	and $P$ is drawn between the connecting vertices of~$S$
	(that is, the edges of $P$ traversing level~$s_3$ are to the right of~$v_1$ and
	the edges of $P$ traversing level~$s_2$ are to the left of~$v_9$);
	see \cref{fig:plugs_sockets_e}.
	
	A defining feature of our constructions is a division into
	vertical strips, which is accom\-modated by the following notion.
	Given a set $L$ of levels, a \emph{wall} of~$L$ is a path that starts in a
	vertex on the bottommost level in~$L$, goes through a vertex on each
	intermediate level of~$L$, and ends at a vertex on the topmost level of~$L$.
	Note that each gadget type essentially has a unique drawing in the sense that it corresponds to an ordered level planar graph for which all of its ordered level planar drawings have the same level planar embedding.
	We use \cref{lem:two_plugs_one_socket} to design specific
	plugs that can or cannot link to specific
	sockets in the same ordered level planar drawing.
	
	\begin{lemma}\label{lem:two_plugs_one_socket}
		Consider an ordered level graph $\G$ that contains
		an $(a_1,a_2,a_3,a_4)$-plug $A$, a $(b_1,b_2,b_3,b_4)$-plug $B$, and an $(s_1,s_2,s_3,s_4)$-socket $S$
		that occupy three disjoint sets of levels.
		There is an ordered level planar drawing of the subgraph of $\G$ spanned by $A$, $B$, and $S$ where
		$A$ and $B$ link to $S$ if and only if both $A$ and $B$ fit
		into $S$ and $a_2 < b_2$ and $a_3 < b_3$, or, both vice versa,
		$a_2 > b_2$ and $a_3 > b_3$.
	\end{lemma}
	
	\begin{proof}
		To distinguish the vertices of the plugs $A$ and $B$,
		we use the name of the plug as superscript, so, e.g., we refer to $u_3^A$ and $u_3^B$.
		As we have only one socket,
		we use just $v_1, \dots, v_{10}$ for the vertices of~$S$.
		For the names of the vertices, we assume
		that none of the involved gadgets is degenerate.
		If some of the gadgets are degenerate,
		adjusting the vertex names is straight-forward.
		
		\begin{figure}[tb]
			\centering
			\begin{subfigure}[t]{.45\textwidth}
				\centering
				\includegraphics[page=1]{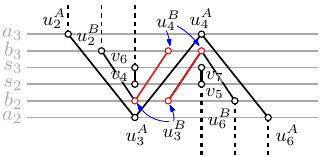}
				\subcaption{If the order of the layers follows
					the pattern $a_2, b_2, b_3, a_3$,
					then there is no crossing-free drawing
					that has both $A$ and $B$ linked to~$S$.}
				\label{fig:two-plugs-one-socket-fail}
			\end{subfigure}
			\hfill
			\begin{subfigure}[t]{.5\textwidth}
				\centering \includegraphics[page=2]{two-plugs-one-socket}
				\subcaption{If the order of the layers follows
					the pattern $a_2, b_2, a_3, b_3$,
					then there is a crossing-free drawing
					that has both $A$ and $B$ linked to~$S$.}
				\label{fig:two-plugs-one-socket-success}
			\end{subfigure}
			\caption{Sketches for the proof of
				\cref{lem:two_plugs_one_socket}.
				The edges of the socket~$S$ are bold.}
		\end{figure}
		
		For the ``only if''-direction.
		Let an ordered level planar drawing where both $A$ and $B$ link to~$S$ be given.
		By the definition of linking, $A$ and $B$ fit into~$S$
		and, hence, $a_2, b_2 < s_2 \le s_3 < a_3, b_3$.
		Suppose for a contradiction that $a_2 < b_2$ and $a_3 > b_3$;
		see \cref{fig:two-plugs-one-socket-fail} for an illustration.
		(The case $a_2 > b_2$ and $a_3 < b_3$ is identical up to
		renaming $A$ and $B$.)
		By the definition of linking and due to the two (uncrossed) paths of $S$,
		the edges $u_2^A u_3^A$ and $u_2^B u_3^B$
		are to the left of $v_4$ and $v_6$,
		while the edges $u_3^A u_4^A$ and $u_3^B u_4^B$
		are to the right of $v_4$ and $v_6$.
		Given $a_2 < b_2$, it follows that this is possible if and only if
		$u_2^B$, $u_3^B$, and $u_4^B$ lie (horizontally) between
		the edges $u_2^A u_3^A$ and $u_3^A u_4^A$.
		Symmetrically,
		the edges $u_3^A u_4^A$ and $u_3^B u_4^B$ are to the left of $v_7$ and $v_5$,
		while the edges $u_4^A u_6^A$ and $u_4^B u_6^B$
		are to the right of $v_7$ and $v_5$.
		Given $a_3 > b_3$, it follows that this is possible if and only if
		$u_3^B$, $u_4^B$, and $u_6^B$ lie between
		the edges $u_3^A u_4^A$ and $u_4^A u_6^A$.
		Therefore, $u_3^B$ and $u_4^B$ lie at the
		same time to the left and to the right
		of $u_3^A u_4^A$, which is a contradiction.
		
		For the ``if''-direction.
		Assume that $s_1 < a_2 < b_2 < s_2 < s_3 < a_3 < b_3 < s_4$;
		see \cref{fig:two-plugs-one-socket-success} for an illustration.
		The case $s_1 < b_2 < a_2 < s_2 < s_3 < b_3 < a_3 < s_4$
		is identical up to renaming $A$ and $B$.
		We can get an ordered level planar drawing where
		both plugs $A$ and $B$ are to the right of~$v_1$,
		and $A$ and $B$ are to the left of~$v_9$ if we place
		the edges occurring between the involved layers
		in the following order from left to right:
		first $v_1 v_2$ and $v_2 v_3$,
		then $u_2^A u_3^A$, then $u_2^B u_3^B$,
		then $v_4 v_6$ and $v_6 v_8$,
		then $u_3^B u_4^B$, then $u_3^A u_4^A$,
		then $v_3 u_5$ and $v_5 v_7$,
		then $u_4^A u_6^A$, then $u_4^B u_6^B$,
		and finally $v_8 u_{10}$ and $v_{10} v_9$.
	\end{proof}
	
	In an (ordered) level planar drawing of~\G,
	we define spatial relations between two gadgets in the following way.
	For a gadget~$X$,
	we say a level $\ell$ (or a gadget~$Y$) is \emph{above} $X$
	if the level $\ell$ (each level hosting vertices of~$Y$) has a larger index
	than any level hosting a vertex of $X$.
	Symmetrically, we define \emph{below} for the levels with a smaller index.
	We say a gadget~$Y$ is \emph{left} (resp.\ \emph{right}) of~$X$
	if on every level all vertices and edges of $Y$ lie to the left (resp. right) of all the vertices of~$X$.
	
	We \emph{connect} plugs and sockets with other gadgets using their connecting vertices.
	The spatial relation of two gadgets is usually clear, so we shortcut
	the description of how the connection is formed in the following way.
	Connecting a socket $S$ to a wall $W$ is done by identifying a vertex $x$
	of $S$ with the vertex of $W$ on level $\gamma(x)$.
	The $x=v_1$ if $W$ is left of $S$ and $x=v_5$ if $W$ is right of $S$.
	Connecting a plug $P$ to a gadget $Y$ means adding an edge from
	a connecting vertex $z$ of $P$ to $Y$.
	We $z=u_1$ if $Y$ is above $P$ and $z=u_5$ if $Y$ is below $P$.
	The other endpoint of the edge depends on $Y$.
	In our case, $Y$ is either a plug or a single vertex.
	In case $Y$ is a plug its connecting vertex
	is determined in the same way we did it for $P$.
	Hence, the connection of $P$ to $Y$ results
	in an edge between $u_1$ of $P$ and $u_5$ of~$Y$, or vice-versa.
	See connecting vertices of plugs and sockets in \cref{fig:plugs_sockets}.
	
	An ordered level planar graph is called \emph{rigid} if all of its
	ordered level planar drawings have a common level planar embedding.
	Recall that $\G$ is proper if $\gamma (v)=\gamma (u)+1$ for every edge~$(u,v)\in E(G)$.
	
	\begin{observation}\label{obs:proper_is_rigid}
		Every proper ordered level planar graph is rigid.
	\end{observation}
	
	\begin{observation}\label{obs:proper_gadgets}
		A $(\ell_1,\ell_2,\ell_3,\ell_4)$-plug or a $(\ell_1,\ell_2,\ell_3,\ell_4)$-socket correspond to a proper ordered level graph if and only if $\ell_{i+1} \in \{\ell_i,\ell_i+1\}$ for every $i\in [3]$.
	\end{observation}
	
	By \emph{inserting a level $i$} into an ordered level graph~\G, we
	mean that the number $h$ of levels of~\G increases by one and,
	for each vertex $v$ of~\G with $\gamma(v) \ge i$, its level
	$\gamma(v)$ increases by one.
	
	\begin{observation}\label{obs:subdivisions}
		A rigid ordered level graph remains rigid even if we insert new levels.
		Moreover, a~rigid graph $\G$ remains rigid even if we subdivide
		an edge~$uv$ by a new vertex~$w$, i.e.,
		if we replace the directed edge $uv$ by the directed edges $uw$ and $wv$,
		and set the x-coordinate of $w$ according to where $uv$ crossed $L_{\gamma(w)}$
		in the unique drawing of $\G$.
	\end{observation}
	
	From Observations~\ref{obs:proper_is_rigid},
	\ref{obs:proper_gadgets}, and \ref{obs:subdivisions},
	it follows that plugs, sockets, and walls are all rigid
	since we can start with a proper ordered level graph
	and then insert additional levels.
	
	We are now done with the auxiliary definitions.
	We first give the full proof of $W[1]$-hardness of \OLP.
	Then we extend the construction to show XNLP-hardness of \OLP.
	Last in this section, we show that the XNLP-hardness result
	also works for connected ordered level graphs.
	
	\paragraph{\Whness.}
	\textsc{Multicolored Independent Set} (MCIS)
	is a well-known \Wh problem~\cite{FellowsHRV09,Pietrzak03}.
	In this problem, we are given a graph $H$, an integer $k$
	(the parameter), and a $k$-partition $C_1,C_2,\dots,C_k$ of $V(H)$,
	and the task is to decide whether $H$ has an independent set $X \subseteq V(H)$ that contains,
	for every $j \in [k]$, exactly one vertex of~$C_j$.
	We may assume without loss of generality
	that in $H$, there is no edge whose endpoints have the same color~$C_j$.
	
	\begin{lemma}
		\label{thm:olp_w_hard}
		\OLPlong is \Wh with respect to the height.
	\end{lemma}
	
	\begin{proof}
		Let $\mathcal S$ be an instance of MCIS, that is, a graph $H$, an
		integer $k$, and a $k$-partition $C_1,C_2,\dots,C_k$ of $V(H)$.
		We prove the theorem by constructing an \OLP instance $\G$, with $h \in \mathcal O(k)$ levels, that has an ordered level planar drawing if and only if $\mathcal S$ contains a multicolored independent set of size $k$.
		
		Let $n = |V(H)|$ and $m = |E(H)|$.
		For a color $j \in [k]$, let $m_j=|\{uv \colon
		c(u)=j \lor c(v)=j\}|$, i.e., $m_j$ is the number of edges incident to
		a vertex with color $j$.
		Note that $\sum_{j=1}^k m_j = 2m$ by double counting
		since we assume that the colors of the endpoints of
		every edge are distinct.
		For this proof, we also assume that,
		without loss of generality, the color classes have
		the same size $n'$, i.e., $n' = |C_j| = n/k$ for each $j \in [k]$.

		\begin{figure}[tb]
			\centering
			\includegraphics[page=2]{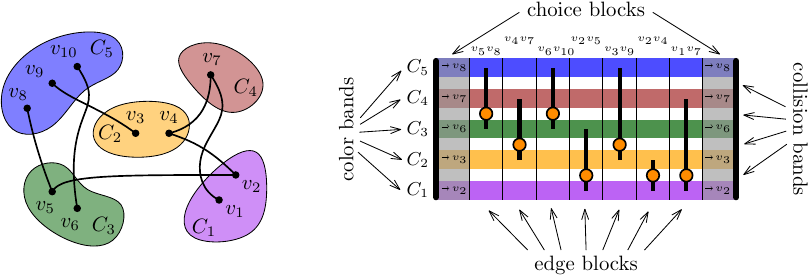}
			\caption{
				Example of our parameterized reduction from
				\textsc{Multicolored Independent Set} (MCIS) to \OLPlong.
				On the left side, there is an instance of MCIS with $k=5$ colors.
				On the right side, there is the schematized grid structure
				of the ordered level graph constructed from the MCIS instance.
				The orange disks in the edge blocks represent the places where
				the collision sockets are placed.
				Here, the solution $\{v_2, v_3, v_6, v_7, v_8\}$ for the MCIS instances is found.
			}%
			\label{fig:w1h_example}
		\end{figure}
		
		\myproofparagraph{Bands, columns, blocks, and cells.}
		We divide the construction vertically and horizontally in the following way, see \cref{fig:w1h_example}.
		The levels are partitioned by their purpose into five types -- \emph{rigid}, \emph{color}, \emph{high}, \emph{pass-through}, and \emph{collision} levels.
		We use (horizontal) \emph{bands} to denote consecutive sets of levels
		that belong together conceptually;
		we use two band types called \emph{color} bands and \emph{collision} bands.
		We consider any vertices that lie on levels of a band
		and any edges between such vertices to be part of the band.
		Note that bands contain a mixture of level types.
		We typically use relative numbering to speak about the levels
		of a band, e.g., the lowest level of a band is numbered $1$.
		The vertical slices of the construction are delimited by the walls.
		
		We begin the construction with $7 \cdot k + 3 \cdot (k-1) = 10k-3$ rigid levels.
		We group these rigid levels into $k$ color bands that alternate
		with $k-1$ collision bands:
		Each color band contains $7$ rigid levels and each collision band contains $3$ rigid levels.
		Let the \emph{color-$j$ band} be the $j$-th color band from the bottom and let it
		be associated with color $j$.
		Similarly, the \emph{collision-$j$ band} denotes the $j$-th collision band from the bottom, i.e., the collision band between the color-$j$ and the color-$(j+1)$ band.
		We later enlarge the bands by inserting more levels using \Cref{obs:subdivisions}.
		
		We add $1 + 2 \cdot (n'-1) + m \cdot (2 n'-1) = (m+1)\cdot
		(2n'-1)$ walls that divide each level created so far into $m \cdot (2n'-1)+2n'-2$ line segments and two rays.
		We order the line segments naturally by $\prec_\ell$ for each level $\ell$.
		For $i \in [m \cdot (2n'-1)+2n'-2]$, we let the $i$-th
		\emph{column} be the union of the $i$-th line segments
		over all levels, vertices placed on those line segments, and edges between those vertices;
		we will not use the rays.
		The first and the last $(n'-1)$ columns on the left/right
		are called left/right \emph{choice block}.
		We partition the remaining $m \cdot (2 n'-1)$ columns into $m$
		\emph{edge blocks} each with $2 n'-1$ consecutive columns.
		Every edge block is associated with one edge of $E(H)$.
		Let the \emph{edge-$uv$ block} be the edge block
		associated with the edge $uv \in E(H)$.
		We typically use relative numbering to denote columns in a block;
		so the leftmost column would be the first one etc.
		
		We call the intersection of a band and a column a \emph{cell}.
		So far, the construction contains many empty cells.
		The aim is to fill each color band cells with plugs and sockets of various types to create a ``shifting'' mechanism that represents a selection of a vertex of that color.
		The collision bands come into play later to prevent a selection of neighboring vertices.
		
		\myproofparagraph{Placement of sockets.}
		Now we place the sockets into the cells.
		Let us call the two walls on the boundary of a cell its left and right wall.
		By \emph{adding} a socket to a cell we mean adding the socket to
		$\mathcal{G}$ and identifying its right connecting vertex with the
		right wall vertex on the same level and doing the same for the left one.
		See \cref{fig:w_hard_gadgets_alone} to see sockets placed in cells.
		Recall that color bands contain $7$ rigid levels.
		We add a $(2,4,4,6)$-socket, see \cref{fig:plugs_sockets_d}, to each cell in all color bands within the
		choice blocks (note that we use relative level indices).
		We call these \emph{choice sockets}.
		For every $uv\in E(H)$ we add $(3,4,4,5)$-sockets, called \emph{color sockets}, to all cells that are in the edge-$uv$ block and in either color-$c(u)$ band or the color-$c(v)$ band.
		To all cells in the color-$i$ bands in the edge-$uv$ block where $i$ is different from $c(u)$ and $c(v)$,
		we add $(1,2,6,7)$-sockets, called \emph{pass-through sockets}.
		Recall that collision bands contain three rigid levels.
		We do not add any sockets to the collision band cells that are in the choice blocks.
		In each edge block, we add a single socket as follows.
		For $uv\in E(H)$, let $q_{uv} = \min\{c(u),c(v)\}$;
		note that $q_{uv}<k$.
		Into the collision-$q_{uv}$ band of edge-$uv$ block, we add to a $n'$-th cell
		a $(1,2,2,3)$-socket, which we call \emph{collision socket}.
		See construction overview in \cref{fig:w1h_example}.
		
		\begin{figure}[tb]
			\centering
			\includegraphics[page=1]{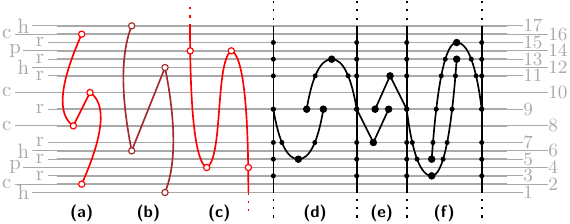}
			\caption{
				Plugs and sockets drawn on color band levels.
				Sockets are placed in cells.
				(a) color plug,
				(b) high plug,
				(c) pass-through plug,
				(d) choice socket,
				(e) color socket,
				(f) pass-through socket.
			}%
			\label{fig:w_hard_gadgets_alone}
		\end{figure}
		
		\myproofparagraph{Refinement.}
		Next, we make the gadgets proper by finding edges that
		span more than one level and subdividing them by vertices on all
		the intermediate levels, setting the x-coordinates of the new vertices according to the unique drawings of the gadgets.
		By \Cref{obs:proper_is_rigid} the construction is rigid and by \Cref{obs:subdivisions} it will remain rigid even if we insert further levels.
		This completes the part of the construction on rigid levels.
		
		Now we add the other types of levels, which host the plugs.
		To each color band, where we already have $7$ rigid levels (for now denoted by
		\texttt{R}), we insert $4$~color levels (\texttt{C}), $4$~high
		levels (\texttt{H}), and $2$~pass-through levels (\texttt{P}) so
		that the final order from bottom to top is
		\texttt{HCRPRHRCRCRHRPRCH};
		see \cref{tab:reserved_levels}.
		Then we subdivide the leftmost and the rightmost wall to contain vertices on all levels -- forcing everything else to be drawn between them.
		We refer to the levels by their relative order within the band.
		Note that the level insertion changed the indices of rigid levels and of the gadgets they contain.
		\begin{table}
			\centering
			\begin{tabular}{l@{\qquad}ccccccccccccccccc}
				\toprule
				level & 1 & 2 & 3 & 4 & 5 & 6 & 7 & 8 & 9 & 10 & 11 & 12 & 13 & 14 & 15 & 16 & 17 \\
				type  & \texttt{H} & \texttt{C} & \texttt{R} & \texttt{P} & \texttt{R} & \texttt{H} & \texttt{R} & \texttt{C} & \texttt{R} & \texttt{C}  & \texttt{R}  & \texttt{H}  & \texttt{R}  & \texttt{P}  & \texttt{R}  & \texttt{C} & \texttt{H} \\
				\bottomrule
			\end{tabular}
			\medskip
			
			\caption{
				Order of the levels within each color band.
				\texttt{H} stands for high levels
				hosting the high plugs,
				\texttt{C} for color levels
				hosting the color plugs,
				\texttt{R} for rigid levels
				hosting the rigid structure, and
				\texttt{P} for pass-through levels.
			}%
			\label{tab:reserved_levels}
		\end{table}
		
		\myproofparagraph{Placement of plugs.}
		To each of the $k$ color bands,
		we now add a set of plugs that will guide the mechanism for selecting vertices in the MCIS instance $\mathcal S$.
		None of these plugs are connected to the rest of the graph, i.e., each plug is an isolated connected component of $\G$.
		In the color-$j$ band, we add as many plugs
		as there are choice and color sockets in total
		(so we do not include pass-through sockets here):
		Recall that the color-$j$ band contains in total
		$(2n'-2)+m_j \cdot (2n'-1)$ choice and color sockets.
		We add $n'-1$ of $(1,6,12,17)$-plugs, which we call \emph{high plugs}, and $n'-1 + m_j \cdot (2n'-1)$ of $(2,8,10,16)$-plugs, which we call \emph{color plugs}.
		See \cref{fig:w_hard_gadgets_alone}a and \ref{fig:w_hard_gadgets_alone}b.
		We place the high plugs from left to right one after another so that their x-coordinates do not overlap, similarly we place the color plugs.
		
		\myproofparagraph{Shifting mechanism (choosing vertices).}
		Observe that color plugs can fit into choice sockets and color
		sockets while high plugs can fit only into choice sockets.
		Further, neither color plugs nor high plugs can fit into pass-through sockets.
		The leftmost and rightmost walls of $\G$ force the plugs to be drawn
		into the cells while a high plug and a color plug cannot
		occupy the same socket due to \cref{lem:two_plugs_one_socket}.
		Hence, in every realization of $\G$ every choice and color socket is linked with exactly one plug.
		When we say that we \emph{shift color $j$ by $i$ steps},
		we mean that a realization of $\G$ links $i$ high plugs in the left choice block
		and $(n'-1)-i$ high plugs in the right choice block
		of the color-$j$ band.
		The high plugs can be interspersed among color plugs in the choice blocks arbitrarily.
		Shifting color $j$ by $i$ represents \emph{selecting
			the $(i+1)$-th vertex of color $j$};
		see \cref{fig:w_hard_shifting}.
		As we have $n'-1$ choice sockets
		in the left choice block,
		the represented vertex choice is in~$[n']$.
		Let $A$ be a color-$j$ plug that links to the $t$-th
		color socket of an edge block when color $j$ is shifted by 0 steps.
		Recall that the relative order of the color plugs is fixed
		because they occupy the same set of levels.
		Hence, with shift $i$ plug $A$ links to the $(t + i)$-th
		color socket of the edge block.
		Note that shift $i$ forces $i$ of the last color-$j$ plugs in the
		edge block overflow to the next block.

		\begin{figure}[tb]
			\centering
			\includegraphics[page=1]{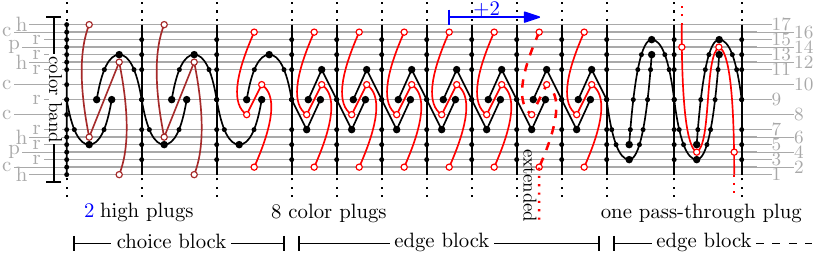}
			\caption{
				An example of a color band in our \Whness reduction for color
				classes of size~4 with high-plugs (h), color-plugs (c),
				pass-through-plugs~(p), and the rigid structure (r).
				Note how two high-plugs within the choice-block force all the color-plugs to be shifted by
				two columns to the right.
				Further note the special role of the dashed color-plug
				that extends to the collision-gadget beyond these levels.
			}%
			\label{fig:w_hard_shifting}
		\end{figure}
		
		\myproofparagraph{Collisions (ensuring independence).}
		Next, we make sure that for each $uv\in E(H)$,
		the vertices $u$ and $v$ cannot be selected at the same time.
		Intuitively, we achieve this by extending two color plugs
		within the edge-$uv$ block so that they end up occupying
		the same collision socket, where at most one can fit~--
		making the shifts that represents selecting $u$ and $v$
		impossible to draw.
		
		Recall that to the edge-$uv$ block we already placed a collision socket
		into $n'$-th cell of collision-$q_{uv}$ band, where $q_{uv}=\min\{c(u),c(v)\}$.
		
		Now to each collision band (whose three rigid levels are
		denoted by \texttt{R}), we insert two types of levels~--
		one called \emph{above levels} (\texttt{A})
		and the other called \emph{below levels} (\texttt{B}).
		The resulting sequence of levels within each collision band is \texttt{ARABRBARB}, also noted in \Cref{tab:collision_levels}.
		We add $(1,3,7,7)$-plugs called \emph{A-collision plugs}
		and $(4,4,6,9)$-plugs called \emph{B-collision plugs}.
		We set their x-coordinates to be in the middle cell but due to the choice of levels the collision plugs are not restricted to be drawn in a particular cell.
		Note that with the inserted levels,
		each collision socket is now a $(2,5,5,8)$-socket.
		
		\begin{table}[tb]
			\centering
			\begin{tabular}{l@{\qquad}ccccccccc}
				\toprule
				level & 1 & 2 & 3 & 4 & 5 & 6 & 7 & 8 & 9 \\
				type  & \texttt{A} & \texttt{R} & \texttt{A} & \texttt{B} 
				& \texttt{R} & \texttt{B} & \texttt{A} & \texttt{R} & \texttt{B}\\
				\bottomrule
			\end{tabular}
			\medskip
			
			\caption{Order of the levels within each collision band.
				\texttt{A} stands for above levels
				hosting the A-collision plugs,
				\texttt{B} for below levels
				hosting the B-collision plugs, and
				\texttt{R} for rigid levels
				hosting the rigid structure.%
			}
			\label{tab:collision_levels}
		\end{table}
		
		For this description let $c(u) < c(v)$, without loss of generality by swapping $u$ and~$v$.
		Let $\idx(u)$ and $\idx(v)$ denote the \emph{index} of
		the vertices $u$ and $v$
		within their color, i.e.,
		vertex~$u$ is the $\idx(u)$-th vertex of color $c(u)$
		and $v$ the $\idx(v)$-th vertex of color~$c(v)$.
		Let~$X$ be the $(n' - \idx(u) + 1)$-th color plug in the color-$c(u)$ band of the edge-$uv$ block and
		let~$Y$ be the $(n' - \idx(v) + 1)$-th color plug in the color-$c(v)$ band of the edge-$uv$ block.

		Finally, we describe how to connect $X$ and $Y$ to the collision plugs.
		We connect $X$ to the B-collision plug, this is straight-forward as they
		occupy neighboring bands.
		However, to connect $Y$ with the A-collision plug the connection needs to traverse over many color bands.
		To accommodate that, we insert $(4,4,14,14)$-plugs called
		\emph{pass-through plugs}, see \cref{fig:w_hard_gadgets_alone}c, into each color band
		between the collision-$c(u)$ band and the color-$c(v)$ band.
		We place the pass-through plugs in the middle cell of the edge-$uv$ block.
		We connect the A-collision plug in the collision-$c(u)$ band with
		the pass-through plug in the color-$(c(u)+1)$ band, then connect
		pass-through plug in color-$i$ band with the pass-through plug in color-$(i+1)$
		band for all $i\in\{c(u)+1,\dots,c(v)-2\}$, and last we connect
		the pass-through plug in color-$(c(v)-1)$ band with $Y$.
		We refer to the two connected components of $X$ and $Y$
		as the (pair of) \emph{extended color plugs} of $uv$.
		Note that the pass-through plugs fit only into pass-through sockets.
		
		\begin{claim}\label{claim:plugs_final_column}
			Let $uv \in E(H)$.
			If $u$ is selected in color $c(u)$
			and $v$ is selected in color $c(v)$,
			the pair of extended color plugs
			in the edge-$uv$ block lies in the
			$n'$-th column of the block.
		\end{claim}
		\begin{claimproof}
			Since color $c(u)$ is shifted by $\idx(u) - 1$
			steps to the right,
			the corresponding extended color plug is in column
			$n'+1-\idx(u) + (\idx(u)-1) = n'$ of the edge-$uv$ block.
			Similarly, for the corresponding extended color plug
			on the color-$c(v)$ band shifted by $\idx(v)-1$,
			we have $n'+1-\idx(v) + (\idx(v)-1) = n'$.
		\end{claimproof}
		
		We finish the construction of $\G$ by performing the described way of adding
		a pair of extended color plugs for each edge of $H$.
		
		\myproofparagraph{Correctness.}
		We started off by adding a rigid structure that divides
		the ordered level planar drawing into cells.
		Then we inserted extra levels that host various plugs and allow the
		plugs to flow freely between cells.
		The bounding walls occupy all levels and so restrict every plug to be
		drawn between them.
		We added sockets to cells that restrict the number of plugs in a cell to one.
		We modelled selection of vertices in MCIS by shifting a set of color plugs
		within the color band.
		Then we ensured that the selected vertices
		must be independent by creating collision via extended color plugs.
		We finish the proof by arguing correctness.
		
		\begin{claim}\label{claim:no_drawing_one_edge}
			For $uv \in E(H)$, there is no ordered level planar drawing of $\mathcal{G}$ where the color-$c(u)$ band is shifted by $\idx(u)-1$ and the color-$c(v)$ band is shifted by $\idx(v)-1$.
		\end{claim}
		\begin{claimproof}
			Towards a contradiction, assume that the color plugs in the color-$c(u)$ band
			are shifted by $\idx(u)-1$ and the color plugs
			in the color-$c(v)$ band are shifted by $\idx(v)-1$.
			Then by \Cref{claim:plugs_final_column} the extended plugs of $uv$ end up in the $n'$-th column -- the same column that has a collision socket of edge-$uv$ block in the collision-$c(u)$ band.
			The two extended color plugs form connected components
			to the collision-$c(u)$ band so in
			an ordered level planar drawing the extended parts must
			stay in the same column as they cannot cross rigid walls.
			We also see that A-collision plug and B-collision plug
			cannot both be drawn in the same collision socket
			due to \cref{lem:two_plugs_one_socket}.
			Hence, such a drawing does not exist.
		\end{claimproof}

		\begin{claim}\label{claim:everything_works}
			There exists an ordered level planar drawing of $\G$ if and only if $\mathcal S$ is a yes-instance.
		\end{claim}
		\begin{claimproof}
			For the one direction, assume $\mathcal S$ is a yes-instance.
			Let $s_j$ be the index of a vertex of color $j$ that is chosen in the solution $\mathcal S$.
			We draw the rigid part of $\G$ in its unique embedding.
			Then we draw high and color plugs of color-$j$ band shifted by $s_j-1$.
			All types of plugs occupy disjoint sets of levels, hence,
			drawing one type does not influence the possibility
			of drawing the other.
			In $\mathcal S$ we have, for each edge, chosen at most
			one incident vertex.
			So in each edge-$uv$ block at most one extended plug links to the
			collision socket of $uv$.
			
			In the other direction, assume we have an ordered level planar drawing of $\G$.
			Focusing on the rigid part we note it still has only one embedding, which must be used in this drawing as well.
			This naturally divides the final drawing into columns and bands.
			The leftmost and rightmost walls force all the remaining gadgets to be drawn in between them, placing each plug in some column.
			For each color $j$ we can identify the number $s_j$ of high plugs in the color-$j$ band left choice block.
			Due to \Cref{claim:no_drawing_one_edge}, we know that these
			shifts cannot represent adjacent vertices in $\mathcal S$
			and, hence, correspond to a solution for~$\mathcal S$.
		\end{claimproof}
		
		MCIS is \Wh with respect to the solution size $k$.
		Our construction uses $17 \cdot k + 9 \cdot (k-1) = 26k-9$ levels and can clearly be carried out in polynomial time.
		Hence, \OLPlong is \Wh with respect to the number of levels.
	\end{proof}

	\paragraph{\XNLPhness.}
	We now extend our construction for \Whness to obtain \XNLPhness as well.
	We design a \emph{parameterized tractable log-space
		reduction}~\cite[Section~V.B.]{XNLP2021}~--
	a parameterized reduction that runs in $\Oh(g(k) \cdot
	n^c)$ time and uses only $\Oh(f(k) \cdot \log n)$ space
	for internal computation, where $g$ and $f$ are computable functions,
	$k$ is the parameter, and $n$ is the input size.
	It suffices to design the reduction such that it returns
	a single bit of the output on demand.
	Observe that with such a reduction,
	we can retrieve the entire output by requesting one bit at a time.
	We reduce from \textsc{Chained Multicolored Independent Set} (CMCIS)~--
	an \XNLP-complete problem defined by Bodlaender et al.~\cite{XNLP2021},
	which is a sequential version of MCIS.
	As in MCIS, we are given a $k$-colored graph~$H$ with
	color classes $C_1, \dots, C_k$ and, additionally,
	there is an $r$-partition $V_1, \dots, V_r$ of $V(H)$
	such that for every $vw \in E(H)$,
	if $v \in V_i$ and $w \in V_j$, then $|i-j| \le 1$.
	The task is to select an independent set $X \subseteq V(G)$
	such that, for each $i \in [r]$ and for each color $j \in [k]$,
	$|X \cap V_i \cap C_j| = 1$.
	
	\begin{theorem}\label{thm:olp_xnlp_hard}
		\OLPlong is \XNLPh with respect to the height.
	\end{theorem}
	\begin{proof}
		Let $\mathcal S$ be an instance of the \textsc{Chained Multicolored Independent Set} (CMCIS).
		We prove this theorem by constructing an \OLP instance $\G$, with $h \in \Theta(k)$ levels, that has an ordered level planar drawing if and only if
		there exists a solution to $\mathcal S$.
		Then we argue how to retrieve a single output bit using limited memory.
		
		The instance we build uses the structures that
		we have described in detail in the proof of \cref{thm:olp_w_hard},
		namely (different types of) plugs, sockets, color/collision bands, the whole mechanism that represents the choice of a vertex in the solution
		by the shifts of plugs within a color band, and the collision mechanism that makes it impossible to use shifts that would represent choosing neighboring vertices.
		We go through the construction while assuming familiarity with the aforementioned terms.
		
		For $i \in [r]$, we denote the edges between vertices of $V_i$ by $E_i$,
		and we denote the edges that connect a vertex from~$V_i$
		with a vertex from~$V_{i+1}$ by $E_{i,i+1}$.
		Moreover, we let $n_i = |V_i|$, $m_i = |E_i|$, and $m_{i,i+1} = |E_{i,i+1}|$.
		For the reduction, we may assume that all
		of the~$r$ partitions have the same size $n_i$,
		for each $V_i$ there are $n_i / k =: \hat{n}$ vertices per color,
		and each $E_i$ contains at least one edge, i.e., $m_i \ge 1$.
		If this was not the case, we could simply add vertices
		(with a specific color~$j$) until all partitions have $n_i$ vertices
		and $\hat{n}$ vertices per color.
		The new vertices get an edge to every vertex in its partition
		except for the vertices of color~$j$.
		Hence, they cannot appear in the solution if there are two or more colors.
		The number of vertices being added is in $\Oh(k \cdot r \cdot \max_{i \in r} n_i)$.
		
		Roughly speaking, the whole construction combines multiple times the \Whness constructions~--
		partially overlapping horizontally.
		More precisely, we first create $2k$ color bands that alternate with $2k-1$ collision bands.
		As before, each color band starts with $7$ rigid levels and each collision band starts with $3$ rigid levels.
		We divide the color bands into two \emph{halves}; we call the $k$ color bands that occupy
		the lower levels \emph{bottom} half and the other~$k$ color bands \emph{top} half;
		see \cref{fig:xnlp_example} for an illustration.
		
		Then, we add $1 + 2 \cdot (\hat{n}-1) + m \cdot (2 \hat{n}-1)$ walls where $m$ is the total number of edges, i.e., $m = \sum_{i=1}^r m_i + \sum_{i=1}^{r-1} m_{i,i+1}$.
		The walls divide the final drawing into columns, each column in a band forms a cell.
		Several columns form a block, however, we will also use blocks that span only one half.
		In such a case, the construction fills all unused cells with pass-through sockets.
		The first and last $(\hat{n}-1)$ columns form the first and the last choice block.
		The remaining columns are divided into $m$ edge blocks each with $(2 \hat{n}-1)$ columns.
		These edge blocks are partitioned into \emph{sectors}.
		Going from left to right, we call the first $m_1$ edge blocks the $(1)$-sector,
		which represent the edges in $E_1$,
		the next $m_{1,2}$ blocks the $(1,2)$-sector,
		which represent the edges in $E_{1,2}$,
		and so on, alternating between $E_i$ and $E_{i,i+1}$ until $E_r$ is reached.
		
		For $i \in [r]$, let $V_i$ be assigned to the top half if $i$ is odd, and to the bottom half otherwise.
		The colors of the vertices in $V_i$ are represented by the $k$ color bands in its assigned half.
		We now add all the remaining parts.
		For each edge block, we add color, pass-through, and collision sockets accordingly
		as in the proof of \cref{thm:olp_w_hard}.
		For edges in $E_{i,i+1}$, we put the color sockets in one of the top-half color bands
		and in one of the bottom-half color bands.
		Note that an edge block representing an edge that is incident to a vertex from~$V_i$
		lies in the $(i-1,i)$-, $(i)$-, or $(i,i+1)$-sectors.
		For each $i \in [r]$, we place two choice blocks in the half where $V_i$ is assigned to:
		one in the $\hat{n}-1$ columns just before the $(i-1,i)$-sector and one just after the $(i,i+1)$-sector.
		Last, we insert the non-rigid levels into each color and collision band as before.
		We subdivide the very first and last wall to have vertices on all levels.
		For each $i \in [r]$, we also subdivide the first wall of the left choice block of $V_i$
		in the half of $V_i$ and do the same for the last wall of the right choice block.
		
		The resulting construction has $r$ times (some instance of) the \Whness construction
		lying in the top or the bottom half,
		where the only difference is that
		some of the extended color plugs are shared between
		two of these constructions in different halves.
		Note that this does not influence the shifting mechanism; it still works the same way.
		Each edge block internally looks the same, the ones for $E_i$
		just span fewer levels than those for $E_{i,i+1}$.
		Hence, the edge block arguments from \cref{claim:plugs_final_column,claim:no_drawing_one_edge}
		work the same and we can derive a statement analogous to \cref{claim:everything_works}.
		
		Recall that for a parameterized tractable log-space	reduction,
		we need to specify how to retrieve from the input CMCIS instance
		in \FPT-time a single bit of the output \OLP instance
		while maintaining only $\Oh(f(k) \cdot \log n)$ space
		across all requests of a bit.
		We traverse our grid structure column by column.
		To locate the plugs somewhere (for the reduction, we just need
		to represent the \OLP instance not a specific drawing of it),
		we assume that all high plugs are assigned to the right choice block.
		This gives every color, pass-through, and collision plug its fixed column.
		Therefore, we need to save as our status only
		the information belonging to the current column,
		e.g., to which edge (if any) this column belongs etc.
		Every column spans $\Oh(k)$ levels and every socket has $\Oh(1)$ vertices and edges.
		Hence, the total number of vertices and edges in one column is $\Oh(k)$
		and what type they are can be computed in polynomial time and $\Oh(k \log n)$ space
		where the $\log n$ factor comes from saving the number of
		the currently considered column, vertex and edge.
		When we need to proceed with the next column,
		we can find the next edge in polynomial time
		if needed, which means we do not need to first
		sort all edges.
		Moreover, if our output encoding expects first
		all vertices, then all edges, etc.,
		our algorithm might do several iterations
		through our column structure
		when considering all requests of single bits together.
	\end{proof}
	
	\begin{figure}[tb]
		\centering
		\includegraphics[page=1]{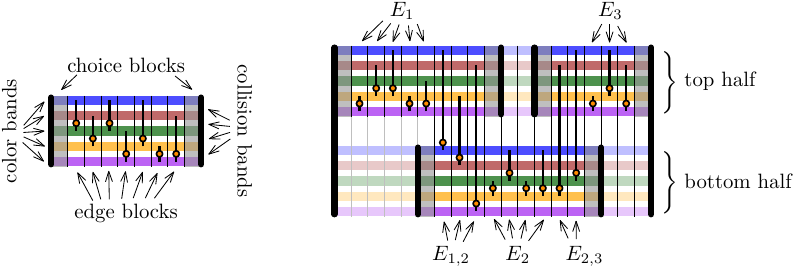}
		\caption{Comparison of the \Whness construction (left) and
			the \XNLPhness construction (right) with $r=3$ and $k=5$.
			The thick walls terminating the choice blocks have vertices
			on all levels.  Vertical bars and orange disks between color
			bands represent the extended plugs and the collision
			sockets, respectively, which are used for the collision
			mechanism.
		}%
		\label{fig:xnlp_example}
	\end{figure}

	The construction in the proof of \cref{thm:olp_xnlp_hard} can be
	altered to make the graph connected
	as we show now.
	
	\begin{theorem}\label{thm:olp_connected_xnlp_hard}
		\OLPlong is \XNLPh with respect to the height, even when the input graph $G$ is connected.
	\end{theorem}
	
	\begin{proof}
		Recall that in the proof of \cref{thm:olp_xnlp_hard}
		(and \cref{thm:olp_w_hard}),
		we first create many walls, and then fill in sockets and plugs.
		The only gadgets that are disconnected from the rigid structure are the plugs.
		In short, our goal is to extend all plugs to connect them
		to a single vertex that is placed at a new topmost level.
		To this end, we add a new \emph{tunnel column} to every column.
		A tunnel column contains sockets that are not much different from pass-through sockets, though they allow up to $3k$
		plugs to fit at the same time.
		These sockets are placed in new bands,
		which are inserted inside the color bands.
		The new bands have their own levels for the different types
		of plugs because they need to be drawn mutually independent.
		Therefore, we use roughly $6k$ extra levels per color band,
		resulting in a total number of levels in $\Theta(k^2)$.
		
		We now proceed to show the construction in more detail.
		We begin with a slightly altered construction from the proof of \cref{thm:olp_xnlp_hard}.
		We create the rigid levels and add the walls that form the columns.
		Then we split each column by a wall into two vertical strips~--
		we refer to the right strip as \emph{tunnel},
		the left one we still call column.
		We continue with the construction from the proof of
		\cref{thm:olp_xnlp_hard} as if the tunnels did not exist.
		After this construction is complete, we continue as follows.
		Observe that each connected component of the rigid structure
		contains a wall because every connected component of every
		socket is attached to a wall.
		We add a new topmost level and create a new vertex $q$ on it.
		We connect $q$ to all topmost vertices of all the walls making the rigid structure connected, see \cref{fig:hard_connected}d.
		
		\begin{figure}[tb]
			\centering
			\includegraphics[page=1]{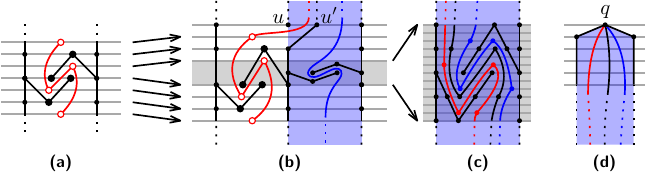}
			\caption{
				Our modifications to make the ordered level graph
				in our \XNLPhness reduction connected.
				Tunnels have a transparent blue background,
				while the tunnel band has a gray background.
				(a) Simplified color band with a color plug in a color socket;
				(b) column divided into the left column and a tunnel, showing how a plug continues in the tunnel;
				(c) a tunnel socket and three tunnel plugs within the tunnel band of a tunnel;
				and (d) the top level where the tunnel ends by connecting to $q$.
			}%
			\label{fig:hard_connected}
		\end{figure}
		
		Into each color band, we insert a \emph{tunnel band} between
		the levels $9$ and $10$ (in the relative level numbering
		of the color bands), see \cref{fig:hard_connected}a and b.
		A tunnel band consists of $6k+1$ levels,
		which are grouped in a top and a bottom half.
		Each half contains $k$ triplets of levels.
		The $j$-th triplet is made of the rigid-$j$ level,
		color-$j$ level, and high-$j$ level,
		as shown in \cref{fig:tunnel_levels}.
		To close off the tunnel band, it contains an extra rigid level on top.
		
		\begin{table}[tb]
			\caption{Purposes of the levels of a tunnel band:
				$\texttt{R}_j$ stands for rigid-$j$ level, $\texttt{C}_j$ for color-$j$ level, and $\texttt{H}_j$ for high-$j$ level.
				Furthermore, \texttt{R} stands for the rigid level on top.}
			\label{fig:tunnel_levels}
			\centering
			\begin{tabular}{l@{\qquad}ccccccccccccccc}
				\toprule
				level & $1$ & $2$ & $3$ & $4$ & $5$ & $\dots$ & $3k-1$ & $3k$ & $3k+1$ & $3k+2$ & $\dots$ & $6k$ & $6k+1$ \\
				type & $\texttt{R}_1$ & $\texttt{C}_1$ & $\texttt{H}_1$ & $\texttt{R}_2$ & $\texttt{C}_2$ & $\dots$ & $\texttt{C}_k$ & $\texttt{H}_k$ & $\texttt{R}_1$ & $\texttt{C}_1$ & $\dots$ & $\texttt{H}_k$ & \texttt{R}\\
				\bottomrule
			\end{tabular}
		\end{table}
		
		Let a $(1,3k+1,3k+1,6k+1)$-socket with respect to a tunnel band
		be called \emph{tunnel socket}.
		Let an $(i,i,i+3k,i+3k)$-plug be called \emph{tunnel-$i$ plug} for $i \in \{2,\dots,3k\}$.
		By \cref{lem:two_plugs_one_socket}, any number of distinct tunnel plugs
		can fit into a tunnel socket, moreover,
		they occupy mutually disjoint sets of levels.
		
		For a vertex $v$ and $i \in \{2,\dots,3k\}$,
		let \emph{tunnel from $v$ to $q$ through $i$} mean adding
		a tunnel-$i$ plug to each color band between
		the color band of~$v$ and the topmost vertex~$q$,
		connecting $v$ to the bottommost of these tunnel-$i$ plugs,
		connecting pairs of tunnel-$i$ plugs in neighboring color bands,
		and connecting the topmost plug to $q$.
		In case $v$ is in the topmost color band, we just connect $v$ with $q$ with an edge.
		Our aim is to make the construction connected by tunneling
		from a vertex of every connected component to $q$.
		The remaining connected components
		are the plugs and extended plugs described
		in the proofs of \cref{thm:olp_w_hard,thm:olp-XNLP-complete}.
		
		Similar to the closing argument of \cref{thm:olp_xnlp_hard}, let us assume for the ease of description that all high plugs are placed in the right choice block.
		Let us focus on a single column-tunnel pair along with their plugs to describe the procedure that is performed on each such pair.
		In each color-$j$ band for $j \in [k-1]$ we take the first rigid vertex $u$ above the color-$j$ band in the wall between the column and its tunnel, we duplicate $u$ to create $u'$, place $u'$ to the right of $u$ and set the edge of the wall incident to $u$ from below to be connected to $u'$ instead.
		This operation splits the wall into two parts and the lower part becomes disconnected from $q$; we fix this by tunnelling from $u'$ to $q$ through $3j-2$.
		In each color-$j$ band for $j \in [k]$ within this column-tunnel pair, if the socket is assigned a high or a color plug, then we take the top vertex $v$ of the plug assigned the socket and tunnel from $v$ to $q$ through $3j-1$ if the plug is a color plug, or through $3j$ if it is a high plug.
		
		Now we argue that our alterations did not impact how the construction works.
		First, each tunnel band was inserted between two adjacent levels
		of a color band so it is evident that no color or high plugs
		can fit the tunnel socket,
		and no tunnel plugs can fit the choice or color sockets.
		We let the connections of the tunnel-$(3i-2)$ plugs,
		which we use to attach vertices of the walls to $q$,
		occupy the rigid levels so their positions within the rigid structure is fixed.
		Similarly, we let the connections of the tunnel-$(3i-1)$ plugs
		occupy levels that are reserved for color plugs,
		and we let the connections of the tunnel-$(3i)$ plugs occupy levels 
		that are used only for high plugs.
		Hence, the set of feasible drawing does not change
		by adding tunnels and this connected instance of \OLPlong
		has a planar drawing if and only if the \textsc{CMCIS}
		instance has a solution.
		
		We argue that to retrieve a single output bit,
		we can still use the same approach as in the proof of \cref{thm:olp_xnlp_hard}.
		We again traverse our construction with
		a vertical sweep-line.
		For the argument, we again assume that all high plugs
		are located in the right choice blocks.
		This fixes not only color, pass-through, and collision plugs, but also the newly added tunnel plugs.
		The cornerstone observation is that even though we altered the construction,
		we know that in each column-tunnel pair,
		every level contains a constant number of vertices.
		Every column-tunnel pair spans $\Oh(k^2)$ levels that
		in total contain no more than $\Oh(k^2)$ vertices and edges.
		For a global orientation of the column-tunnel,
		$\Oh(n)$ bits suffices.
		To get one bit of the output, we can go over the columns-tunnel pairs
		one by one.
	\end{proof}
	
	\begin{remark}\label{rmk:connected-many-levels}
		We now hint at the reason why the bound of $\Theta(k^2)$ levels
		might be necessary to make our construction connected.
		At first glance, it may seem feasible to restrict the number of levels in each tunnel socket to two.
		Every edge block contains at most two plugs in each column so why do we need $\Omega(k)$ levels in each tunnel socket?
		The choice blocks do contain many plugs in each column, but those could be spread out into separate choice blocks, each being relevant for only one color.
		The main issue is that even for the lowest numbered color we need to connect its plugs to the topmost vertex $q$.
		To do that we build a path over all other color bands.
		If we did not add extra levels to such bands, then choices for the lowest color would not have been independent of the upper colors.
	\end{remark}

	\section{NP-Hardness of 4-Level \CLPlong}
	\label{sec:clp-hardness}
	
	We show the hardness of 4-level \CLP by reducing from the
	\textsc{$3$-Partition} problem.
	In the \textsc{$3$-Partition} problem,
	we are given a multiset $S = \{s_1, \dots, s_n\}$
	of $n=3m$ positive integers whose sum is~$mB$.
	The task is to decide whether there exist a partition of $S$
	into $m$ triplets $S_1,S_2,\dots,S_m$ such that
	the sum of the numbers in each triplet is equal to $B$.
	\textsc{$3$-Partition} is known to be strongly NP-hard even when every integer in $S$ is strictly between $B/4$ and $B/2$ \cite{GJ75}.
	This is important for us because we represent each number~$k$ in~$S$ in
	unary encoding by a $k$-clip, a graph structure that we define next;
	see \cref{fig:clip}.
	The bounds on these numbers also play a role in ensuring that each triplet
	contains exactly three numbers.
	
	\begin{definition}[$k$-Clip]\label{def:clip}
		A \emph{$k$-clip} is graph that contains $2k+1$ edges that have one
		endpoint on level $2$ and alternate their second endpoint between
		levels $3$ and $1$ (starting with level $3$); the clip also contains
		a central vertex
		on level $4$ which is connected to all vertices on level~$3$.  The
		order of the level-$2$ vertices is fixed to a linear order by
		constraints.
	\end{definition}
	
	\begin{figure}[tb]
		\centering
		\includegraphics[page=1]{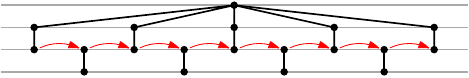}
		\caption{
			An example of a $4$-clip that represents number $4$.
			Red arrows show second-level constraints that fix the drawing of the clip.
		}%
		\label{fig:clip}
	\end{figure}
	
	\begin{lemma}\label{lem:disjoint_clips}
		If $\G$ contains two clips $P_1$ and $P_2$, then, in any $4$-level
		planar
		drawing $\Gamma$ of~$\G$, the clips $P_1$ and $P_2$ do
		not cross, i.e., on every level either all vertices of $P_1$ lie
		before all vertices of $P_2$ or vice versa.
	\end{lemma}
	\begin{proof}
		First, consider the level planar drawing $\Gamma$ and let $v_1 \in P_1$ and $v_2 \in P_2$ be the unique level-$4$ vertices of these clips.
		We assume without loss of generality (by swapping names of $P_1$ and $P_2$) that $v_1$ is drawn to the left of $v_2$.
		We immediately see that on the levels $2$ and $3$ all vertices of the clip $P_1$ must be drawn before all vertices of the clip $P_2$.
		For the level-$1$ vertices of a clip it holds that they must be drawn in the same order as their adjacent vertices on level $2$.
		Therefore, the statement also follows for the level-$1$ vertices.
	\end{proof}
	
	For each $j \in [3m]$, the number $s_j \in S$ shall have its own $s_j$-clip and be forced to be drawn in a mountain chain, shown in \Cref{fig:mountanWhole}.
	Each mountain chain represents one of the final $m$ ``buckets'' of size $B$.
	An assignment of $k$-clips onto these mountain chains represents
	a partition of $S$ into triplets.
	
	\begin{definition}[Mountain Chain]\label{def:mountainchain}
		A \emph{mountain} is a path on $5$ vertices $(v_1,v_2,v_3,v_4,v_5)$
		where $\gamma(v_1) = 1$, $\gamma(v_2) = 2$, $\gamma(v_3) = 3$,
		$\gamma(v_4) = 2$, and $\gamma(v_5) = 1$.
		
		A \emph{mountain $k$-chain} is created by combining $k$ mountains
		in a chain by identifying, for each pair of consecutive mountains,
		$v_5$ of the first one with $v_1$ of the second one.
		After forming the chain, we add two walls that consist of a path from level $1$ to level $4$.
		The level-$1$ vertices of these walls are identified with the leftmost and rightmost level-$1$ vertices of the mountain chain.
		We can have sequences of mountain chains if we identify
		the second and the first wall of two mountain chains.
	\end{definition}
	We observe that each mountain $k$-chain contains $k+1$ \emph{valleys}, one between each pair of adjacent mountains, and another between the outer mountains and the adjacent walls.
	
	\begin{figure}[tb]
		\centering
		\includegraphics[page=2]{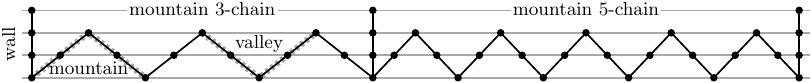}
		\caption{
			A drawing of mountains in a sequence of two mountain chains.
			The mountain chains are made up of mountains glued together
			by their first-level vertices and with a wall from
			level $1$ to level $4$ on both ends.
			Observe that the valleys are formed between adjacent mountains
			and between the left-, respectively, rightmost
			mountain and the neighboring walls.
		}%
		\label{fig:mountanWhole}
	\end{figure}
	
	Having all basic building blocks at hand,
	we are ready to prove the hardness of \CLP.
	
	\begin{theorem}
		\CLPlong is \NPh for four levels.
		\label{thm:CLP-NPhard}
	\end{theorem}
	
	\begin{proof}
		We describe a polynomial-time reduction from \textsc{$3$-Partition}.
		We are given integers~$m$, $B$ and a multiset of integers $S$ of $n=3m$ positive integers such that their sum is~$mB$.
		We create the \CLP instance $\G = (G, \gamma, (\prec_i)_{i\in [4]})$ as follows.
		For each $s_j \in S$ where~$j \in [3m]$, create an $s_j$-clip.
		We also create a sequence of $m$ mountain $B$-chains $C_1,\dots,C_m$.
		Let $w_0,\dots,w_m$ denote (in order)
		the vertices of the walls on level~4.
		We add the constraint $w_0 \prec_4 w_m$ and, for each $s_j$-clip,
		we add the constraints $w_0 \prec_4 x$ and $x \prec_4 w_m$,
		where $x$ is the level-4 vertex of the $s_j$-clip.
		
		First, let us prove that, given a Yes-instance of
		\textsc{$3$-Partition}, the constrained level planar
		obtained by our reduction admits a level planar drawing on
		four levels.
		We draw the mountain chain naturally from left to right in the
		ordering which satisfies $w_0 \prec_4 w_m$.
		Consider a solution of the \textsc{$3$-Partition} instance.
		For each triplet $(a,b,c)$ in that solution,
		we take the respective $a$-clip, $b$-clip, and $c$-clip,
		and draw them all within a single mountain $B$-chain.
		In particular, we take one clip and draw it such that
		the edges between the levels~$1$ and~$2$ occupy the spaces below
		a set of consecutive mountains, and the edges between
		the levels~$2$ and~$3$ occupy the valleys around these mountains.
		This way the $a$-clip fills $a$ of the mountains,
		the $b$-clip fills $b$ of them, and $c$-clip fills $c$ mountains.
		Together, their edges on the levels~$1$ and~$2$ occupy
		exactly the $B$ mountains of a mountain $B$-chain as $a+b+c=B$.
		See \cref{fig:mountainClipConsB} for a simplified example of such a drawing.
		The $m$ triplets fit into the $m$ mountain $B$-chains as described above.
		This drawing is clearly planar and satisfies all the constraints, concluding the first part of the proof.
		\begin{figure}
			\centering
			\includegraphics[page=3]{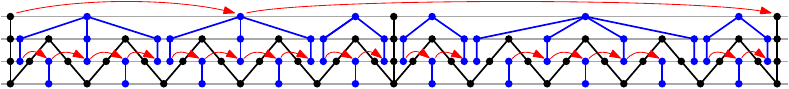}
			\caption{
				A level planar drawing of clips of size $2,2,1,1,3,1$
				into mountain chains of size $5$.
				For clarity, we show the constraints that force clips
				between first and last wall only for a single clip.
			}%
			\label{fig:mountainClipConsB}
		\end{figure}
		
		Second, we prove a series of claims showing that if there is
		a level planar drawing on four levels,
		then it must look exactly like we described above.
		Hence, such a planar drawable constrained level graph
		represents a Yes-instance of \textsc{$3$-Partition}.
		We start with the drawing of the mountain chains.
		
		\begin{claim}
			There is a unique level planar drawing of the mountain $B$-chains $C_1,\dots,C_m$.
		\end{claim}
		\begin{claimproof}
			The mountains of the mountain chains together
			form of a path of length $m \cdot 4B + 1$,
			which are arranged on the levels in order
			$1, 2, 3, 2, 1, \dots, 2, 1$ and that can be drawn in two ways~--
			left-to-right or right-to-left.
			Then, as the walls go up to level $4$,
			they can only be drawn on the upper side of that long path,
			so there is a unique space for each wall to be drawn.
			Finally, due to the $w_0 \prec_4 w_m$, only one of
			the two ways to draw the path and the walls remains.
		\end{claimproof}
		
		\begin{claim}\label{clm:single_edge}
			In any level planar drawing,
			any mountain contains at most one
			edge (between levels~1 and~2) of any clip.
		\end{claim}
		\begin{claimproof}
			There is a unique way to draw a clip because
			the constraints give a total order on its level-$2$ vertices.
			Notice that in this drawing every edge~$e$ between the level~1
			and~2 is next to two edges between levels~2 and~3 of the same clip.
			These edges between levels~2 and~3 cannot be drawn ``under''
			the mountain so one of them ends up to the left and one
			to the right of the mountain containing~$e$.
			Therefore, no two edges between levels~1 and~2
			of the same clip can end up ``under'' the same mountain.
			
			Assume that two edges between levels~1 and~2
			of different clips end up under the same mountain.
			Then their neighboring edges between the levels~2 and~3
			would interleave each other,
			which contradicts contradicting \cref{lem:disjoint_clips}.
		\end{claimproof}
		
		\begin{claim}
			In any level planar drawing,
			every $s_j$-clip  ($j \in [3m]$) occupies mountains of
			a single mountain $B$-chain and is drawn in such a way that
			each of its edges between levels~1 and~2 occupies one mountain.
		\end{claim}
		\begin{claimproof}
			The unique drawing of the sequence of mountain $B$-chains
			and the constraints $w_0 \prec_4 x$ and $x \prec_4 w_m$
			for the level-$4$ vertex of each $s_j$-clip imply that
			the clips must to be drawn in the same space as the mountains.
			The walls split the drawing into a right-of-the-wall part
			and a left-of-the-wall part,
			so each clip can be drawn onto at most a single mountain $B$-chain.
			The only space to the draw level-$1$ vertices is ``under''
			the mountains, hence, the edges between levels~1 and~2
			of each clip occupy mountains.
			\cref{clm:single_edge} showed that each mountain contains at most one such edge.
		\end{claimproof}
		
		By design, the number of edges between levels~1 and~2 within
		the clips is equal to the number of mountains within
		the mountain $B$-chains.
		In the claims, we proved that each mountain holds exactly one 
		such edge and that each clip falls into a single mountain $B$-chain.
		As the clips represent positive numbers that are strictly between
		$B/2$ and $B/4$ it is clear that each mountain $B$-chain contains
		exactly $3$ clips, which directly corresponds to a solution
		of the corresponding \textsc{$3$-Partition} instance.
	\end{proof}

	\section{A Linear-Time Algorithm for 2-Level \CLPlong}
	\label{sec:clp-2lvl}
	
	In this section, we present an optimal linear (in the input size) time algorithm for 2-level
	\CLP, which serves as a warm-up
	for our much more involved approach for 3-level \CLP (discussed in
	\cref{sec:clp-3lvl}).
	We first discuss how to handle isolated vertices and
	then how to handle connected input graphs.
	Finally, we observe that connected components of a disconnected input graph can be handled individually, which allows us to state our main algorithm.
	
	\paragraph{Handling isolated vertices.}
	Observe that a 1-level \CLP instance $\mathcal{G} = (G, \gamma, (\prec_i)_{i \in [1]})$
	contains no edges, so that all drawings of~$\mathcal G$ are crossing-free by default.
	Hence, finding a constrained level drawing of~$\mathcal G$ is equivalent to finding a linear
	extension of the partial order $\prec_1$, which can be done in linear time by performing a
	topological sorting of the directed graph corresponding to~$\prec_1$ (recall that a topological sorting of a directed (multi)graph can be performed in time that is linear in its number of vertices and edges; here the number of edges corresponds to the number of constraints).
	Following this idea, any constrained level graph can be transformed into an equivalent
	simpler instance by removing its isolated vertices:
	
	\begin{lemma}[{\hspace{1sp}\cite[Lemma 4]{clp-vc}}]\label{lem:isolated_vertices}
		Let~$\mathcal{G} = (G, \gamma, (\prec_i)_i)$ be a constrained level graph, let~$G'$ be
		the subgraph of~$G$ induced by the non-isolated vertices $V'$,
		and let~$\gamma'$ and~$(\prec'_i)_i$ be the restrictions of~$\gamma$ and~$(\prec_i)_i$
		to~$V'$, respectively.
		There is an algorithm that, given~$\mathcal G$ and a constrained level planar drawing~$\Gamma'$
		of $\mathcal{G}'=(G',\gamma',(\prec'_i)_i)$, constructs a constrained level
		planar drawing of~$\mathcal G$ in linear\footnote{The authors of \cite{clp-vc} did not
			analyze the polynomial factors in the runtimes of their algorithms and just stated that the
			runtime of \cite[Lemma 4]{clp-vc}] is polynomial. However, it is easy to see that their
			approach only requires linear time.
		} (in the size of the input) time.
	\end{lemma}
	
	To give a more complete picture, we include the proof idea:
	
	\begin{proof}[Proof of \cref{lem:isolated_vertices}]
		For each level~$V_\ell$, the drawing~$\Gamma'$ induces a linear order~$<_\ell'$ on
		$V_\ell'$ that extends $\prec'_\ell$.
		Klemz and Sieper~\cite{clp-vc} showed that the order $<_\ell=(<_\ell'\cup \prec_\ell)$ is
		acyclic and that, hence, the desired drawing can be obtained by computing a linear
		extension of  $<_\ell$ for each level~$V_\ell$.
		This can be done by performing a topological sorting of the directed
		graph corresponding to
		$\prec_\ell$ augmented with the directed edges corresponding to the transitive reduction
		(i.e., not including transitive edges) of~$<_\ell'$.
		Note that the transitive reduction of~$<_\ell'$ contains only $\mathcal O(V_\ell)$ edges,
		which are given as part of the encoding of~$\Gamma'$.
		Thus, the {\it total} runtime for computing the topological sortings on all levels is indeed
		linear in the size of the input ($\mathcal G$ and $\Gamma'$).
	\end{proof}
	
	\begin{figure}[tb]
		\centering
		\includegraphics{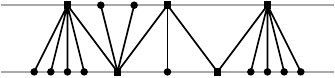}
		\caption{A 2-level drawing of a caterpillar. The vertices of its spine are drawn as squares.}
		\label{fig:caterpillar_level}
	\end{figure}
	
	\paragraph{Handling connected graphs.}
	A \emph{caterpillar} is a tree in which all vertices are within distance at most one of a
	central path (in other words, removing all degree-1 vertices results in a graph that is a path),
	which is called \emph{spine}.
	For a caterpillar $G$ with at least three vertices, let $p(G)$ be the subgraph of $G$ induced
	by $\{v \in V(G)\mid \mathrm{deg}(v) \geq 2\}$, i.e., $p(G)$ is the spine of $G$.
	A drawing of a caterpillar on two levels can be seen in
	\cref{fig:caterpillar_level}.
	Caterpillars are precisely the connected graphs that can be embedded on two
	levels~\cite{HS71} and we can construct constrained level planar drawings of them
	efficiently:
	
	\begin{lemma}
		\label{lem:catInPolyTime}
		Let $\mathcal{G} = (G=(V,E), \gamma, (\prec_i)_{i \in [2]})$ be a constrained level graph
		where $G$ is a connected.
		There is an algorithm that, given~$\mathcal G$, either creates a constrained level planar
		drawing of $\mathcal{G}$ or determines that such a drawing does not exist in
		linear (in the size of the input) time.
	\end{lemma}
	
	\begin{proof}
		If $G$ is not a caterpillar (which can be tested in linear time), $\mathcal{G}$ is not
		(constrained) level planar.
		So assume that $G$ is a caterpillar.
		We may assume that~$G$ has at least three vertices, otherwise a constrained level planar
		drawing is trivial to construct in constant time.
		
		There are two possible options to draw $p(G)=(v_1,v_2,\dots,v_k)$ 
		on two levels in a level planar fashion: as a zigzag path alternating between
		the two levels with either $v_1v_2$ as the left-most edge and $v_{k-1}v_k$ as the right
		most edge or vice versa.
		Suppose we want to test whether a constrained level planar drawing of the former form exists (testing whether a
		drawing of the latter form exists can be done symmetrically).
		We will do so by performing two topological sortings, one in a directed graph obtained
		by augmenting the graph~$H_1$ corresponding to~$\prec_1$, the other in a directed graph
		obtained by augmenting the graph~$H_2$ corresponding to~$\prec_2$.
		The desired drawing (if it exists) is then obtained by simply placing the vertices of each
		level in the order given by its topological sorting.
		Without loss of generality, we may assume that $v_1\in V_1$.
		
		To ensure that $p(G)$ is drawn crossing-free in the desired zigzag fashion, we add the
		edges $(v_1,v_3),(v_3,v_5),(v_5,v_7),\dots$ to $H_1$ and we add the edges
		$(v_2,v_4),(v_4,v_6),(v_6,v_8),\dots$ to $H_2$.
		To ensure that the remaining edges also have to be drawn in a crossing-free fashion,
		we observe that for each $i\in \{2,3,\dots,k-1\}$, every neighbor $u$ of the vertex
		$v_i$ has to be drawn between the vertices $v_{i-1}$ and $v_{i+1}$ on the level that does
		not contain~$v_i$.
		Similarly, every neighbor $u$ of $v_1$ has to be drawn on $L_2$ and to the left of
		$v_2$ (if it exists) and every neighbor $u$ of $v_k$ has to be drawn on the level that does
		not contain~$v_k$ and to the right of $v_{k-1}$ (if it exists).
		It is easy to see that these conditions are not only necessary, but also sufficient to
		guarantee a crossing free level drawing.
		Moreover, they can be enforced by adding at most two constraints per neighbor $u$ of a spine vertex.
		So overall, it suffices to augment $H_1$ and $H_2$ with at most $2n$ edges, where
		$n$ is the number of vertices in $G$, to ensure that the drawing corresponding to the two
		topological orderings (if they exist) is crossing-free and realizes $p(G)$ in the desired
		fashion.
		Moreover, by definition of $H_1$ and $H_2$, the drawing also satisfies all constraints,
		i.e., it is constrained level planar.
		Conversely, if at least one of our two augmented directed graphs~$H_1',H_2'$ admits
		no topological ordering, then there is no constrained level planar drawing of~$\mathcal G$.
		
		Given that the sizes of $H_1'$ and $H_2'$ are both linear in the size of~$\mathcal G$,
		the two topological sortings can be obtained in the desired runtime.
	\end{proof}
	
	\paragraph{Main algorithm.}
	The fact that each edge of a 2-level planar drawing spans all (two) levels implies that the
	connected components can be handled individually (\cref{lem:catsDontCross}), which leads
	to our linear-time algorithm for 2-level CLP (\cref{thm:clp-2lvl}).
	
	\begin{lemma}\label{lem:catsDontCross}
		Let $\mathcal{G} = (G, \gamma)$ be a 2-level graph, and let $A$ and $B$ be two
		distinct 
		connected components of $G$, both containing at least one edge.
		Then for each level drawing of~$\mathcal G$ there exists an unbounded y-monotone curve that separates $A$ from $B$ (i.e., $A$ lies to the left of $B$ on every level or $B$ lies to the left of $A$ on every level).
	\end{lemma}
	
	\begin{proof}
		Let $\Gamma$ be a level planar drawing of $\mathcal{G}$,
		and let $e = vu \in E(A)$.  Then $e$ splits the horizontal strip bounded by $L_1$
		and $L_2$ into a
		left and a right part, and all paths in $G$ from the left to the right part
		have to intersect $u$ or $v$.  Given that $A$ and $B$ are disjoint, it follows that $B$
		must lie entirely to the
		left or to the right of $e$.
		A symmetric statement applies to each edge $e'\in E(B)$.
		In combination, this yields the claim.
	\end{proof}
	
	\begin{theorem}
		\label{thm:clp-2lvl}
		There
		exists an algorithm that, given a constrained 2-level graph
		$\mathcal{G} = (G, \gamma,$ $(\prec_i)_{i\in [2]})$, either computes a
		constrained level planar drawing of $\mathcal{G}$ or determines that such a drawing does
		not exist in linear (in the size of the input)
		time.
	\end{theorem}
	\begin{proof}
		In view of \cref{lem:isolated_vertices}, we may assume that~$G$ contains no isolated
		vertices, as they can be handled in a linear time postprocessing step.
		We begin by computing the connected components, denoted by $C^1,C^2,\dots,C^k$, of
		$G$, which is easily done in linear time via BFS.
		Then we create a constrained level graph $\mathcal G^j=(C^j,\gamma^j,(\prec_i^j)_{i\in [2]})$
		for each connected component~$C^j$, where $\gamma^j$ and $(\prec_i^j)i$ are the
		restrictions of $\gamma$ and $(\prec_i)_i$ to the vertices of $C^j$, respectively.
		This can be done in linear time by sweeping the set of constraints once.
		While doing so, we also create (in linear time) a directed multigraph~$H$ with a node for
		each connected component~$C^j$ and with a directed edge from a component $C^v$ to a
		component~$C^u$ for each constraint $v\prec_i u$ with $i\in [2], v\in C^v, u\in C^u$.
		
		In view of \cref{lem:catsDontCross}, a constrained level planar drawing of $\mathcal G$
		exists if and only if (a) $H$ is acyclic and (b) each of the graphs $\mathcal G^j$ admits a
		constrained level planar drawing.
		Condition (a) can be tested by applying a topological sorting algorithm on $H$ in linear time.
		Condition (b) can be tested by applying \cref{lem:catInPolyTime} for each $\mathcal G^j$.
		Given that the graphs $\mathcal G^j$ partition $\mathcal G$, this takes linear time in total.
		Finally, to construct the desired drawing (if it exists), we simply arrange the drawings of the
		graphs $\mathcal G^j$ in the order given by the topological sorting of~$H$.
	\end{proof}

	\section{Tractability of 3-Level \CLPlong}
	\label{sec:clp-3lvl}

	In this section we show that, given a constrained level graph
	$\G = (G, \gamma, (\prec_i)_{i\in [3]})$ of height~3,
	we can decide in polynomial time whether \G
	admits a constrained level planar drawing.
	We assume that \G is level planar (which can be tested in linear
	time~\cite{DBLP:conf/gd/HeathP95, DBLP:conf/gd/JungerLM98})
	and proper (otherwise we subdivide edges that connect levels 1 and~3).
	Besides referring to levels 1, 2, and~3 by their indices,
	we also call them \emph{bottom} level, \emph{middle} level, and \emph{top} level, respectively.
	Furthermore, we call the pair of bottom and middle level \emph{lower band}
	and the pair of middle and top level \emph{upper band}.
	
	Throughout this section, we successively add new constraints to the
	given constrained level graph~\G that we deduce from the structure
	of~$G$, $\gamma$, and the current set of constraints.  In the end, this yields a
	total order of the vertices for each of the three levels 
	that corresponds to a
	constrained level planar drawing of~\G, or we come to the conclusion
	that \G does not admit such a drawing.  In the very beginning and
	whenever we add new constraints, we exhaustively add the following
	\emph{implicit} constraints:
	\begin{itemize}%
		\item Transitivity:
		For every triplet $(a, b, c)$ of vertices on the same level,
		if there is a constraint from $a$ to $b$ and from $b$ to $c$,
		then there is also a constraint from $a$ to $c$.
		Formally, \\
		$\forall a,b,c\in V(G)$ with $i := \gamma(a) = \gamma(b) = \gamma(c) \colon (a \prec_i b) \land (b
		\prec_i c) \Rightarrow (a \prec_i c)$.
		\item Planarity:
		For every pair $(ab, cd)$ of edges such that $a$ and $c$,
		and $b$ and $d$ lie on the same level,
		if there is a constraint from $a$ to $c$,
		then there is also a constraint from $b$ to $d$ or $b = d$.
		Formally,
		$\forall ab, cd \in E(G)
		\text{ with } i := \gamma(a) = \gamma(c)
		\text{ and }  j := \gamma(b) = \gamma(d) \colon$
		$(a \prec_i c) \Rightarrow (b \prec_j d) \lor (b = d)$.
	\end{itemize}
	The transitivity constraints ensure that the orderings
	$(\prec_i)_{i\in[3]}$ remain transitive while the planarity
	constraints can be added without violating realizability, as they
	need to be respected in every constrained level planar drawing. The
	propagation of these constraints is quite useful as it can dictate the
	relative positions of vertices that are initially unrelated, see
	\cref{fig:planar-leaves}.
	Note that we ignore the runtime of adding implicit constraints
	throughout this section and get back to it in \cref{thm:clp-3lvl}.
	
	In this section, when we state asymptotic running times such as
	``linear'' or ``quadratic'', we refer to the number of vertices
	of~$G$.  Note that, due to the constraints, the size of the input can
	already be quadratic (in the number of vertices).
	
	\begin{figure}[t]
		\centering
		\begin{subfigure}[t]{.29\textwidth}
			\centering
			\includegraphics[page=1]{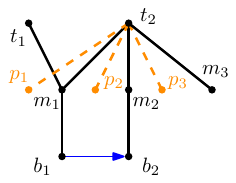}
			\subcaption{Initial situation with the given constraint
				$b_1 \prec_1 b_2$ but without implicit constraints.}
			\label{fig:planar-leaves-a}
		\end{subfigure}
		\hfill
		\begin{subfigure}[t]{.34\textwidth}
			\centering%
			\includegraphics[page=3]{clp3_leaves}
			\subcaption{Constraint $b_1 \prec_1 b_2$ and edges
				$b_1m_1$ and $b_2m_2$ yield
				$m_1 \prec_2 m_2$; edges $m_1t_1$ and $m_2t_2$ then yield
				$t_1 \prec_3 t_2$.}
			\label{fig:planar-leaves-b}
		\end{subfigure}
		\hfill
		\begin{subfigure}[t]{.29\textwidth}
			\centering
			\includegraphics[page=4]{clp3_leaves}
			\subcaption{Constraint $t_1 \prec_3 t_2$ and edges
				$t_1m_1$ and $t_2m_3$ yield
				$m_1 \prec_2 m_3$.
				Hence, $m_3$ cannot be placed at~$p_1$.}
			\label{fig:planar-leaves-c}
		\end{subfigure}
		
		\caption{Constrained level graph with only one constraint
			($b_1 \prec_1 b_2$).  Among the three possible positions $p_1$,
			$p_2$, and $p_3$ to place $m_3$ relative to
			$m_1$ and $m_2$, position
			$p_1$ is excluded due to the implicit constraints ensuring
			planarity.}
		\label{fig:planar-leaves}
	\end{figure}

	\paragraph{Some simplifying assumptions.}
	For the remainder of this section, we assume that, for each level,
	the current set of constraints among the vertices on this level is acyclic.
	Otherwise, we reject the instance.
	Due to \cref{lem:isolated_vertices}, we can assume, without loss of
	generality, that~\G contains no isolated vertices.
	We now make some general observations for constrained level graphs
	that we will later apply to the concrete constrained level graph~\G
	that we are given as input.
	To this end, let $\J = (J, \gamma, (\prec_i)_{i\in [h]})$
	be a constrained level graph of height~$h$.
	
	We define the \emph{component--constraint graph $H_\J$} of \J
	as follows.
	The nodes of~$H_\J$
	are the connected components of~$J$, and there is an arc
	from a component~$C$ to a component~$D$ if there are vertices~$u$
	in~$C$ and~$v$ in~$D$ with $u \prec_i v$ for some level~$i$;
	see \cref{fig:hooking} (the terms ``hook chain'', ``hook anchor'', and ``hook piece'' will be explained later).
	We define, for a component $C$ of $J$ (that is, a node of $H_\J$), the graph $\J_C$
	as the subgraph of~\J induced by the vertices of~$C$.
	Similarly, we define, for a set~\C of components of~$J$ (that is, a set of nodes of~$H_\J$),
	the subgraph \emph{$\J_\C$} $= \bigcup_{C \in \C} \J_C$ of~\J.

	\begin{lemma}\label{lem:3lclp_strongly_connected}
		A constrained level graph \J admits a constrained level planar
		drawing if and only if, for every strongly connected component
		$\mathcal C$ of the component--constraint graph~$H_\J$,
		the corresponding subgraph $\J_\C$
		admits a constrained level planar drawing.
	\end{lemma}
	\begin{proof}
		Obviously, a constrained level planar drawing of \J
		contains, for each strongly connected component $\mathcal C$ of $H_\J$,
		a constrained level planar drawing of $\J_\C$.
		
		Now, given, for each strongly connected component $\mathcal C$ of $H_\J$,
		a constrained level planar drawing of $\J_\C$,
		we construct a constrained level planar drawing of \J.
		Let $H_\J'$ be the directed graph obtained from $H_\J$ by contracting
		each strongly connected component to a single node.
		Clearly, $H_\J'$ is acyclic. We sort $H_\J'$
		topologically and place the individual drawings of the strongly connected
		components according to that order from left to right.
	\end{proof}
	
	Hence, from now on we assume, without loss of generality,
	that the component--constraint graph~$H_\J$ of any level graph~\J
	is strongly connected.
	(Otherwise, we compute the component--constraint graph $H_\J$ of $\J$ in quadratic time
	and treat each strongly connected component of $H_\J$ individually.)
	In particular, this assumption holds for~\G.
	Since we assume that \G also does not contain
	any isolated vertices and is proper,
	every connected component contains at least one vertex on the
	middle level.
	For brevity, we call connected components just \emph{components}
	and if we speak of components of a constrained level graph
	$\J = (J, \gamma, (\prec_i)_{i\in [3]})$,
	we mean the components of~$J$.
	
	Next, we investigate
	some properties that are specific to constrained level planar drawings
	on three levels.
	Provided \G is constrained level planar,
	these properties will lead to further assumptions on the structure of \G.
	
	\paragraph{On the interaction of components.}

	\begin{figure}[tb]
		\centering
		\begin{subfigure}[t]{.7\textwidth}
			\centering
			\includegraphics[page=1]{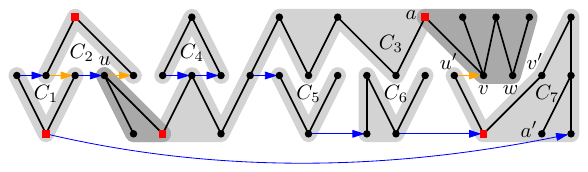}
			\subcaption{constrained level planar drawing of~\J:
				components $C_1,\dots,C_7$ in gray, hook anchors in red,
				and hook pieces of~$C_3$ in dark gray}
			\label{fig:hooking1}
		\end{subfigure}
		\hfill
		\begin{subfigure}[t]{.27\textwidth}
			\centering
			\includegraphics[page=2]{clp3_H}
			\subcaption{the corresponding
				com\-ponent--constraint graph~$H_\J$}
			\label{fig:hooking2}
		\end{subfigure}
		
		\caption{Finding the hook chain $\langle C_1, C_2, C_3, C_7
			\rangle$ (corresponding arcs marked in orange).}
		\label{fig:hooking}
	\end{figure}
	
	Let $\Gamma$ be a constrained level planar drawing of
	a proper constrained level graph $\J = (J, \gamma, (\prec_i)_{i\in [3]})$
	of height~3.
	Let $\sigma$ be the order of
	the vertices on the middle level of $\Gamma$ from left to right.
	For a subgraph $X$ of $J$, let the \emph{span} of $X$ in~$\Gamma$,
	$\Span_\Gamma(X)$, be the smallest interval of~$\sigma$ that contains all
	middle-level vertices of~$X$.
	A component~$C$
	\emph{encloses} a component~$C'$ if
	$\Span_\Gamma(C')\subsetneq\Span_\Gamma(C)$;
	we then call $C'$ an \emph{enclosed} component.
	For example, in \cref{fig:hooking}, $C_3$ encloses $C_4$, $C_5$, and $C_6$.
	Recall that a \emph{caterpillar} is a tree in
	which all vertices are within distance at most one of a central path,
	which we call its \emph{spine}; see \cref{fig:caterpillar_level} for an example.
	
	\begin{observation}
		\label{obs:enclosed-components-caterpillars}
		An enclosed component occupies exactly two levels
		and is a caterpillar.
	\end{observation}
	\begin{proof}
		An enclosed component~$C$ occupies at least two levels
		because we assume that the given constrained level graph has no isolated vertices.
		Further, $C$ cannot occupy three levels because otherwise,
		the component that encloses~$C$ could not
		have middle-level vertices to the left and to the right
		of the middle level vertices of~$C$.
		Caterpillars are precisely the connected graphs
		that can be embedded into two levels~\cite{HS71}.
	\end{proof}
	
	Let $C$ and $C'$ be two distinct components of~\J,
	and let $\Gamma$ be a constrained level planar drawing of~\J.
	We say that $C'$ \emph{hooks into}~$C$ \emph{from the right}
	(in~$\Gamma$)
	if there are vertices $u \ne v$ of
	$C$ and vertices $u' \ne v'$ of $C'$ such that $u$, $u'$, $v$, $v'$
	occur in this order on the middle level, and $u$ and
	$v'$ are the leftmost and rightmost middle-level vertices of
	$C \cup C'$, respectively.  For example, in \cref{fig:hooking},
	$C_7$ hooks into~$C_3$ from the right,
	but $C_3$ does not hook into~$C_7$ from the right.
	For a pair of components $(C,C')$ such that $C'$ hooks
	into~$C$ from the right,
	we introduce names for some of the vertices of $C$;
	the same names apply symmetrically to the vertices of~$C'$.
	We call the vertices in $C \cap \Span_\Gamma(C')$
	\emph{hook vertices} of~$C$ towards~$C'$.
	For example, in \cref{fig:hooking}, $u$ is a hook vertex
	of~$C_3$ towards~$C_2$,
	and $v$ and $w$ are hook vertices of~$C_3$ towards~$C_7$.
	Let the \emph{hook piece} of $C$ towards $C'$ be the subgraph of~$J$ induced by the closed
	neighborhood of the hook vertices of $C$ towards~$C'$.
	For example, in \cref{fig:hooking1}, the hook pieces of~$C_3$
	are marked in dark gray.
	
	\begin{observation}
		\label{obs:hooking}
		Let \J be a proper constrained level graph of height~3,
		and let $\Gamma$ be a constrained level planar drawing of~\J.
		If~\J has two components $C$ and $C'$ such that one hooks into
		the other from the right in $\Gamma$, then the following
		statements hold.
		\begin{enumerate}
			\item The hook piece~$p$ of $C$ towards $C'$ is a
			caterpillar on two levels.
			\item\label{item:hook-anchor}
			The hook piece~$p$ of $C$ towards $C'$ contains
			exactly one vertex with a neighbor outside of~$p$.
			\item All neighbors of the hook vertices of~$C$ towards $C'$
			lie on the same level
			(see, e.g., the neighbors of $v$ and $w$ in~\cref{fig:hooking1}).
			\item If all neighbors of the hook vertices of~$C$ towards $C'$
			lie on the top level, then all neighbors of the hook vertices
			of~$C'$ towards $C$ lie
			on the bottom level, and vice versa
			(see, e.g., the neighbors of $v$ and $w$ versus the neighbor
			of~$u'$ in~\cref{fig:hooking1}).
		\end{enumerate}
	\end{observation}
	\begin{proof}
		~
		\begin{enumerate}
			\item Consider the drawing $\Gamma'$ where we
			have removed $C$ except for~$p$.
			By definition, $C'$ encloses $p$ in $\Gamma'$.
			Due to \cref{obs:enclosed-components-caterpillars},
			$p$ is a caterpillar on two levels.
			\item By definition, the middle-level vertices of $p$
			and the middle-level vertices of $C \setminus p$
			form two disjoint intervals.
			Hence, their closed neighborshoods intersect
			in at most one vertex.
			\item Since $p$ is a caterpillar on two levels,
			all neighbors of the hook vertices of~$p$ lie on the same level.
			\item If the neighbors of the hook vertices of~$C$ towards $C'$
			and the neighbors of the hook vertices of~$C'$ towards $C$
			would lie on the same level, their incident edges would cross.
		\end{enumerate}
	\end{proof}
	
	The vertex in the hook piece of~$C$
	that has neighbors outside the hook
	piece is unique due to item~\ref{item:hook-anchor} of \cref{obs:hooking};
	we call it the \emph{hook anchor} of~$C$ towards~$C'$.
	In \cref{fig:hooking}, the vertices that are hook anchors
	are marked by red squares.
	Let a \emph{hook chain} be a maximal sequence
	$\langle C_1,\dots,C_k\rangle$ of $k\ge 2$
	components such that, for every $i\in [k-1]$,
	$C_{i+1}$ hooks into $C_{i}$ from the right.
	If a component $C$ neither hooks into another component from the right
	nor another component hooks into $C$ from the right
	and $C$ is not enclosed by another component,
	then we call $\langle C \rangle$ a
	\emph{degenerate} hook chain;
	see the red component in \cref{fig:3levelclp}.
	
	\begin{lemma}\label{clm:span-empty-intersection}
		Let \J be a proper constrained level graph of height~3,
		and let $\Gamma$ be a constrained level planar drawing of~\J.
		If $C$, $C'$, and $C''$ are three distinct components of~\J, then
		$\Span_\Gamma(C) \cap \Span_\Gamma(C') \cap \Span_\Gamma(C'') = \emptyset$.
	\end{lemma}
	
	\begin{proof}
		Suppose for a contradiction that there are three distinct
		components~$C$, $C'$, and $C''$ with
		$\Span_\Gamma(C) \cap \Span_\Gamma(C') \cap \Span_\Gamma(C'') \ne \emptyset$.  Without
		loss of generality, we can assume that~$C'$ and~$C''$
		have at least two vertices on the middle level each.
		Let~$v$ be a vertex of $C$ that lies in
		$\Span_\Gamma(C) \cap \Span_\Gamma(C') \cap \Span_\Gamma(C'')$.
		We can also assume that $v$ has a neighbor~$\bar{v}$ on the top level. 
		Let $u'$ be the middle-level vertex of~$C'$
		to the left of~$v$ that is closest to~$v$.
		Let $w'$ be the middle-level vertex of~$C'$ to the right
		of~$v$ that is closest to~$v$.  Such vertices must exist since
		$v \in \Span_\Gamma(C')$.  Let $u''$ and $w''$ be vertices of $C''$
		defined in the same way as $u'$ and $w'$.
		Due to the edge $v\bar{v}$,
		all paths from~$u'$ to~$w'$ must contain a vertex~$v'$ on
		the bottom level whose predecessor lies to the left of~$v$,
		while its successor lies to the right of~$v$.
		Symmetrically, all paths from~$u''$ to~$w''$ must contain a vertex~$v''$ with the same properties as~$v'$.
		Then, however, one of the edges incident to~$v'$
		must cross one of the edges incident to~$v''$.
		This yields the desired contradiction.
	\end{proof}
	
	We next prove some properties of hook chains
	in constrained level planar drawings on three levels.
	
	\begin{lemma}
		\label{lem:hook-chain-enclosed}
		In a constrained level planar drawing~$\Gamma$
		of a proper constrained level graph~\J of height~3,
		every component is either part of a hook chain or
		enclosed by a component that is part of a hook chain.
	\end{lemma}
	\begin{proof}
		Observe that \cref{clm:span-empty-intersection} implies
		that no enclosed component can enclose another component.
		
		Now, we prove that if there is a component~$C$
		that is not part of a hook chain,
		then $C$ is an enclosed component.
		Because $\langle C \rangle$ is not a degenerate hook chain
		and because of the strong connectivity of~$H_\J$,
		there is a component $C'$ such that
		$\Span_\Gamma(C) \cap \Span_\Gamma(C') \ne \emptyset$
		and $\Span_\Gamma(C) \nsupseteq \Span_\Gamma(C')$.
		If $\Span_\Gamma(C) \subseteq \Span_\Gamma(C')$,
		then $C$ is an enclosed component.
		Otherwise $C$ and $C'$ overlap.
		Without loss of generality, there is a vertex~$v$ in $C$
		on the middle level that lies to the left of $\Span_\Gamma(C')$.
		There is no constraint from a vertex of $C'$ to a vertex of $C$
		because otherwise $C'$ would hook into~$C$ from the right.
		However, because of the strong connectivity of~$H_\J$,
		there is, for some $\ell > 2$,
		a path $\langle C' = C_1, C_2, \dots, C_\ell = C \rangle$
		in $H_\J$ such that, for every $i \in [\ell - 1]$,
		it holds that $\Span_\Gamma(C_i) \cap \Span_\Gamma(C_{i+1}) \ne \emptyset$.
		Since $\Span_\Gamma(C') \cap \Span_\Gamma(C_2)$ is to the right
		of $\Span_\Gamma(C) \cap \Span_\Gamma(C')$,
		$\Span_\Gamma(C) \cap \Span_\Gamma(C_{\ell-1})$ is to the left
		of $\Span_\Gamma(C) \cap \Span_\Gamma(C')$,
		and, for every $i \in \{2, \dots, \ell - 1\}$,
		$\Span_\Gamma(C) \cap \Span_\Gamma(C') \cap \Span_\Gamma(C_i) = \emptyset$
		(due to \cref{clm:span-empty-intersection}), this is not possible.
		Hence, every component is part of a hook chain
		or is enclosed by a component of a hook chain.
	\end{proof}
	
	\begin{lemma}\label{lem:3lclp_hookchain1}
		A constrained level planar drawing~$\Gamma$
		of a proper constrained level graph~\J of height~3
		contains exactly one (possibly degenerate) hook chain.
		Every cycle in the component--constraint graph $H_\J$ contains at least one node that corresponds to
		a component that is part of the hook chain.
	\end{lemma}
	\begin{proof}
		First, we show that there is exactly one hook chain.
		There cannot be two different hook chains $\mathcal C_1$ and~$\mathcal C_2$
		because $\Span_\Gamma(\mathcal C_1) \cap \Span_\Gamma(\mathcal C_2) = \emptyset$,
		but then, if say $\mathcal C_1$ is to the left of $\mathcal C_2$,
		there cannot be a directed path from a component in $\mathcal C_2$
		to a component in $\mathcal C_1$ in $H_\J$,
		which contradicts the strong connectivity of~$H_\J$.
		Since there needs to be at least one component~$C$
		whose span is not contained in the span of another component,
		$C$ is part of a hook chain or $\langle C \rangle$
		is a degenerate hook chain.
		
		Finally, suppose that there is a cycle in $H_\J$ containing no
		node representing a component of the hook chain.
		Then, two consecutive components $C_i$ and $C_{i+1}$
		of the cycle may not hook.
		Hence, $C_{i+1}$ must be enclosed by $C_i$ (or vice versa).
		However, by the same argument $C_{i+1}$ must
		enclose a component $C_{i+2}$ (if it exists), which
		is, again, not possible by \cref{clm:span-empty-intersection}.
		If the cycle contains only two components, $C_1$ may
		enclose $C_2$ (or vice versa), but this would
		mean that $\langle C_1 \rangle$ (or $\langle C_2 \rangle$) is a degenerate hook chain.
	\end{proof}
	
	\begin{lemma}
		\label{lem:3lclp_hookchain2}
		Let $\Gamma_1$ and $\Gamma_2$ be two constrained level planar
		drawings of a proper constrained level graph \J of height~3
		whose hook chains are $\C_1$ and~$\C_2$, respectively.
		If $\C_1$ and $\C_2$ start in the same component~$C$
		and end in the same component $C'$, then $\C_1 = \C_2$.
		
		If \J admits a constrained level planar drawing
		with a hook chain starting
		in a component~$C$ and ending in a component~$C'$,
		we can determine this hook chain in linear time given~$C$, $C'$, and~$H_\J$.
	\end{lemma}
	\begin{proof}
		In a hook chain, for each pair of consecutive components,
		the later component has a constraint to the earlier component.
		Therefore, a hook chain from $C$ to $C'$ in~$\Gamma_1$
		(resp.~$\Gamma_2$) corresponds to a simple directed path~$\pi_1$
		(resp.~$\pi_2$) from~$C'$ to~$C$ in~$H_\J$.
		If $\pi_1=\pi_2$, then clearly $\C_1 = \C_2$.
		Now suppose for a contradiction that $\pi_1 \ne \pi_2$.
		Traverse $\pi_1$ and $\pi_2$ simultaneously until
		they differ for the first time.
		Let $\hat{C}$ be the last component that lies on both $\pi_1$ and $\pi_2$,
		and let $\hat{C}_1$ and~$\hat{C}_2$
		be its distinct successors on $\pi_1$ and $\pi_2$, respectively.
		In~$\Gamma_1$, $\hat{C}$ hooks into $\hat{C}_1$ from the right,
		and hence $\Span_{\Gamma_1}(\hat{C}) \cap \Span_{\Gamma_1}(\hat{C}_1)
		\ne \emptyset$.
		According to \cref{clm:span-empty-intersection},
		the vertices of $\Span_{\Gamma_1}(\hat{C}) \cap \Span_{\Gamma_1}(\hat{C}_1)$
		(i.e., the middle-level vertices of the hook pieces of $\hat{C}$ and $\hat{C}_1$)
		divide~$\Gamma_1$ into a left and a right part.
		Hence, every component distinct from $\hat{C}$ and $\hat{C}_1$
		lies either completely in the left or completely in the right part.
		Since there is a constraint from a vertex of $\hat{C}$
		to a vertex of $\hat{C}_2$,
		$\hat{C}_2$ lies in the right part.
		Repeating this argument, all components of $\pi_2$
		succeeding $\hat{C}$ lie in the right part,
		which, in particular, includes $C$.
		However, due to \cref{lem:3lclp_hookchain1,lem:hook-chain-enclosed},
		the first vertex on the middle level in $\Gamma_1$
		belongs to~$C$, which means that~$C$ lies in the left part~--
		a contradiction.
		
		We can find the unique simple path in $H_\J$ from~$C'$ to~$C$
		in linear time using a breadth-first search starting in~$C'$.
	\end{proof}
	
	Due to \cref{lem:3lclp_hookchain2}, it suffices to check,
	for a quadratic number of pairs of
	components $(C,C')$,  whether~\G admits a constrained level planar
	drawing with a hook chain starting in~$C$ and ending in~$C'$.
	We simply check each pairwise combination of components
	(including degenerate hook chains where $C = C'$).
	Therefore, from now on, we assume that we know the hook chain of~\G.
	
	\begin{lemma}
		\label{lem:3lclp_hookanchor}
		Given a proper constrained level graph $\J = (J, \gamma, (\prec_i)_{i\in [3]})$ of height~3 and a sequence of components of~\J
		that appears as the hook chain in a level planar drawing of \J,
		we can find in linear time four sets of vertices such that
		there is a constrained level planar drawing of \J
		whose set of hook anchors is one of these sets.
	\end{lemma}
	
	\begin{proof}
		For simplicity, we just use ``hook chain'' when we speak
		of the input sequence of components that appears
		as the hook chain in at least one level planar drawing.  Note that
		the statement of the lemma is trivially true for degenerate hook chains. 
		Otherwise, we consider every pair $(C,C')$ of consecutive components
		in the hook chain and compute a hook anchor of~$C$ towards~$C'$.
		The hook anchor of~$C'$ towards~$C$ can be computed symmetrically.
		
		Let $V_{C \gets C'}$ be the set of middle-level vertices of $C$ that are involved in
		constraints from~$C'$, and let $V_{C'\to C}$ be the set of middle-level vertices
		of~$C'$ that are involved in constraints to~$C$.
		For example, in \cref{fig:shift-hook-pieces},
		$V_{C \gets C'} = \{v,w'\}$.
		Observe that in any constrained level planar drawing of \J,
		the vertices in $V_{C \gets C'}$ are hook vertices
		of~$C$ towards~$C'$, and the vertices in $V_{C' \to C}$ are
		hook vertices of~$C'$ towards~$C$.
		Therefore, the neighbors of $V_{C \gets C'}$ all lie on the top level
		and the neighbors of $V_{C' \to C}$ all lie on the bottom level (or vice versa).
		We assume, without loss of generality, that the neighbors
		of~$V_{C \gets C'}$ all lie on the top level.
		
		If $C$ has a predecessor $\overline{C}$ in the hook chain,
		then there exists a vertex $u \in C$ with a constraint to some vertex in $\overline{C}$.
		Consider the BFS-tree~$T$ (of $J$)
		rooted in $u$.  Let $a$ be the vertex
		farthest from $u$ in $T$ such that the subtree rooted in $a$
		contains $V_{C \gets C'}$.  We claim that there is a constrained
		level planar drawing of~\J where~$a$ is the hook anchor of~$C$ towards~$C'$.
		Note that a hook anchor is a cut vertex of~$C$ separating $u$ and $V_{C \gets C'}$.
		Moreover, if there is a constrained level planar drawing~$\Gamma$ of \J
		where $C'$ hooks into $C$ from the right and the hook anchor of $C$ towards $C'$
		is not adjacent to a vertex of $V_{C \gets C'}$,
		then there is another constrained level planar drawing~$\Gamma'$ of~\J,
		which is identical to $\Gamma$ except that
		the hook vertices of~$C$ towards~$C'$ that are to
		the left of all vertices from~$V_{C \gets C'}$ are ``shifted''
		to the left of the leftmost middle-level vertex of~$C'$.
		For example, in \cref{fig:shift-hook-pieces},
		$u_1$, $u_2$, and $u_3$ are shifted to the left of~$u'$.
		Then, the ``shifted'' vertices are no longer hook vertices
		and the hook anchor of $C$ towards $C'$
		is adjacent to a vertex of $V_{C \gets C'}$.
		The correctness follows from the fact that the vertices that are
		not in the subtree
		of~$a$ in~$T$ do not have incoming constraints from~$C'$. 
		Note that, due to \cref{clm:span-empty-intersection},
		there cannot be another component between $C$ and $C'$.
		
		\begin{figure}[tb]
			\centering
			\begin{subfigure}[t]{.47\textwidth}
				\centering
				\includegraphics[page=1]{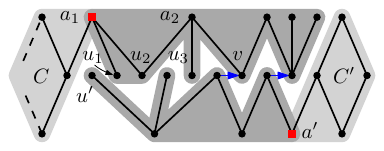}
				\subcaption{Constrained level planar drawing $\Gamma$
					with hook anchor $a_1$ in $C$ towards $C'$,
					which is not adjacent to a vertex in $V_{C \gets C'}$.}
				\label{fig:shift-hook-pieces-before}
			\end{subfigure}
			\hfill
			\begin{subfigure}[t]{.47\textwidth}
				\centering
				\includegraphics[page=2]{clp3_shift_hook_pieces}
				\subcaption{Constrained level planar drawing $\Gamma'$
					with hook anchor $a_2$ in $C$ towards $C'$,
					which is adjacent to vertex $v \in V_{C \gets C'}$.}
				\label{fig:shift-hook-pieces-after}
			\end{subfigure}
			
			\caption{By ``shifting'' hook vertices of $C$ to the left,
				we obtain, from a constrained level planar drawing $\Gamma$,
				another constrained level planar drawing $\Gamma'$
				where the hook anchor of $C$ towards $C'$ is adjacent
				to a vertex with a constraint from a vertex of $C'$ (here, $v$).
				Hook pieces are indicated by dark gray background color,
				hook anchors by red squares, and constraints by blue arrows.}
			\label{fig:shift-hook-pieces}
		\end{figure}
		
		Now, we consider the case that $C$ is the first component of the hook chain.
		If $C$ has a vertex~$u$ on the bottom level, we determine the hook
		anchor of~$C$ towards~$C'$ using a BFS-tree rooted in~$u$ as above.
		Otherwise, $C$ is contained in the upper band and is a caterpillar. 
		Since $C$ hooks, $C$ has at least three vertices and a nonempty spine.
		We take each of the at most two endpoints of the spine of~$C$
		as~$u$ and execute the rest of the algorithm once for each of them
		in order to find the hook anchor of~$C$ towards~$C'$.
		In the symmetric case, when $C'$ is the last component and we
		want to find the hook anchor towards $C$, we do the same as above.
		This results in at most four sets of vertices such that there exists a
		constrained level planar drawing that has precisely one of these sets
		as hook anchors.
		
		Regarding the running time, we execute a (linear-time) BFS
		at most four times per component of the hook chain.
		Together, the components have size linear in the size of~$J$,
		Hence, the overall running time is linear.
	\end{proof}
	
	Due to \cref{lem:3lclp_hookanchor},
	we can find, for \G and a given hook chain~\C,
	four sets of vertices one of which must be the set of hook anchors
	if \G admits a constrained level planar drawing with hook chain~\C.
	We simply check each of the four possibilities.
	Therefore, from now on, we assume that we know the hook chain
	and its set of hook anchors.
	
	Further note that every component contains at least one
	middle level vertex and that both the number of components
	and the number of vertices on the middle level are bounded by $O(n)$.
	Hence, within $O(n^2)$ time, instead of guessing $C$ and $C'$ directly,
	we can guess two middle level vertices $s$ and $t$
	and choose the components such that $s\in C$ and $t\in C'$.
	We will later use $s$ and $t$ as the first and last vertex
	on the middle level.

	\paragraph{Drawing a constrained level graph with given hook chain, hook anchors, $s$ and $t$.}
	It remains to find a constrained level planar drawing
	of \G, given a (possibly degenerate) hook chain,
	a set of hook anchors as well as the first and last vertices
	of the middle level $s$ and $t$.
	In several steps, we modify \G
	(we call the resulting constrained level graphs $\G = \G_0, \G_1, \dots, G_{10}$)
	such that, in step $i$, $\G_i$ is constrained level planar
	if and only if $\G_{i-1}$ is constrained level planar.
	This culminates in a proper constrained level graph~$\G_{10}$
	whose constraints yield a total order on each level
	such that we can easily check whether~$\G_{10}$
	is constrained level planar.
	In the affirmative case, we know that \G is
	constrained level planar, too, and we
	can find a constrained level planar drawing of~\G
	in $\mathcal{O}(n^5)$ time, where $n$ is the number of vertices of~\G;
	see \cref{thm:clp-3lvl}.  We can assume that every intermediate graph
	is level planar (when ignoring the constraints) and that its set of
	constraints is acyclic because these properties can be tested in at
	most quadratic time and if one does not hold, we can abort.
	
	We start the construction of $\G_1$ with a copy of \G.
	Given the hook chain $\langle C_1,\dots,C_k\rangle$
	together with a set of hook anchors
	$\{a_1, a_2', a_2, a_3', \dots, a_{k-1}, a_k' \}$
	determined according to \cref{lem:3lclp_hookanchor},
	we add edges to $\G_1$ in order to
	connect all components of the hook chain, as follows.
	For each $i \in [k-1]$, let $w_i$ 
	be the vertex on the spine of the hook piece of~$C_{i+1}$ 
	towards~$C_i$ that is farthest from~$a'_{i+1}$.
	Let $w'_{i+1}$ be defined symmetrically within $C_i$
	and with respect to~$a_i$.
	For example, in \cref{fig:maincomponent},
	on the spine of the hook piece of $C_7$ towards $C_6$, $w_6$ is farthest from~$a'_7$, and,
	on the spine of the hook piece of $C_6$ towards $C_7$, $w'_7$ is farthest from~$a_6$.
	Note that $w_i$ and $w'_{i+1}$ can lie on any level of the band
	where its hook pieces lies.
	We now add the edges $a_iw_i$ and $a'_{i+1}w'_{i+1}$.
	We subdivide every edge that has one endpoint on the upper and one endpoint on
	the lower level with a new vertex on the middle level in order to make $\G_1$ proper.
	See \cref{fig:maincomponent} for an example with $i = 6$.
	We call the resulting component of $\G_1$ the \emph{main component}.
	
	\begin{figure}[tb]
		\centering
		\includegraphics{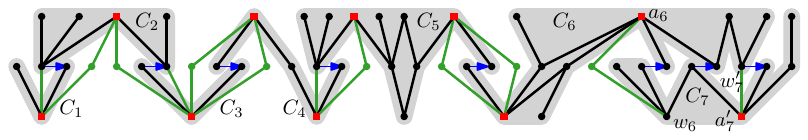}
		\caption{Constructing the main component: original components
			in gray, hooking constraints in blue, hook anchors
			(according to \cref{lem:3lclp_hookanchor}) in red,
			and edges in $E(\G_1) \setminus E(\G)$ in green.
		}
		\label{fig:maincomponent}
	\end{figure}
	
	\begin{lemma}
		\label{lem:3lclp_g1}
		The constrained level graph $\G_1$ is constrained level planar if
		and only if \G is.
		Furthermore, given \G, a hook chain, and a set of hook anchors
		that appear in some constrained level planar drawing of \G,
		we can construct $\G_1$ in linear time.
	\end{lemma}
	
	\begin{proof}
		If~$\G$ is not constrained level planar, then certainly~$\G_1$ is
		not constrained level planer either (because it has the same set of
		constraints, but additional edges).
		
		On the other hand, if~$\G$ is
		constrained level planar, then $\G$ admits a constrained level
		planar drawing~$\Gamma$ with the hook chain $\langle C_1, \dots, C_k \rangle$.
		If the hook chain is degenerate ($k=1$)
		it holds that $\G=\G_1$ and we are done, so we assume $k\ge 2$.
		We now have 
		the hook anchors $\{a_1, a'_2, a_2, a'_3, \dots, a_{k-1}, a'_k \}$.
		
		Consider, for any $i \in [k-1]$,
		the vertices $a_i$ and $w_i$
		(the argument is symmetric for $a'_i$ and $w'_{i}$ with $i \in \{2,\dots,k\}$).
		The hook piece of $C_{i+1}$ towards $C_{i}$
		is a caterpillar whose spine is oriented away
		(i.e., to the left) of its hook anchor $a'_{i+1}$ in $\Gamma$.
		Since $a_i$ is a hook anchor, the leftmost vertex of $C_{i+1}$
		on the middle level is visible from $a_i$. This leftmost vertex is either
		the last vertex $w_i$ of the spine of the hook piece of $C_{i+1}$ towards $C_{i}$,
		or adjacent to it.
		We need to show that we can include a $y$-monotone curve between $a_i$ and $w_i$ in $\Gamma$.
		This is obvious if $w_i$ is on the middle level. If $w_i$ is not
		on the middle level, then we can draw the curve connecting
		$a_i$ with a point in the left vicinity of the leftmost vertex of $C_{i+1}$
		on the middle level, and follow its edge towards $w_i$.
		Thus, $w_i$ can be reached by $a_{i}$ in $\Gamma$
		and, by adding the (possibly subdivided) edge $a_{i}w_i$ to $\Gamma$,
		$\Gamma$ can be extended to a constrained level planar drawing of~$\G_1$.
		
		We can find each last vertex of the spine in linear time
		and the size of the spines together is at most linear in the size of \G.
		Inserting the new edges can again be done in linear time.
		Hence, we can construct $\G_1$ in linear time.
	\end{proof}
	
	\begin{figure}[tb]
		\centering
		\includegraphics[page=2]{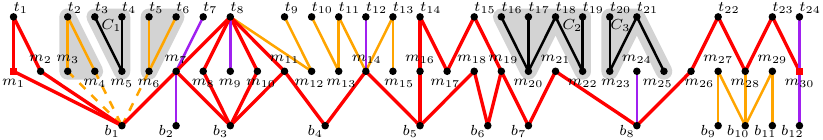}
		\caption{Drawing of the main component of a level graph
			with special vertices $s=m_1$ and $t=m_{30}$: backbone
			edges are thick and red, enclosed components gray, pieces orange
			(edges connecting fingers to their anchors dashed), and
			edges incident to leaves purple.}
		\label{fig:3levelclp}
	\end{figure}

	Next we fix the orientation of the main component.
	To this end, we use the ``guessed'' vertices $s$ and $t$
	such that
	(i)~the vertices along every simple $s$--$t$ path are ordered from left to right and
	(ii)~the vertices that do not lie on any simple
	$s$--$t$ path induce components with a ``simple'' structure.
	Accordingly, let the \emph{backbone spanned by $s, t$}
	(or, for short, the \emph{backbone} if $s$ and $t$ are clear from the context)
	of the main component be the subgraph
	of $\G_1$ that corresponds to the union of all simple $s$--$t$ paths.
	In \cref{fig:3levelclp}, the edges of the backbone spanned by $s = m_1$ and $t = m_{30}$ are drawn thick and red.
	
	Let a \emph{separator vertex} be
	a middle-level vertex on the backbone
	that has a backbone neighbor on both the top and the bottom level.
	Further let an \emph{anchor} be a vertex on the backbone that is incident
	to vertices that are not on the backbone.
	Consider the connected components of $\G_1$ after removing the backbone.
	Each such component is incident to a unique anchor in $\G_1$,
	otherwise there would be a simple $s$--$t$ path through two of its
	anchors.  In this case the path would be part of the backbone.
	Let a \emph{piece} be such a component including its anchor and
	all edges connecting the component with the anchor.
	If there is a piece that occupies all three levels ignoring its anchor,
	we reject the current guess of~$s$ and~$t$.
	
	\begin{lemma}\label{obs:3lclp_pieces}
		If $\G_1$ admits a constrained level planar drawing,
		then there is a choice of~$s$ and~$t$ such that each piece
		ignoring its anchor lies on at most two levels, and that $s$ and $t$ are
		the leftmost and rightmost vertex of $\G_1$ on the middle level, respectively.
		(Hence, ignoring its anchor, every piece is a caterpillar.)
	\end{lemma}
	\begin{proof}
		Let $\Gamma$ be a constrained level planar drawing of $\G_1$,
		let $s$ be the leftmost and $t$ be the rightmost middle-level vertex.
		Since the hook chain is degenerate, $s$ and $t$ both lie in the
		main component
		and hence there is a simple path~$\pi$ from $s$ to $t$.
		The path $\pi$, which is part of the backbone spanned by $s$ and $t$,
		already partitions the main component
		into connected components of height at most~2.
	\end{proof}
	
	Following \cref{obs:3lclp_pieces},
	we define $\G_2$ to be the constrained level graph based on $\G_1$
	with the additional constraints $s \prec_2 t$ as well as
	$s \prec_2 v$ and $v \prec_2 t$ for
	every vertex $v$ on the middle level.
	(Recall that further implicit
	constraints are added automatically.)

	\begin{lemma}\label{lem:3lclp_g2b}
		If $\G_2$ admits a constrained level planar drawing,
		then, for every simple $s$--$t$ path~$\pi$
		and for every pair $(u,v)$ of vertices of~$\pi$
		such that~$u$ precedes~$v$ in~$\pi$ and $\gamma(u) = \gamma(v) = i$,
		$\G_2$ contains the constraint $u \prec_i v$.
	\end{lemma}
	\begin{proof}
		Fix any simple $s$--$t$ path $\pi$.
		
		First, consider the case that $\pi$ lies on only one band,
		say, the lower band.
		Then, $\pi$ consists of an odd number $l$ of vertices.
		Pick $k$ such that $l = 2 k - 1$.
		We name the vertices along~$\pi$
		$s = m_1, b_1, m_2, b_2, m_3, \dots, b_{k-1}, m_k = t$,
		where $m_1, m_2, \dots, m_k$ are the middle-level vertices
		of~$\pi$ in order and $b_1, b_2, \dots, b_{k-1}$ are the
		bottom-level vertices of~$\pi$ in order.
		By definition, $\G_2$ contains the constraints
		$m_1 \prec_2 m_2$, $m_1 \prec_2 m_3$, $\dots$, $m_1 \prec_2 m_k$.
		These constraints imply (due to the implicit planarity constraints)
		together with the edge $m_1 b_1$ on the one hand and
		the edges $m_2 b_2, m_3 b_3, \dots, m_{k-1}b_{k-1}$ on the
		other hand that $\G_2$ contains also the constraints
		$b_1 \prec_1 b_2$, $b_1 \prec_1 b_3$, $\dots$, $b_1 \prec_1 b_{k-1}$.
		In turn, these constraints imply
		together with the edge $b_1 m_2$ on the one hand and
		the edges $b_2 m_3, b_3 m_4, \dots, b_{k-1} m_k$ on the other hand 
		that $\G_2$ contains also the constraints
		$m_2 \prec_2 m_3$, $m_2 \prec_2 m_4$, $\dots,$ $m_2 \prec_2 m_k$.
		Repeating this argument shows that,
		for every pair $(u,v)$ of vertices on~$\pi$ such
		that~$u$ precedes~$v$ in~$\pi$ and $\gamma(u) = \gamma(v) = i$,
		$\G_2$ contains the constraint $u \prec_i v$.
		
		Second, consider the case that $\pi$ lies on all three levels.
		Let $r$ be the first (middle-level) vertex of $\pi$ whose predecessor
		and successor on $\pi$ lie on distinct levels.
		Since we have added the constraint $s \prec_2 r$,
		we repeat the first argument of this proof to show
		the statement for the $s$--$r$ sub-path of~$\pi$
		(where $r$ takes the role of~$t$).
		If $r$ is the only (middle-level) vertex of $\pi$ whose predecessor
		and successor on $\pi$ lie on distinct levels,
		then we apply the first argument of this proof to show
		the statement for the $r$--$t$ sub-path of~$\pi$.
		Otherwise, we recursively apply the second argument of this proof
		to show the statement for the $r$--$t$ sub-path of~$\pi$
		(where $r$ takes the role of~$s$).
	\end{proof}

	Note that, in $\G_2$, the order~$\prec_2$ restricted to the backbone vertices is not necessarily total;
	for example,
	in \cref{fig:3levelclp},
	the order of $m_8$ and $m_{10}$
	as well as the order of $m_{17}$ and~$m_{18}$
	is not fixed.
	For each such set of permutable middle-level vertices
	on the backbone individually, pick an
	arbitrary topological order (with respect to the constraints) and fix it by adding new constraints.
	We call the resulting constrained level graph~$\G_3$.
	
	\begin{lemma}
		\label{lem:fix-middle-level-backbone}
		The constrained level graph $\G_3$ is constrained level planar if and only if $\G_2$ is.
		Furthermore, $\G_3$ can be constructed from $\G_2$ in quadratic time.
	\end{lemma}
	
	\begin{proof}
		If there is a constrained level planar drawing of $\G_3$,
		we also have a constrained level planar drawing of $\G_2$
		because $\G_3$ has the same underlying level graph as $\G_2$
		and a superset of the constraints.
		
		Suppose for a contradiction that $\G_2$
		admits a constrained level planar drawing $\Gamma$, but~$\G_3$ does not.
		Consider each set of vertices
		whose order has been fixed by
		adding constraints in $\G_3$.
		Since $\G_3$ is not constrained level planar,
		the involved vertices cannot be
		permuted in $\Gamma$ to match the order
		prescribed in $\G_3$.
		This can only happen if there is a constraint in~$\G_2$
		contradicting the order prescribed in~$\G_3$.
		Then, however, this does not correspond to a topological order.
		
		To detect each such set of vertices,
		construct the auxiliary graph that has
		a node for every middle-level vertex on the backbone
		and an edge for each pair of vertices
		that are not ordered by a constraint.
		This graph can be constructed in quadratic time
		and each connected component
		corresponds to such a set of vertices.
		Finding topological orders and inserting
		the new constraints can be done in overall quadratic time.
	\end{proof}

	From \cref{lem:fix-middle-level-backbone,lem:3lclp_g2b}
	the following observation follows.
	\begin{observation}\label{obs:fix-backbone}
		If $\G_3$ admits a constrained level planar drawing,
		then the orders on all levels
		restricted to the backbone vertices
		are total.
	\end{observation}
	
	Now let a \emph{finger} be a piece (including its anchor) that
	lies on all three levels
	and hence has an anchor on the top or bottom level.
	Further let a \emph{hand} be the union of fingers incident to
	the same anchor.
	For example, in \cref{fig:3levelclp},
	the two yellow components incident to~$b_1$
	are fingers that together form a hand.
	Let~$\alpha$ be the anchor of a hand,
	let $N_\alpha$ be the set of neighbors of $\alpha$ within the hand, and
	let $L_\alpha$ ($R_\alpha$) be the set of backbone
	neighbors of~$\alpha$ that lie before (after) $\alpha$ on some simple
	$s$--$t$ path.
	For example, consider the hand with anchor $\alpha = b_1$ in \cref{fig:3levelclp}.
	We have $L_\alpha = \{m_1,m_2\}$,
	$N_\alpha = \{m_3, m_4, m_6\}$, and
	$R_\alpha=\{m_7\}$.
	Note that $L_\alpha$ and $R_\alpha$
	are nonempty because $\alpha$ cannot be $s$ or~$t$.
	
	We construct the constrained level graph $\G_4$ from $\G_3$ by
	iterating over each hand of~$\G_3$ and performing the following
	operations:  (i)~Let $\alpha$ be the anchor of the hand,
	(ii)~add constraints from all vertices in $L_\alpha$
	to all vertices in~$N_\alpha$,
	(iii)~add constraints from all vertices in $N_\alpha$
	to all vertices in~$R_\alpha$, and
	(iv)~remove all edges between $\alpha$ and~$N_\alpha$.
	
	Note that fingers in $\G_3$ are represented by
	enclosed components in $\G_4$.
	Due to the implicit constraints,
	two such enclosed components that originate from different hands are
	ordered according to their anchors.
	Further note that the component--constraint graph $H_{\G_4}$ of~$\G_4$ remains
	strongly connected since every finger in $\G_3$ that becomes
	a component in $\G_4$ has both incoming and outgoing constraints
	with respect to the main component.
	
	\begin{lemma}\label{lem:fingers}
		The constrained level graph $\G_4$ is constrained level planar if and only if $\G_3$ is.
		Furthermore, $\G_4$ can be constructed from $\G_3$ in quadratic time.
	\end{lemma}
	\begin{proof}
		If there is a constrained level planar drawing~$\Gamma_3$ of $\G_3$,
		we can obtain a constrained level planar drawing of $\G_4$ by removing,
		for each hand with anchor~$\alpha$,
		the edges between $\alpha$ and $N_\alpha$.
		Note that the additional constraints
		in $\G_4$ are satisfied because
		otherwise edges of the fingers in~$\G_3$
		would cross backbone edges in~$\Gamma_3$.
		
		If there is a constrained level planar drawing~$\Gamma_4$ of $\G_4$,
		we can obtain a constrained level planar drawing of $\G_3$ by adding,
		for each hand with anchor~$\alpha$,
		the edges between $\alpha$ and~$N_\alpha$.
		These edges do not cross any
		other edges
		because the additional constraints
		imply that the vertices of
		$N_\alpha$ lie between
		the vertices of $L_\alpha$ and $R_\alpha$.
		Since the edges between
		$L_\alpha \cup R_\alpha$ and~$\alpha$
		have no crossings,
		the same holds for the new edges.
		
		Since we know, for each vertex, whether it lies on the backbone,
		we can determine all anchors and all fingers in linear time.
		Trivially, the number of constraints that we add is at most quadratic.
	\end{proof}

	Now each remaining piece lies within one of the two bands (including
	its unique anchor); hence it is a caterpillar.
	In the simplest case, this is just a leaf.
	Next, we will get rid of such leaves.
	Let $\G_5$ be the constrained level graph
	obtained from $\G_4$ by removing all leaves
	attached to backbone vertices except for possibly $s$ and~$t$.
	Recall that if $\G_4$ is not level planar (when ignoring the constraints)
	or its constraints contain a cycle, we abort here before constructing~$\G_5$.
	
	\begin{lemma}\label{lem:leaves}
		The constrained level graph $\G_5$ is constrained level planar if and only if $\G_4$ is.
		Furthermore, $\G_5$ can be constructed from $\G_4$ in linear time.
	\end{lemma}
	
	\begin{proof}
		If~$\G_4$ is constrained level planar, then certainly~$\G_5$ is
		also constrained level planar (because we have only removed
		vertices and edges).
		
		Suppose for a contradiction that $\G_5$
		is constrained level planar, but $\G_4$ is not.
		We remove the leaves of~$\G_4$ iteratively.
		Let $l$ be the leaf whose removal makes the resulting
		constrained level graph~$\G_4'$ constrained level planar.
		Let $V^{\rightarrow}$ be the set of vertices having
		a constraint to~$l$, and
		let $V^{\leftarrow}$ be the set of vertices having
		a constraint from~$l$.
		Furthermore, let $a$ be the unique neighbor of~$l$.
		Observe that in every constrained level planar drawing of~$\G_4'$,
		all vertices from~$V^{\rightarrow}$ precede all
		vertices from~$V^{\leftarrow}$ because of
		the transitivity constraints due to~$l$ in~$\G_4$.
		Fix any constrained level planar drawing~$\Gamma$ of~$\G_4'$.
		In~$\Gamma$, place $l$ on level~$\gamma(l)$ immediately to
		the right of the rightmost vertex of~$V^{\rightarrow}$
		or as the first vertex on level~$\gamma(l)$ if $V^{\rightarrow} = \emptyset$.
		We call the resulting drawing~$\Gamma'$.
		Note that $\Gamma'$ fulfills all constraints and the only crossings
		of~$\Gamma'$ involve the edge~$al$.  Let~$X$ be the set of edges
		that cross the edge~$al$.  Note that $X$ is not empty, otherwise
		$\G_4'+al$ would be constraint level planar.
		
		Suppose that there is an edge $cd \in X$ such that
		$\gamma(a) = \gamma(c) =: i$, $\gamma(l) = \gamma(d) =: j$,
		and there is a constraint between $a$ and~$c$
		or between $l$ and $d$ in~$\G_4$.
		If~$\G_4$ contains the constraint $l \prec_j d$,
		then, due to the planarity constraints, $\G_4$ contains
		also the constraint $a \prec_i c$, and vice versa.
		As $d \in V^{\leftarrow}$, $c$ needs to be to the left
		of $a$ in~$\Gamma'$ to cross $al$;
		see \cref{fig:no-leaves-a}.
		This, however, contradicts the constraint $a \prec_i c$.
		A symmetric argument holds if
		there are the constraints $d \prec_j l$ and $c \prec_i a$.
		Hence, no endpoint of an edge in~$X$ is involved in
		a constraint with~$a$ or~$l$.
		
		\begin{figure}[tb]
			\centering
			\begin{subfigure}[t]{.23\textwidth}
				\centering
				\includegraphics[page=1]{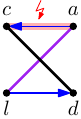}
				\subcaption{If a constraint between $a$ and $c$ or
					between $l$ and $d$ exists, the other exists as well
					and is violated.}
				\label{fig:no-leaves-a}
			\end{subfigure}
			\hfill
			\begin{subfigure}[t]{.17\textwidth}
				\centering
				\includegraphics[page=2]{clp3_leaves_3}
				\subcaption{Edges do not cross from the left (when
					going from~$l$ to~$a$).}
				\label{fig:no-leaves-b}
			\end{subfigure}
			\hfill
			\begin{subfigure}[t]{.195\textwidth}
				\centering
				\includegraphics[page=3]{clp3_leaves_3}
				\subcaption{The edge $al$ crosses two enclosed components
					that have no constraints with~$a$ or~$l$.}
				\label{fig:no-leaves-c}
			\end{subfigure}
			\hfill
			\begin{subfigure}[t]{.28\textwidth}
				\centering
				\includegraphics[page=4]{clp3_leaves_3}
				\subcaption{The vertices of the enclosed components to
					the right of~$a$ and~$l$ can be moved to the free space
					directly to the left of~$a$ and~$l$.}
				\label{fig:no-leaves-d}
			\end{subfigure}
			
			\caption{Configurations appearing in the proof of
				\cref{lem:leaves}, where we obtain $\G_5$ from $\G_4$ by
				removing leaves attached to the backbone.  We consider the
				leaf~$l$ and its unique anchor~$a$ and assume that the edge
				$cd$ crosses~$al$.}
			\label{fig:no-leaves}
		\end{figure}
		
		Next, suppose that there is an edge $cd \in X$ such that
		$a$ precedes~$c$ on level~$i = \gamma(a) = \gamma(c)$ and
		$d$ precedes~$l$ on level~$j = \gamma(l) = \gamma(d)$
		in~$\Gamma'$; see \cref{fig:no-leaves-b}.
		Consider the rightmost vertex~$b$ of~$V^{\rightarrow}$.
		Since all crossings in~$\Gamma'$ involve the edge~$al$,
		vertex~$b$ is not incident to any edge crossing $cd$.
		Moreover, $b$ is not incident to any edge crossing~$al$
		because $b$ is involved in a constraint with~$l$.
		This, however, would force~$b$ to be an isolated vertex,
		which $\G_4$ does not contain.
		Hence, all edges of $X$ cross $al$ in the other direction (that is,
		for every edge~$cd$ in~$X$, $c$ precedes~$a$ and $l$ precedes~$d$). 
		This means that $a$ has no neighbor on level~$j$
		that is to the left of~$l$ and, hence,
		we can place vertices with neighbors on level~$j$ directly
		to the left of~$a$, which we will exploit in the next step.
		Similarly, $l$ has no neighbor of level~$i$ that is to the left
		or right of~$a$, simply because~$l$ is a leaf.
		
		Now, let $cd$ be the leftmost edge of~$X$, again with
		$\gamma(a) = \gamma(c) =: i$ and $\gamma(l) = \gamma(d) =: j$.
		Note that $c$ cannot be a vertex of the backbone
		because then it would have a constraint with~$a$
		due to~\cref{lem:fix-middle-level-backbone}.
		Hence, $c$ (and $cd$) is part on an enclosed component or
		of a piece (potentially being another leaf).  This enclosed
		component or piece lies on the band formed by levels~$i$ and~$j$.
		In either case, no vertex of the enclosed component or piece is
		involved in a constraint with~$a$ or~$l$.
		
		If $c$ lies in an enclosed component~$C$,
		we can move all vertices of~$C$ on level~$i$ that
		are to the right of~$a$ directly to the left of~$a$, and
		we can move all vertices of~$C$ on level~$j$ that
		are to the right of~$l$ directly to the left of~$l$.
		For an illustration of this step, see
		\cref{fig:no-leaves-c,fig:no-leaves-d}.
		
		If $c$ lies in a piece~$p$, we distinguish two subcases depending on
		where $p$ is attached to the backbone with respect to~$al$.  If~$p$
		is attached to the backbone on the left side of~$al$ (or the anchor
		of~$p$ is~$a$), we can move all vertices of~$p$ on level~$i$ that
		are to the right of~$a$ directly to the left of~$a$, and
		we can move all vertices of~$p$ on level~$j$ that
		are to the right of~$l$ directly to the left of~$l$.
		Otherwise, if $p$ is attached to the backbone on the right side
		of~$al$, we can move all vertices of~$p$ on level~$i$ that
		are to the left of~$a$ directly to the right of~$a$
		(in this case, $a$ cannot have a neighbor on level~$j$
		that is to the right of~$l$), and
		we can move all vertices of~$p$ on level~$j$ that
		are to the left of~$l$ directly to the right of~$l$.
		
		If not all endpoints of~$X$ lie in the enclosed component or piece
		that we have treated above, we let $cd$ be the leftmost untreated
		edge in~$X$ and repeat the above treatment.  In the end, after
		treating all edges in~$X$, we obtain a constrained level planar
		drawing of~$\G_4'+al$; a contradiction.
		Summing up, we conclude that, if $\G_4$ is not constrained level
		planar, then neither $\G_5$ is constrained level planar.
		
		Regarding the running time,
		we can check, for each neighbor of a backbone vertex,
		in constant time whether it is a leaf and, if so, remove it.
	\end{proof}

	Now, due to \cref{lem:leaves}, each remaining piece
	(including its unique anchor) is a caterpillar on one band
	that is not a single edge.
	This leads to the following simple observation.
	
	\begin{observation}
		\label{obs:no-three-pieces}
		In $\G_5$, there is no set of three pieces that share an anchor.
	\end{observation}
	
	Note that the backbone subdivides the bands
	into regions above and below the backbone.
	In the next step, we construct a constrained level graph $\G_6$ by
	adding further constraints to $\G_5$ that
	assign pieces to these regions.
	To this end, we make the following definitions.
	
	Recall that a separator vertex is
	a middle-level vertex on the backbone
	that has a backbone neighbor on both the top and the bottom level.
	Let a \emph{gap} be a pair $(a, c)$ of separator vertices such that $a\prec_2 c$ and there
	is no separator vertex $b$ with $a\prec_2 b \prec_2 c$.
	If, for every simple $s$--$t$ path~$\pi$ through $a$ and $c$, the
	subpath from $a$ to $c$ goes through the lower band, then we
	call $(a,c)$ an \emph{upper gap}.
	Symmetrically, if it goes through the upper band, then we
	call $(a,c)$ a \emph{lower gap}.
	Otherwise, that is, if there
	is a path through $a$ and~$c$ whose subpath goes through the lower
	band and there is a path through $a$ and~$c$ whose subpath goes
	through the upper band, then we call $(a,c)$ a \emph{closed gap}
	(for instance see the gap between $m_{16}$ and~$m_{19}$ in
	\cref{fig:3levelclp}).
	Closed gaps can enclose only leaves and isolated vertices.
	Hence, they are empty in $\G_5$.
	Consider the sequence of upper and lower gaps.
	(If this sequence is empty, simply omit the following
	construction and continue with $\G_6:=\G_5$.)
	In this sequence, we call a maximal sequence of consecutive upper (lower) gaps an \emph{upper (lower) gap group}.
	We sort these gap groups from left to right according to the order
	of their separator vertices
	and name them $W_1, \dots, W_l$ for some $l \ge 1$
	(see \cref{fig:2satG} for an example).
	Furthermore, we have two \emph{unbounded gap groups}
	$W_0$ and $W_{l+1}$ representing
	the space to the left of the first separator vertex and to the right
	of the last separator vertex, respectively.
	
	Now we construct a 2-SAT formula
	whose solutions represent assignments of
	the pieces to gap groups.
	Note that some pieces 
	can only be placed into one gap group;
	hence we ignore them in this step.
	These are all pieces whose anchor lies on the middle level
	(e.g., the pieces with anchors $m_{14}$ and $m_{28}$
	in \cref{fig:3levelclp})
	and whose anchor is incident to a closed gap
	(e.g., the piece with anchor $t_8$ in \cref{fig:3levelclp}).
	Two pieces with the same anchor
	need to be assigned to different gap groups.
	(For example, in \cref{fig:2satG},
	$p_5$ and $p_6$ must lie in different gap groups.)
	Three or more pieces with the same anchor do
	not exist due to \cref{obs:no-three-pieces}.
	
	The assignment of pieces to gap groups depends also on
	the placement of the enclosed components of $\G_5$
	(including previous fingers); see \cref{fig:2satG}.
	To model these dependencies, we define a directed auxiliary graph~$H_\mathrm{gap}$ as follows.
	The nodes of $H_\mathrm{gap}$ are the pieces
	(that we do not ignore),
	the enclosed components (which both span only one band),
	and the gap groups.
	We now define the arc set of $H_\mathrm{gap}$ in order to
	express that certain pieces and components must
	be placed to the left (or right) of certain gap groups or other pieces and components.
	For an enclosed component~$C$,
	$H_\mathrm{gap}$ has an arc from
	the gap group node of~$W_i$ to $C$
	if $W_{i+1}$ is the first gap group
	into which $C$ can be placed due to its band
	and its constraints from backbone vertices.
	Symmetrically, $H_\mathrm{gap}$ has an arc from
	$C$ to the gap group vertex of $W_j$
	if $W_{j-1}$ is the last gap group
	into which $C$ can be placed due to its band
	and its constraints to backbone vertices.
	For an enclosed component or a piece $x$
	and an enclosed component or a piece $y$
	on distinct bands,
	$H_\mathrm{gap}$ has an arc from $x$ to $y$
	if there is a constraint from a vertex of $x$
	to a vertex of $y$
	(see \cref{fig:2satHgap} for an example).
	
	Next we describe our 2-SAT formula $\Phi$.
	For each piece $p$, $\Phi$ has a variable $x_p$
	whose truth value (\texttt{true}/\texttt{false}) represents
	the position of $p$ (left/right) with respect to its anchor.
	For two pieces $p$ and $q$ with the same anchor,
	we add the condition $x_p \ne x_q$,
	which can be expressed by the clauses $(x_p \lor x_q)$ and $(\lnot x_p \lor \lnot x_q)$ (see $p_5$ and $p_6$ in \cref{fig:2satG}).
	We now use $H_\mathrm{gap}$ to construct further clauses of~$\Phi$.
	Let $g_\mathrm{left}(p)$ and $g_\mathrm{right}(p)$ be the indices of the left and right gap group, respectively, where a piece $p$ can be placed.
	Note that these are the indices of the gap groups to the left and to the right of the anchor of~$p$, respectively.
	
	\begin{figure}[tb]
		\centering
		\begin{subfigure}[t]{\textwidth}
			\centering
			\includegraphics[page=1]{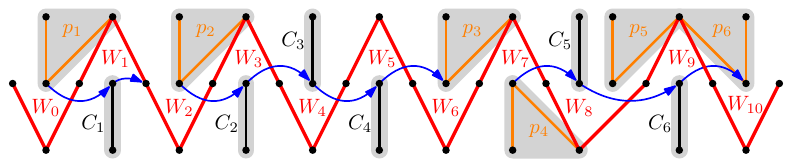}
			\subcaption{Clauses of~$\Phi$: $p_1$ is fixed to the left,
				$p_5$ and $p_6$ must be different,
				$p_2$ cannot be right while $p_3$ is left,
				$p_4$ and $p_6$ force each other to the left and right respectively.
				(Main component in red, pieces in orange, enclosed components in black
				and constraints in blue.)}
			\label{fig:2satG}
		\end{subfigure}
		
		\bigskip
		
		\begin{subfigure}[t]{.48\textwidth}
			\centering
			\includegraphics[page=3]{clp3_2sat}
			\subcaption{The graph $H_\mathrm{gap}$ with two types
				of constraints.  The blue constraints are directly
				induced by constraints in $G$ while the violet
				constraints are induced by planarity.}
			\label{fig:2satHgap}
		\end{subfigure}
		\hfill
		\begin{subfigure}[t]{.48\textwidth}
			\centering
			\includegraphics[page=4]{clp3_2sat}
			\subcaption{A solution of the 2-SAT instance allows
				for a consistent labeling of the nodes in
				$H_\mathrm{gap}$.  The depicted situation of two
				inconsistently labeled paths $p \leadsto q$ and
				$p' \leadsto q'$ cannot occur.}
			\label{fig:2satProof}
		\end{subfigure}
		\caption{A constrained level planar graph $\G$ and the corresponding
			graph $H_\mathrm{gap}$.}
		\label{fig:2sat}
	\end{figure}

	For every pair of a gap group $W_i$ ($i \in \{0, \dots, l+1\}$)
	and a piece $p$, we determine the length $L$ of a longest directed path
	from $W_i$ to $p$ in $H_\mathrm{gap}$ if such a path exists.
	This can be done in linear time since $H_\mathrm{gap}$ is acyclic (if $\G$ is constrained level planar).
	If a path exists, each node on the path needs to be assigned to
	a different gap group (in increasing order) since the arcs of $H_\mathrm{gap}$
	connect objects (i.e., pieces, enclosed components, or gap groups) on distinct bands.
	Hence, if $L > g_\mathrm{right}(p) - i$, then $\G_5$ is not constrained level planar.
	If $L = g_\mathrm{right}(p) - i$, then $p$ needs to lie on the right side of its anchor.
	Thus, we set $x_p$ to \texttt{false}.
	Otherwise, if $L < g_\mathrm{right}(p) - i$, then the path does not restrict the placement of~$p$.
	Symmetrically, if there is a path from $p$ to $W_i$,
	we let $L$ be the length of a longest such path.
	If $L > i - g_\mathrm{left}(p)$, then $\G_5$ is not constrained level planar.
	Else, if $L = i - g_\mathrm{left}(p)$, then $p$ needs to lie on the left side of its anchor.
	Thus, we set $x_p$ to \texttt{true} (see $p_1$ in \cref{fig:2sat}).
	Otherwise, if $L < i - g_\mathrm{left}(p)$, then the path does not restrict the placement of~$p$.
	
	For every pair $(p, q)$ of pieces,
	we determine the length $L$ of a longest directed path from $p$ to~$q$ if it exists.
	If $L > g_\mathrm{right}(q) - g_\mathrm{left}(p)$, then $\G_5$ is not constrained level planar.
	If $L = g_\mathrm{right}(q) - g_\mathrm{left}(p)$,
	then $p$ needs to be placed on the left side of its anchor
	while $q$ needs to be placed on the right side of it anchor
	(see $p_4$ and $p_6$ in \cref{fig:2sat}).
	Hence, we set $x_p$ to \texttt{true} and $x_q$ to \texttt{false}.
	If $L = g_\mathrm{right}(q) - g_\mathrm{left}(p) - 1$, then $p$ cannot be placed on the right side of its anchor
	if $q$ is placed on the left side of its anchor
	(see $p_2$ and $p_3$ in \cref{fig:2sat}).
	Hence, we add the clause $(p \lor \lnot q)$ to $\Phi$.
	Otherwise, the placements of $p$ and $q$ do not depend on each other.
	
	If $\Phi$ admits a feasible truth assignment,
	we assign the pieces to gap groups accordingly.
	We do this by adding constraints between
	the corresponding separator vertices and
	the middle-level vertices of the pieces.
	The resulting constrained level graph is $\G_6$.%
	
	\begin{lemma}\label{lem:3lclp_g6}
		The constrained level graph $\G_6$ is constrained level planar if and only if $\G_5$ is.
		Furthermore, $\G_6$ can be constructed from $\G_5$ in quadratic time.
	\end{lemma}
	\begin{proof}
		If~$\G_5$ is not constrained level planar, then certainly~$\G_6$ is
		not constrained level planar either (because it 
		is the same constrained level graph but with additional constraints).
		
		Otherwise assume that $\G_5$ is constrained level planar but $\G_6$ is not.
		This means that we have assigned the pieces to gap groups
		such that, for this assignment,
		$\G_5$ does not admit a constrained level planar drawing.
		
		Recall that our assignment takes all constraints to/from
		the backbone into account and that
		arbitrarily many enclosed components in one band
		can be assigned to the same gap (group).
		The only remaining case is that chains
		of enclosed components or pieces on alternating bands
		(which are both nodes of $H_\mathrm{gap}$)
		cannot be accommodated.
		It suffices to consider
		objects on alternating bands
		(as they appear in $H_\mathrm{gap}$)
		because two consecutive objects
		on the same band fit into
		the same gap group, and
		the predecessors of
		the first object
		are also predecessors of the
		second object due to
		transitivity constraints
		(symmetrically for successors).
		We claim, however, that we can specify a labeling~$\ell$ of the
		nodes of~$H_\mathrm{gap}$ to indices of gap groups
		such that all pieces and enclosed components
		can be accommodated while observing all constraints.
		
		First label the nodes corresponding to gap groups with their number
		and label the nodes corresponding to pieces
		with the gap group determined by the 2-SAT solution.
		Consider the longest directed paths between
		two labeled nodes $p$ and $q$
		with $\ell(p) \le \ell(q)$ in $H_\mathrm{gap}$.
		Note that $\ell(q) - \ell(p)$ cannot be less
		than the length of a path between $p$ and $q$
		because then we would have violated a constraint in
		the 2-SAT formula
		or there was no solution for~$\G_5$ already.
		We next show a strategy how to iteratively label
		the nodes on such paths.
		We repeatedly determine the smallest $k \in \mathbb{N}_0$
		such that there is a path that has labeled end nodes $p$ and $q$
		(which can also be enclosed components)
		but unlabeled inner nodes and has length $\ell(q) - \ell(p) - k$
		(i.e., a path with minimum ``slack'' for its labels).
		If such a path exists, we label its inner nodes with
		$\ell(p) + 1, \dots, \ell(q) - k - 1$.
		Note that this process terminates with a consistent labeling
		of pieces and enclosed components.
		
		Suppose for a contradiction that, in some step,
		there was an inner node~$z$ on a path between $p$ and $q$
		such that the label $\ell(z)$ became inconsistent
		with another path between~$p'$ and $q'$
		(where possibly $p=p'$ or $q=q'$);
		see \cref{fig:2satProof}.
		Without loss of generality, assume that $\ell(z) - \ell(p')$
		is less than the length of the path between $p'$ and $z$.
		As the path between $z$ and $q$ has length
		$\ell(q) - \ell(z) - k$, the path from $p'$ via $z$ to $q$
		has length at most $\ell(q) - \ell(p') - k - 1$
		and we would have labeled this path first.
		Hence, $\G_6$ admits a constrained-level planar drawing, too.
		
		Concerning planarity note that pieces and enclosed components
		are assigned to gap groups and hence, by implicit constraints,
		they need to lie within a single gap each.
		Since gaps do not contain any edges in their corresponding band,
		crossings with the backbone are impossible.
		If two pieces are forced to have crossing edges
		they may not have the same anchor because this would violate 
		the 2-SAT formula~$\Phi$. Otherwise, and in the cases of one
		piece and one enclosed components or two enclosed component,
		any pair of edges that is forced to cross implies
		a constraint cycle in $\G_5$.

		Regarding the running time,
		note that we can determine and index all gaps and gap groups in linear time.
		Then, for each enclosed component and each piece~$X$,
		we can find its gap group interval
		in linear time by taking into account the constraints
		between $X$ and the backbone.
		Overall, this takes quadratic time.
		To construct $H_\mathrm{gap}$,
		it remains to add arcs between non-gap-group nodes.
		For each pair of such nodes on different bands,
		we check whether there is a constraint between them.
		Again, this can be accomplished in quadratic time.
		The 2-SAT formula has at most a linear number of variables
		(corresponding to the pieces)
		and, thus, a quadratic number of clauses.
		To set up these clauses, we determine, for each piece,
		the longest paths in $H_\mathrm{gap}$ to all other pieces.
		Overall, this takes quadratic time.
		Finding a satisfying truth assignment of the 2-SAT formula
		takes time linear in its length~\cite{eis-ctmfp-SICOMP76},
		hence, quadratic time in our case.
	\end{proof}
	
	Note that the pieces that we have assigned to gap groups in the previous step
	have also been assigned to gaps
	and the orientations of their spines\footnote{%
		Recall that the	\emph{spine} of a caterpillar is
		the subgraph induced by the set of non-leaves.
	} have been fixed~--
	both due to implicit constraints.
	Now consider the other pieces.
	Recall that, due to \cref{obs:no-three-pieces},
	no three pieces share an anchor.
	If two pieces share an anchor and their
	placement to the left or right of this anchor
	is not restricted by constraints,
	then we arbitrarily assign them to distinct sides.
	This assignment is fixed by adding constraints
	from or to their anchor.
	Again, this fixes the orientation of their spines.
	For each remaining piece that is not restricted
	to a single gap due to constraints,
	we arbitrarily assign it to one of its two possible gaps.
	For each remaining piece where the orientation of its spine
	is not yet fixed due to constraints,
	we arbitrarily fix the orientation of its spine~--
	unless two pieces are incident to the same anchor and need to lie
	in the same gap.
	In this case we fix their orientation such that one extends
	to the left and the other extends to the right;
	see the two pieces with anchor $m_{14}$ in \cref{fig:3levelclp}.
	We call the resulting constrained level graph $\G_7$.
	
	\begin{lemma}\label{lem:3lclp_g7}
		The constrained level graph $\G_7$ is constrained level planar if and only if $\G_6$ is.
		Furthermore, $\G_7$ can be constructed from $\G_6$ in quadratic time.
	\end{lemma}
	\begin{proof}
		If~$\G_6$ is not constrained level planar, then certainly~$\G_7$ is
		not constrained level planar either (because it 
		is the same constrained level graph but with additional constraints).
		
		Otherwise assume that $\G_6$ is constrained level planar but $\G_7$ is not.
		First suppose that by assigning a piece~$p$ to a specific gap~$g$,
		the resulting constrained level planar graph
		does not admit a constrained level planar drawing any more.
		Note that then, in any constrained level planar drawing of $\G_6$,
		$p$ is assigned to the the gap on the other side of its anchor.
		If $p$ cannot be drawn inside of~$g$  without crossing edges
		of the main component, then there would have been
		an implicit constrained forcing $p$ to be drawn in the other gap.
		Further, observe that $p$ can be drawn arbitrarily close
		to a backbone edge incident to $p$'s anchor
		because otherwise there would be a chain of constraints
		from $p$ to a backbone vertex to $p$'s anchor (or a neighbor of its anchor)
		or the same chain in the other direction.
		Then, however, there would have been no choice for the gap of~$p$.
		
		If a vertex not in $p$ needs to be drawn between the vertices of~$p$
		due to constraints, then the orientation of $p$ is already fixed in~$\G_6$.
		Otherwise, $p$ can be drawn arbitrary small, and then its orientation
		does not affect the remaining drawing.
		
		It remains to argue about the running time.
		We can find all pieces in linear time and
		check whether they are fixed by considering
		the constraints of the involved vertices.
		This can be done in overall quadratic time.
		Fixing the gaps and orientations can again be done in quadratic time.
	\end{proof}
	
	Now let $H_{\G_7}'$ be the component--constraint graph of $\G_7$ without the main component.
	For both the upper and the lower band, there is a sequence of gaps
	where the enclosed components can be placed.
	Let the \emph{gap interval} of an enclosed component $C$ be the subsequence of
	gaps within its band where $C$ can be placed according to its constraints.
	We go through the enclosed components in some topological order of $H_{G_7}'$.
	We assign each component $C$ (by adding constraints)
	into the leftmost gap of its gap interval
	such that $C$ respects the constraints from components placed before.
	If a component cannot be assigned to a gap in its gap interval due to these
	constraints, $\G_7$ is not constrained level planar.
	We call the resulting constrained level graph $\G_8$.
	
	\begin{lemma}\label{lem:3lclp_g8}
		The constrained level graph $\G_8$ is constrained level planar if and only if $\G_7$ is.
		Furthermore, $\G_8$ can be constructed from $\G_7$ in quadratic time.
	\end{lemma}
	\begin{proof}
		If $\G_7$ is not constrained level planar, then $\G_8$ is certainly
		not constrained level planar since the underlying graphs are
		identical, but the set of constraints of $\G_8$ is a superset of
		that of~$\G_7$.
		
		Therefore, we can assume that $\G_7$ is constrained level planar.
		Let~$\Gamma_7$ be a constrained level planar drawing of~$\G_7$.
		Note that the backbones of~$\G_7$ and~$\G_8$ and hence their sets of
		gaps are identical.  Let~$\sigma$ be the topological order of~$H_{\G_7}'$
		which we used to assign the enclosed components to gaps when
		defining~$\G_8$.  Let $C$ be the first enclosed component in~$\sigma$
		such that~$C$ is assigned to a gap~$g_8$ in~$\G_8$, whereas $C$ lies
		in a different gap~$g_7$ in~$\Gamma_7$.  By construction,
		the gap $g_8$ is to the left of the gap $g_7$.
		We can construct a new constrained level
		graph~$\G_7^{(1)}$ from~$\G_7$ by fixing~$C$ to the gap~$g_8$ with constraints.
		
		We claim that the graph~$\G_7^{(1)}$ is constrained level planar.
		Note that the only difference between~$\G_7^{(1)}$ and~$\G_7$ is
		that $\G_7^{(1)}$ fixes~$C$ to a specific gap~$g_8$.  By the
		choice of~$g_8$, all constraints of type $a \prec_i c$ that
		involve a vertex~$a \not\in C$ and a vertex~$c \in C$ are satisfied
		in~$\G_7^{(1)}$.  On the other hand, assigning~$C$ to~$g_8$
		restricts the placement of components that succeed~$C$ in~$\sigma$
		less than the assignment of~$C$ to~$g_7$ in~$\Gamma_7$.
		Note that there is always space at the right side of~$g_7$.
		
		By repeating the above exchange argument, we show that there is a
		positive integer~$k$ such that $\G_8=\G_7^{(k)}$ and, hence, $\G_8$
		is constrained level planar due to the fact that the graphs
		$\G_7^{(k-1)}, \dots, \G_7^{(1)}$ are all constrained level planar.
		
		We can determine the gap intervals of all enclosed components
		in quadratic total time because we have at most a quadratic
		number of constraints in~$\G_7$.
		Furthermore, we can construct $H_{\G_7}'$ in quadratic time
		(there can be a linear number of enclosed components in $\G_7$
		and a quadratic number of constraints between them).
		We can sort $H_{\G_7}'$ topologically in time linear in its size,
		that is, in at most quadratic time.
	\end{proof}
	
	\begin{figure}[tb]
		\centering
		\begin{subfigure}[t]{\textwidth}
			\centering
			\includegraphics[page=2]{clp3_Hg}
			\subcaption{Constrained level planar drawing of
				$g$ with middle-level vertices $v_1, \dots, v_{13}$
				and its enclosed components.
				Vertices of $A_g$ are marked by (green) dotted ovals.
				Constraints are indicated by blue arcs
				(transitive constraints are omitted).
			}
			\label{fig:gapsort-a}
		\end{subfigure}
		
		\bigskip
		
		\begin{subfigure}[t]{.565\textwidth}
			\centering
			\includegraphics[page=3]{clp3_Hg}
			\subcaption{The directed auxiliary graph~$A_g$ has a
				node for every enclosed component or piece and every
				non-covered middle-level vertex (transitive edges
				are omitted).}
			\label{fig:gapsort-b}
		\end{subfigure}
		\hfill
		\begin{subfigure}[t]{.385\textwidth}
			\centering \includegraphics[page=1]{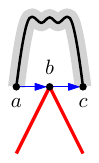}
			\subcaption{Vertex $b$ of the main component (red) is \emph{covered} by an enclosed component or a piece.}
			\label{fig:gapsort-c}
		\end{subfigure}
		
		\caption{Arranging the enclosed components $C_1, \dots, C_5$
			in a gap~$g$ of the main component.  Constraints between
			vertices of $\G_8$ are indicated by blue arrows; arcs
			between nodes of~$A_g$ are indicated by green arrows.  For
			clarity, implicit constraints, constraints along the
			backbone, and transitive arcs are omitted.}
		\label{fig:gapsort}
	\end{figure}
	
	It remains to find a total order of all enclosed components,
	pieces, and middle-level vertices of the backbone
	within each gap~$g$, and to fix the orientation of each enclosed component.
	If there are constraints of the type $a \prec_2 b \prec_2 c$,
	for a vertex $b$ of the backbone and vertices $a, c$ in an enclosed
	component or a piece $C$  ($a, b, c$ on the middle level), then~$b$ lies 
	in the span of $C$ in every constrained level planar drawing of $\G_8$.
	We say that $b$ is \emph{covered} (by $C$);
	see \cref{fig:gapsort-c}.
	
	For every gap~$g$, we define a directed auxiliary graph~$A_g$ that has
	a node for every enclosed component
	and every piece assigned to~$g$ via constraints and
	for every middle-level vertex in~$g$ that is not covered;
	see \cref{fig:gapsort}.
	Note that these middle-level
	vertices are totally ordered
	due to \cref{obs:fix-backbone}.
	In $A_g$, two nodes~$x$ and~$y$ are connected by an arc $(x,y)$
	if there is a constraint between two vertices $v_x$ and $v_y$
	in~$\G_8$ that are represented or covered by $x$ and $y$, respectively.
	We sort $A_g$ topologically and, for each pair $(x, y)$ of consecutive nodes,
	we add a constraint for each pair $(v_x, v_y)$ of vertices on the same level
	with the property that $v_x$ is represented or covered by~$x$
	and~$v_y$ is represented or covered by~$y$.
	This results in a total order of the enclosed components, pieces,
	and uncovered middle-level vertices of the main component.
	
	Now, for each enclosed component,
	we fix the orientation of its spine by testing both possibilities.
	(If both orientations are possible,
	we keep an arbitrary one.)
	Note that this choice can be made independently for 
	each enclosed component.
	We call the resulting constrained level graph $\G_9$.
	
	\begin{lemma}\label{lem:3lclp_g9}
		The constrained level graph $\G_9$ is constrained level planar if and only if $\G_8$ is.
		Furthermore, $\G_9$ can be constructed from $\G_8$ in quadratic time.
	\end{lemma}
	\begin{proof}
		If $\G_8$ is not constrained level planar, then $\G_9$ is certainly
		not constrained level planar since the underlying graphs are
		identical, but the set of constraints of $\G_9$ is a superset of
		that of~$\G_8$.
		
		Otherwise, assume that $\G_8$ is constrained level planar but $\G_9$ is not.
		If we add the new constraints (in an arbitrary order) one by one to $\G_8$,
		then there is a constraint $u \prec_i v$ (for some $i \in [3]$)
		such that adding $u \prec_i v$ makes
		the resulting constrained level graph not constrained level planar.
		Note that $u$ and $v$ cannot both be vertices of the backbone
		because the backbone vertices are already totally ordered in~$\G_8$.
		Ler~$x$ be the node in~$A_G$ to which $u$ belongs and let $y$
		be the node of~$A_g$ to which~$v$ belongs.  We consider two
		cases.
		
		\noindent
		Case~I: $x=y$.
		
		\noindent
		This means that~$x$ is an enclosed component whose orientation
		has been fixed now because there was no restriction
		due to constraints.
		This choice of the orientation does not influence the rest of
		the constrained level graph,
		so in any constrained level planar drawing of $\G_8$,
		we can flip the orientation and still have a constrained level planar drawing.
		
		\noindent
		Case~II: $x \ne y$.  There are three subcases depending on
		what~$x$ and~$y$ represent.
		
		\noindent
		First, we assume that $x$ and $y$ both represent enclosed
		components or pieces.  Clearly, there is no arc between $x$
		and~$y$ in~$A_g$ (because no constrained level planar drawing
		of~$\G_8$ contradicts $A_g$ and $\G_9$ is based on a
		topological sorting of~$A_g$)
		and the vertices of $x$ and $y$ do not interleave
		because they share a common band.
		We claim that we can modify any constrained level planar
		drawing $\Gamma$ of $\G_8$ where $y$ precedes $x$
		such that $x$ precedes $y$ and the drawing is still constrained level planar.
		Note that at most one of $x$ and $y$ has a covered vertex
		because otherwise there would be an arc in $A_g$.
		Say $y$ may cover some vertices and let $W$ be the set of vertices
		that are not part of $x$ or $y$
		but have a constraint to a vertex of $x$
		and lie, within~$\Gamma$, between the first vertex of $y$ and the last vertex of~$x$.
		Clearly, each node of $A_g$ containing a vertex from $W$
		has an arc to $x$, but then it does not have an arc from~$y$
		(as otherwise, there would be an arc from $y$ to $x$ due to transitivity).
		Hence, in $\Gamma$, we can move the vertices in $W$ and~$x$
		to the left of the vertices of $y$ (while maintaining their internal order)
		and still have a constrained level planar drawing.
		So, $x$ and $y$ cannot both be enclosed components or pieces.
		
		Second, we assume that $x$ represents an enclosed component or
		a piece and that the vertex $v$ (represented by~$y$) is
		a middle-level vertex of the backbone.
		Clearly, $x$ cannot cover a vertex then.
		Consider any constrained level planar drawing $\Gamma$ of $\G_8$.
		If $v$ lies in $\Span_\Gamma(x)$, we can move $v$ immediately
		to the right of the rightmost vertex of $\Span_\Gamma(x)$
		because there is no constraint between vertices in $x$ and $v$.
		If $v$ lies to the left of all vertices from $x$,
		we again define $W$ as the set of vertices
		that are not part of $x$ or $y$
		but have a constraint to a vertex of~$x$
		and lie, within~$\Gamma$, between $v$ and the last vertex of~$x$.
		By the same argument as before, we can move all vertices in $W$ and $x$
		to the left of~$v$ and still have a constrained	level planar drawing.
		
		Third, we assume that $y$ represents an enclosed component or
		a piece and that the vertex $u$ (represented by~$x$) is a
		middle-level vertex of the backbone.  This subcase is
		symmetric to the second subcase.
		
		This finishes the case distinction.  In every case we have
		obtained the desired contradiction.
		
		It remains to analyze the running time for
		constructing~$\G_9$.  Detecting all covered vertices can be
		done in at most quadratic time: for each middle-level vertex~$b$,
		first iterate over each incoming constraint $a \prec_2 b$
		and mark the enclosed component or piece that $a$ is part of,
		afterwards iterate over each outgoing constraint $b \prec_2 c$
		and check whether the enclosed component or piece that $c$ is part of
		is marked.
		Constructing the auxiliary graphs can then be done in at most
		quadratic time over all gaps.  Each auxiliary graph can be
		sorted topologically in time proportional to its size.
		In total, this takes at most quadratic time.
		Adding the new constraints and fixing the orientation
		of each enclosed component can both again be done in
		at most quadratic time.
	\end{proof}
	
	Now we consider the remaining leaves and group them
	by their neighbors.
	For each such \emph{group}, we find a topological order and fix it.
	If the group lies on the middle level, then we need to include
	the backbone vertices in the topological order.
	Now let $\G_{10}$ be the constrained level graph
	obtained from $\G_9$ by fixing the orders of every group
	of leaves.
	
	\begin{lemma}\label{lem:leaves2}
		The constrained level graph $\G_{10}$ is constrained level planar if and only if $\G_9$ is.
		Furthermore, $\G_{10}$ can be constructed from $\G_9$ in linear time.
	\end{lemma}
	
	\begin{proof}
		If there is a constrained level planar drawing~$\Gamma_{10}$ of~$\G_{10}$,
		then it is clearly also a constrained level planar drawing of~$\G_9$.
		
		If there is a constrained level planar drawing~$\Gamma_9$ of $\G_9$,
		we can obtain a constrained level planar drawing of $\G_{10}$ by
		reordering the leaves of every group according to the topological
		order they have in $\G_{10}$. If this violates a constraint in
		$\G_{10}$, that constraint must have been in $\G_9$ already.
		
		We claim that groups on the same level are independent.
		Consider the unique neighbors of two such groups.
		Recall that these neighbors are not part of the backbone.
		If both neighbors lie on the same level, the
		spans of their groups cannot overlap due to planarity constraints.
		If one lies on the top level, the other on the bottom level,
		their groups are separated by the backbone. In other words,
		if the spans of the two groups overlapped,
		the backbone would be blocked.
		
		Finding and grouping all leaves can be done in linear time
		by going through the graph once.
		Sorting all groups topologically can be done in linear
		total time.
	\end{proof}
	
	Now, we have added sufficiently many constraints such that, on each
	level, the corresponding vertices of~$\G_{10}$ are totally ordered.
	This yields a polynomial-time algorithm to check whether~\G is
	constrained level planar and, if yes,
	to compute a constrained level planar drawing of~\G.
	Our estimation of the running time is most likely not tight.
	
	\begin{theorem}
		\label{thm:clp-3lvl}
		$3$-level \CLP can be decided in $\mathcal{O}(n^5)$ time,
		where $n$ is the number of vertices in the input graph.
		More precisely, given a 3-level constrained level graph~\G,
		we can test in $\mathcal{O}(n^5)$ time whether~\G
		is constrained level planar and, if yes, we can compute a
		constrained level planar drawing of~\G within the same time bound.
	\end{theorem}
	
	\begin{proof}
		According to the Lemmas~\ref{lem:3lclp_g1} to~\ref{lem:leaves2},
		\G is constrained level planar if and only if~$\G_{10}$ is.
		We next argue that the constraints of~$\G_{10}$
		are total orders for all three levels.
		This follows from the following facts.
		\begin{itemize}%
			\item There are no isolated vertices.
			(\cref{lem:isolated_vertices})
			\item
			There is only one component (the main component)
			that spans three levels;
			the other components are enclosed components on two levels.
			(\cref{lem:3lclp_g1,lem:fingers})
			\item The main component consists of the backbone and pieces, which are caterpillars on two levels.
			(\cref{obs:3lclp_pieces,lem:fingers})
			\item The backbone vertices are totally ordered.
			(\cref{obs:fix-backbone})
			\item There are no leaves adjacent to the backbone and the remaining
			leaves (including potentially $s$ an $t$) are totally ordered.
			(\cref{lem:leaves,lem:leaves2})
			\item The vertices in every enclosed component
			and in every piece are totally ordered.
			(\cref{lem:3lclp_g7,lem:3lclp_g9})
			\item The backbone vertices, pieces, and enclosed components
			are totally ordered.
			(\cref{lem:3lclp_g9})
		\end{itemize}
		
		Therefore, for each $i \in [3]$, the constraints of $\G_{10}$
		for level~$i$ give a total order of the vertices.
		As $\G_{10}$ is proper, we can check in quadratic total time,
		for each pair of edges, whether it induces a crossing.
		The graph~\G is constrained level planar if and only if no pair of
		edges induces a crossing.
		
		For computing a constrained level planar drawing of~\G
		in the affirmative case,
		start with the constrained level planar drawing of~$\G_{10}$
		and reverse the operations done before to obtain
		constrained level planar drawings of~$\G_9, \dots, \G_1, \G$.
		Note that all operations can straight-forwardly be reversed
		in at most the same asymptotic running time as before.
		
		It remains to bound the running time.
		We first guess the first and last vertex on the middle level
		$s$ and $t$ and test all $\mathcal{O}(n^2)$ pairs.
		This guess yields the first and the last component
		of a potential hook chain
		(see \cref{lem:3lclp_hookchain2}).
		For this pair of components, we test four sets of vertices as hook anchors
		(see \cref{lem:3lclp_hookanchor});
		this does not influence the asymptotic running time.
		Given the hook chain and the hook anchors $s$ and $t$,
		we connect the hook chain in linear time (see \cref{lem:3lclp_g1}).
		Then we add, in linear time, constraints that fix $s$ and $t$
		as the first and the last
		vertex on the middle level, which implicitly fixes the
		orientation of the backbone (see \cref{lem:3lclp_g2b}).
		For each guess of $s$ and $t$, we construct, starting with $\G_2$,
		the sequence~$\G_3, \G_4, \dots, \G_{10}$,
		which can be done in $\mathcal{O}(n^2)$ time each
		(see \cref{lem:fix-middle-level-backbone,lem:fingers,lem:leaves,lem:3lclp_g6,lem:3lclp_g7,lem:3lclp_g8,lem:3lclp_g9,lem:leaves2}).
		It remains to consider the running time of checking for and adding
		implicit constraints (see the beginning of \cref{sec:clp-3lvl}). To this end, note that (i) each of
		the $O(n^2)$ constraints can be added at most once and (ii)
		every time a constraint is added, (ii-a) at most $O(n)$
		many other existing constraints must be checked for
		transitivity constraints and (ii-b) at most $O(n^2)$ combinations
		of neighbors must be checked for planarity constraints.
		However, for each pair of edges $ab$ and $cd$ on the same band,
		note that the constraint $b\prec d$ is only checked
		if the constraint $a\prec c$ is added, and vice versa.
		Hence, there is only a total of $O(n^2)$ checks for planarity constraints.
		In summary, checking for and adding implicit constraints
		takes $\mathcal{O}(n^3)$ total time over all steps.
		
		Hence, our algorithm runs in $\mathcal{O}(n^5)$ total time.
	\end{proof}

	\section{Open Problems and Further Remarks}
	\label{sec:conc}
	
	First, we propose the following two open problems.
	
	\begin{enumerate}
		\item We have shown that 3-level \CLP can be solved in
		$\mathcal{O}(n^5)$ time (\cref{thm:clp-3lvl}).  Our approach has
		two bottlenecks.  First, we do not know how to add all implicit
		constraints (in particular the implicit transitivity constraints) in 
		$o(n^3)$ time after each step of the algorithm.
		Note that after each step, we could run the algorithm for computing the transitive closure
		in the same running time as matrix multiplication~\cite{FischerM71,Munro71},
		for which the currently best known algorithm has
		a running time in $\mathcal{O}(n^{2.37216})$~\cite{DBLP:conf/soda/WilliamsXXZ24}.
		However, if the constraints added by this algorithm
		lead to new planarity constraints, which in turn imply further transitivity constraints,
		this could trigger a chain of dependences, which we do not
		know how to handle in overall $o(n^3)$ time.
		Second, our algorithm relies on guessing the first and last vertex
		of the backbone, which means that a quadratic number of guesses need
		to be checked, each in currently $\mathcal{O}(n^3)$ time.
		Can we instead directly construct the backbone or guess just one of
		the two vertices and compute the other one?
		\item Consider the problem \textsc{Connected} \CLP, which is the
		special case of \CLP where the given constrained level graph must be
		connected.  It is easy to make the graph in our NP-hardness proof
		for 4-level \CLP connected by adding two new levels, so 6-level
		\textsc{Connected} \CLP is NP-hard.  What is the complexity of
		5-level \textsc{Connected} \CLP?
	\end{enumerate}
	
	Finally, we remark that a reviewer pointed us
	to~\cite{dfklmnrrww-pclgd-Algorithmica08},
	whose authors
	study problems where the input is an undirected graph {\it without} a level assignment and
	the task is to find a crossing-free drawing with y-monotone edges that, if
	interpreted as a level planar drawing,
	satisfies certain criteria, e.g., being proper and having
	height at most $h$.
	At the end, they claim an FPT-algorithm for problem variants
	``where some vertices have been
	preassigned to layers and some vertices must respect a given partial order'' \cite[Item 5 in Theorem 5]{dfklmnrrww-pclgd-Algorithmica08}.
	As suggested by the reviewer, this variant appears to be \CLP, which would contradict our hardness results.
	However, the statement of \cite[Theorem 5]{dfklmnrrww-pclgd-Algorithmica08} is rather unclear and comes without proof or further explanation.
	In conversation with the authors, we confirmed that they never considered horizontal constraints (``$u$ is to the right of $v$'') and, so, their claim is not related to our results.
	We presume that they refer to vertical constraints (``$u$ is above~$v$'') instead.

	\pdfbookmark[1]{References}{References}
	\bibliographystyle{plainurl}
	\bibliography{olp}

\begin{thebibliography}{10}

\bibitem{DBLP:journals/tcs/AngeliniLBFR15}
Patrizio Angelini, Giordano {Da Lozzo}, Giuseppe {Di Battista}, Fabrizio Frati,
  and Vincenzo Roselli.
\newblock The importance of being proper: (in clustered-level planarity and
  {$T$}-level planarity).
\newblock {\em Theor. Comput. Sci.}, 571:1--9, 2015.
\newblock Conference version in Proc.\ GD 2014
  (\href{https://doi.org/10.1007/978-3-662-45803-7_21}{\texttt{doi:10.1007/978-3-662-45803-7\_21}}).
\newblock \href {https://doi.org/10.1016/j.tcs.2014.12.019}
  {\path{doi:10.1016/j.tcs.2014.12.019}}.

\bibitem{DBLP:journals/tcs/AngeliniLBFPR20}
Patrizio Angelini, Giordano~Da Lozzo, Giuseppe~Di Battista, Fabrizio Frati,
  Maurizio Patrignani, and Ignaz Rutter.
\newblock Beyond level planarity: Cyclic, torus, and simultaneous level
  planarity.
\newblock {\em Theor. Comput. Sci.}, 804:161--170, 2020.
\newblock Conference version in Proc.\ GD 2016
  (\href{https://doi.org/10.1007/978-3-319-50106-2_37}{\tt
  doi:10.1007/978-3-319-50106-2\_37}).
\newblock \href {https://doi.org/10.1016/J.TCS.2019.11.024}
  {\path{doi:10.1016/J.TCS.2019.11.024}}.

\bibitem{DBLP:journals/jgaa/BachmaierBF05}
Christian Bachmaier, Franz{-}Josef Brandenburg, and Michael Forster.
\newblock Radial level planarity testing and embedding in linear time.
\newblock {\em J. Graph Algorithms Appl.}, 9(1):53--97, 2005.
\newblock \href {https://doi.org/10.7155/jgaa.00100}
  {\path{doi:10.7155/jgaa.00100}}.

\bibitem{DBLP:conf/esa/BachmaierB08}
Christian Bachmaier and Wolfgang Brunner.
\newblock Linear time planarity testing and embedding of strongly connected
  cyclic level graphs.
\newblock In Dan Halperin and Kurt Mehlhorn, editors, {\em Proc. 16th Ann.
  European Sympos. Algorithms (ESA)}, volume 5193 of {\em LNCS}, pages
  136--147. Springer, 2008.
\newblock \href {https://doi.org/10.1007/978-3-540-87744-8_12}
  {\path{doi:10.1007/978-3-540-87744-8_12}}.

\bibitem{DBLP:conf/compgeom/BlazejKKS0024}
V{\'{a}}clav Blazej, Boris Klemz, Felix Klesen, Marie~Diana Sieper, Alexander
  Wolff, and Johannes Zink.
\newblock Constrained and ordered level planarity parameterized by the number
  of levels.
\newblock In Wolfgang Mulzer and Jeff~M. Phillips, editors, {\em Proc. 40th
  Int. Symp. Comput. Geom. (SoCG)}, volume 293 of {\em LIPIcs}, pages
  20:1--20:16. Schloss Dagstuhl~-- Leibniz-Zentrum f{\"{u}}r Informatik, 2024.
\newblock \href {https://doi.org/10.4230/LIPICS.SOCG.2024.20}
  {\path{doi:10.4230/LIPICS.SOCG.2024.20}}.

\bibitem{Bodlaender2022}
Hans~L. Bodlaender, Carla Groenland, Hugo Jacob, Lars Jaffke, and Paloma~T.
  Lima.
\newblock {XNLP}-completeness for parameterized problems on graphs with a
  linear structure.
\newblock In Holger Dell and Jesper Nederlof, editors, {\em Proc. 17th Int.
  Symp. Param. Exact Comput. (IPEC)}, volume 249 of {\em LIPIcs}, pages
  8:1--8:18. Schloss Dagstuhl -- Leibniz-Zentrum f{\"u}r Informatik, 2022.
\newblock \href {https://doi.org/10.4230/LIPIcs.IPEC.2022.8}
  {\path{doi:10.4230/LIPIcs.IPEC.2022.8}}.

\bibitem{XNLP2021}
Hans~L. Bodlaender, Carla Groenland, Jesper Nederlof, and Céline M.~F.
  Swennenhuis.
\newblock Parameterized problems complete for nondeterministic {FPT} time and
  logarithmic space.
\newblock In {\em Proc. 62nd Ann. IEEE Symp. Foundat. Comput. Sci. (FOCS)},
  pages 193--204, 2022.
\newblock \href {https://doi.org/10.1109/FOCS52979.2021.00027}
  {\path{doi:10.1109/FOCS52979.2021.00027}}.

\bibitem{DBLP:conf/soda/BrucknerR17}
Guido Br{\"{u}}ckner and Ignaz Rutter.
\newblock Partial and constrained level planarity.
\newblock In Philip~N. Klein, editor, {\em Proc. 28th Ann. {ACM-SIAM} Symp.
  Discrete Algorithms (SODA)}, pages 2000--2011. {SIAM}, 2017.
\newblock \href {https://doi.org/10.1137/1.9781611974782.130}
  {\path{doi:10.1137/1.9781611974782.130}}.

\bibitem{DBLP:conf/isaac/BrucknerR20}
Guido Br{\"{u}}ckner and Ignaz Rutter.
\newblock An {SPQR}-tree-like embedding representation for level planarity.
\newblock In Yixin Cao, Siu{-}Wing Cheng, and Minming Li, editors, {\em Proc.
  31st Int. Symp. Algorithms Comput. (ISAAC)}, volume 181 of {\em LIPIcs},
  pages 8:1--8:15. Schloss Dagstuhl~-- Leibniz-Zentrum f{\"{u}}r Informatik,
  2020.
\newblock \href {https://doi.org/10.4230/LIPIcs.ISAAC.2020.8}
  {\path{doi:10.4230/LIPIcs.ISAAC.2020.8}}.

\bibitem{CyganFKLMPPS2015}
Marek Cygan, Fedor~V. Fomin, \L{}ukasz Kowalik, Daniel Lokshtanov, D{\'{a}}niel
  Marx, Marcin Pilipczuk, Micha\l{} Pilipczuk, and Saket Saurabh.
\newblock {\em Parameterized Algorithms}.
\newblock Springer, Cham, 2015.
\newblock \href {https://doi.org/10.1007/978-3-319-21275-3}
  {\path{doi:10.1007/978-3-319-21275-3}}.

\bibitem{DBLP:journals/tsmc/BattistaN88}
Giuseppe {Di Battista} and Enrico Nardelli.
\newblock Hierarchies and planarity theory.
\newblock {\em {IEEE} Trans. Syst. Man Cybern.}, 18(6):1035--1046, 1988.
\newblock \href {https://doi.org/10.1109/21.23105}
  {\path{doi:10.1109/21.23105}}.

\bibitem{dfklmnrrww-pclgd-Algorithmica08}
Vida Dujmovi{\'c}, Michael~R. Fellows, Matthew Kitching, Giuseppe Liotta,
  Catherine McCartin, Naomi Nishimura, Prabhakar Ragde, Frances Rosamond, Sue
  Whitesides, and David~R. Wood.
\newblock On the parameterized complexity of layered graph drawing.
\newblock {\em Algorithmica}, 52:267--292, 2008.
\newblock \href {https://doi.org/10.1007/s00453-007-9151-1}
  {\path{doi:10.1007/s00453-007-9151-1}}.

\bibitem{DBLP:journals/algorithmica/ElberfeldST15}
Michael Elberfeld, Christoph Stockhusen, and Till Tantau.
\newblock On the space and circuit complexity of parameterized problems:
  Classes and completeness.
\newblock {\em Algorithmica}, 71(3):661--701, 2015.
\newblock \href {https://doi.org/10.1007/s00453-014-9944-y}
  {\path{doi:10.1007/s00453-014-9944-y}}.

\bibitem{eis-ctmfp-SICOMP76}
Shimon Even, Alon Itai, and Adi Shamir.
\newblock On the complexity of timetable and multicommodity flow problems.
\newblock {\em SIAM J. Comput.}, 5(4):691--703, 1976.
\newblock \href {https://doi.org/10.1137/0205048} {\path{doi:10.1137/0205048}}.

\bibitem{FellowsHRV09}
Michael~R. Fellows, Danny Hermelin, Frances~A. Rosamond, and St{\'{e}}phane
  Vialette.
\newblock On the parameterized complexity of multiple-interval graph problems.
\newblock {\em Theor. Comput. Sci.}, 410(1):53--61, 2009.
\newblock \href {https://doi.org/10.1016/j.tcs.2008.09.065}
  {\path{doi:10.1016/j.tcs.2008.09.065}}.

\bibitem{FischerM71}
Michael~J. Fischer and Albert~R. Meyer.
\newblock Boolean matrix multiplication and transitive closure.
\newblock In {\em Proc. 12th Ann. Symp. Switching and Automata Theory (SWAT)},
  pages 129--131. {IEEE} Computer Society, 1971.
\newblock \href {https://doi.org/10.1109/swat.1971.4}
  {\path{doi:10.1109/swat.1971.4}}.

\bibitem{DBLP:conf/sofsem/ForsterB04}
Michael Forster and Christian Bachmaier.
\newblock Clustered level planarity.
\newblock In Peter van Emde~Boas, Jaroslav Pokorn{\'{y}}, M{\'{a}}ria
  Bielikov{\'{a}}, and Julius Stuller, editors, {\em Proc. 30th Conf. Curr.
  Trends Theory \& Practice Comput. Sci. (SOFSEM)}, volume 2932 of {\em LNCS},
  pages 218--228. Springer, 2004.
\newblock \href {https://doi.org/10.1007/978-3-540-24618-3_18}
  {\path{doi:10.1007/978-3-540-24618-3_18}}.

\bibitem{fulek2013hanani}
Radoslav Fulek, Michael~J. Pelsmajer, Marcus Schaefer, and Daniel
  {\v{S}}tefankovi{\v{c}}.
\newblock Hanani--{T}utte, monotone drawings, and level-planarity.
\newblock In {\em Thirty Essays on Geometric Graph Theory}, pages 263--287.
  Springer, 2013.
\newblock \href {https://doi.org/10.1007/978-1-4614-0110-0_14}
  {\path{doi:10.1007/978-1-4614-0110-0_14}}.

\bibitem{GJ75}
Michael~R. Garey and David~S. Johnson.
\newblock Complexity results for multiprocessor scheduling under resource
  constraints.
\newblock {\em SIAM J. Comput.}, 4(4):397--411, 1975.
\newblock \href {https://doi.org/10.1137/0204035} {\path{doi:10.1137/0204035}}.

\bibitem{DBLP:journals/siamcomp/GargT01}
Ashim Garg and Roberto Tamassia.
\newblock On the computational complexity of upward and rectilinear planarity
  testing.
\newblock {\em {SIAM} J. Comput.}, 31(2):601--625, 2001.
\newblock Conference version in Proc.\ GD 1994
  (\href{https://doi.org/10.1007/3-540-58950-3_384}{\texttt{doi:10.1007/3-540-58950-3\_384}}).
\newblock \href {https://doi.org/10.1137/S0097539794277123}
  {\path{doi:10.1137/S0097539794277123}}.

\bibitem{HS71}
Frank Harary and Allen Schwenk.
\newblock Trees with {Hamiltonian} square.
\newblock {\em Mathematika}, 18(1):138--140, 1971.
\newblock \href {https://doi.org/10.1112/S0025579300008494}
  {\path{doi:10.1112/S0025579300008494}}.

\bibitem{DBLP:conf/gd/HeathP95}
Lenwood~S. Heath and Sriram~V. Pemmaraju.
\newblock Recognizing leveled-planar dags in linear time.
\newblock In Franz{-}Josef Brandenburg, editor, {\em Proc. Int. Symp. Graph
  Drawing (GD)}, volume 1027 of {\em LNCS}, pages 300--311. Springer, 1995.
\newblock \href {https://doi.org/10.1007/BFb0021813}
  {\path{doi:10.1007/BFb0021813}}.

\bibitem{DBLP:journals/jda/HongN10}
Seok{-}Hee Hong and Hiroshi Nagamochi.
\newblock Convex drawings of hierarchical planar graphs and clustered planar
  graphs.
\newblock {\em J. Discrete Algorithms}, 8(3):282--295, 2010.
\newblock \href {https://doi.org/10.1016/j.jda.2009.05.003}
  {\path{doi:10.1016/j.jda.2009.05.003}}.

\bibitem{DBLP:conf/gd/JungerLM97}
Michael J{\"{u}}nger, Sebastian Leipert, and Petra Mutzel.
\newblock Pitfalls of using {PQ}-trees in automatic graph drawing.
\newblock In Giuseppe {Di Battista}, editor, {\em Proc. 5th Int. Symp. Graph
  Drawing (GD)}, volume 1353 of {\em LNCS}, pages 193--204. Springer, 1997.
\newblock \href {https://doi.org/10.1007/3-540-63938-1_62}
  {\path{doi:10.1007/3-540-63938-1_62}}.

\bibitem{DBLP:conf/gd/JungerLM98}
Michael J{\"{u}}nger, Sebastian Leipert, and Petra Mutzel.
\newblock Level planarity testing in linear time.
\newblock In Sue Whitesides, editor, {\em Proc. 6th Int. Symp. Graph Drawing
  (GD)}, volume 1547 of {\em LNCS}, pages 224--237. Springer, 1998.
\newblock \href {https://doi.org/10.1007/3-540-37623-2_17}
  {\path{doi:10.1007/3-540-37623-2_17}}.

\bibitem{DBLP:conf/esa/Klemz21}
Boris Klemz.
\newblock Convex drawings of hierarchical graphs in linear time, with
  applications to planar graph morphing.
\newblock In Petra Mutzel, Rasmus Pagh, and Grzegorz Herman, editors, {\em
  Proc.\ 29th Ann. Europ. Symp. Algorithms (ESA)}, volume 204 of {\em LIPIcs},
  pages 57:1--57:15. Schloss Dagstuhl~-- Leibniz-Zentrum f{\"{u}}r Informatik,
  2021.
\newblock \href {https://doi.org/10.4230/LIPIcs.ESA.2021.57}
  {\path{doi:10.4230/LIPIcs.ESA.2021.57}}.

\bibitem{DBLP:journals/talg/KlemzR19}
Boris Klemz and G{\"{u}}nter Rote.
\newblock Ordered level planarity and its relationship to geodesic planarity,
  bi-monotonicity, and variations of level planarity.
\newblock {\em {ACM} Trans. Algorithms}, 15(4):53:1--53:25, 2019.
\newblock Conference version in Proc.\ GD 2017
  (\href{https://doi.org/10.1007/978-3-319-73915-1_34}{\texttt{doi:10.1007/978-3-319-73915-1\_34}}).
\newblock \href {https://doi.org/10.1145/3359587} {\path{doi:10.1145/3359587}}.

\bibitem{clp-vc}
Boris Klemz and Marie~Diana Sieper.
\newblock Constrained level planarity is {FPT} with respect to the vertex cover
  number.
\newblock In Karl Bringmann, Martin Grohe, Gabriele Puppis, and Ola Svensson,
  editors, {\em Proc. 51st Int. Colloq. Autom. Lang. Program. (ICALP)}, volume
  297 of {\em LIPIcs}, pages 99:1--99:17. Schloss Dagstuhl~-- Leibniz-Zentrum
  f{\"u}r Informatik, 2024.
\newblock \href {https://doi.org/10.4230/LIPIcs.ICALP.2024.99}
  {\path{doi:10.4230/LIPIcs.ICALP.2024.99}}.

\bibitem{Munro71}
J.~Ian Munro.
\newblock Efficient determination of the transitive closure of a directed
  graph.
\newblock {\em Inf. Process. Lett.}, 1(2):56--58, 1971.
\newblock \href {https://doi.org/10.1016/0020-0190(71)90006-8}
  {\path{doi:10.1016/0020-0190(71)90006-8}}.

\bibitem{Pietrzak03}
Krzysztof Pietrzak.
\newblock On the parameterized complexity of the fixed alphabet shortest common
  supersequence and longest common subsequence problems.
\newblock {\em J. Comput. Syst. Sci.}, 67(4):757--771, 2003.
\newblock \href {https://doi.org/10.1016/S0022-0000(03)00078-3}
  {\path{doi:10.1016/S0022-0000(03)00078-3}}.

\bibitem{ignaz-pc}
Ignaz Rutter.
\newblock Personal communication, 2022.

\bibitem{DBLP:conf/soda/WilliamsXXZ24}
Virginia {Vassilevska Williams}, Yinzhan Xu, Zixuan Xu, and Renfei Zhou.
\newblock New bounds for matrix multiplication: from alpha to omega.
\newblock In {\em Proc. 2024 Ann. {ACM-SIAM} Symp. Discrete Algorithms (SODA)},
  pages 3792--3835. {SIAM}, 2024.
\newblock \href {https://doi.org/10.1137/1.9781611977912.134}
  {\path{doi:10.1137/1.9781611977912.134}}.

\bibitem{DBLP:journals/dam/WotzlawSP12}
Andreas Wotzlaw, Ewald Speckenmeyer, and Stefan Porschen.
\newblock Generalized $k$-ary tanglegrams on level graphs: {A}
  satisfiability-based approach and its evaluation.
\newblock {\em Discrete Appl. Math.}, 160(16--17):2349--2363, 2012.
\newblock \href {https://doi.org/10.1016/j.dam.2012.05.028}
  {\path{doi:10.1016/j.dam.2012.05.028}}.

\end{thebibliography}
	
\end{document}